\documentclass[a4paper]{article}
\usepackage{typearea}
\typearea{16}
\usepackage{hyperref}
\usepackage{makeidx}
\usepackage{graphicx}
\usepackage{amsmath}
\usepackage{mathtools}
\usepackage{amssymb}
\usepackage{amsthm}
\usepackage{mathrsfs}

\makeatletter
\DeclareRobustCommand*\cal{\@fontswitch\relax\mathcal}
\makeatother
\makeatletter
\newcommand*{\rom}[1]{\expandafter\@slowromancap\romannumeral #1@}
\makeatother

\newtheorem{definition}{Definition}[section]
\newtheorem{theorem}[definition]{Theorem}

\newtheorem{remark}[definition]{Remark}

\theoremstyle{definition}
\newtheorem{example}[definition]{Example}
\newcommand{\QED}{\nobreak \ifvmode \relax \else
      \ifdim\lastskip<1.5em \hskip-\lastskip
      \hskip1.5em plus0em minus0.5em \fi \nobreak
      \vrule height0.75em width0.5em depth0.25em\fi}
\newcommand{\bm}[1]{{\mbox{\boldmath $#1$}}}
	\def\slash#1{\not\!#1}

\def\Tr{{\mathrm{Tr}\hspace{-.05em}}}

\def\stars{{\star\star}}

\def\bcdot{{\circ}}
\def\bcdots{{\circ\circ}}
\def\bcdott{{\circ\circ\circ\circ}}

\def\bullets{{\bullet\bullet}}

\def\Tm{{T\hspace{-.2em}M}}
\def\Tsm{{T^*\hspace{-.2em}M}}
\def\TM{{T\hspace{-.2em}\M}}
\def\TsM{{T^*\hspace{-.2em}\M}}

\def\TMM{{T\hspace{-.2em}\MM}}
\def\TsMM{{T^*\hspace{-.2em}\MM}}

\def\Gso{{{G}_{\hspace{-.2em}S\hspace{-.1em}O}}}
\def\Gsu{{{G}_{\hspace{-.2em}S\hspace{-.1em}U}}}

\def\TsE{{T^*\hspace{-.2em}E}}

\def\GammaSp{{\Gamma_{\hspace{-.1em}\Sp}}}

\def\aaa{{\mathfrak a}}
\def\AAA{{\mathfrak A}}
\def\Aa{{\mathscr{A}}}
\def\aa{{\bm{a}}}

\def\bbb{\mathfrak{b}}

\def\ooo{\mathfrak{o}}
\def\OOO{\mathfrak{O}}
\def\Fym{\mathfrak{ym}}

\def\C{\mathbb{C}}
\def\Z{\mathbb{Z}}

\def\uuu{\mathfrak{u}}

\def\XXX{\mathfrak{X}}

\def\TT{{\mathcal T}}
\def\TTT{{\mathfrak T}}

\def\d{{\cal D}}
\def\D{{\slash{\cal D}}}

\def\ds{{\slash{d}}}

\def\vvv{\mathfrak{v}}
\def\VVV{\mathfrak{V}}

\def\FF{{\mathcal F}}

\def\FFF{\mathfrak{F}}
\def\f{\mathscr{F}}

\def\ggg{\mathfrak{g}}

\def\sss{\mathfrak{s}}
\def\S{{\cal S}}
\def\SSS{{\mathfrak S}}

\def\Ss{{S}}

\def\LL{{\cal L}}

\def\LLL{{\mathfrak L}}

\def\lll{\mathfrak{l}}

\def\iii{\mathfrak{i}}

\def\bbb{{\mathfrak b}}

\def\C{\mathbb{C}}

\def\ccc{{\mathfrak c}}

\def\Varepsilon{{\mathcal E}}

\def\EEE{\mathfrak{E}}
\def\eee{{\mathfrak e}}

\def\I{\mathscr{I}}

\def\M{{\cal M}}
\def\MM{\mathscr{M}}

\def\R{\mathbb{R}}
\def\Ri{R}
\def\RRR{\mathfrak{R}}

\def\W{{\cal W}}

\def\vomega{{\mathfrak w}}
\def\www{{\mathfrak w}}
\def\WWW{{\mathfrak W}}

\def\={{\hspace{-.1em}=\hspace{-.1em}}}

\def\LLambda{\bm{\Lambda}}
\def\Wedge{{\Omega}}

\def\SO{{\hspace{-.1em}S\hspace{-.1em}O\hspace{-.1em}}}
\def\SU{{\hspace{-.1em}S\hspace{-.1em}U\hspace{-.1em}}}
\def\Uo{{\hspace{-.1em}U\hspace{-.1em}(1)\hspace{-.1em}}}
\def\Sp{{\hspace{-.1em}s\hspace{-.1em}p\hspace{.1em}}}
\def\SG{{\hspace{-.2em}s\hspace{-.1em}g}}

\def\SUO{{\SU\hspace{-.1em}\otimes\SO}}

\def\Cl{{\mathcal C\hspace{-.1em}l}}

\def\sss{{\mathfrak s}}
\def\uuu{{\mathfrak u}}

\def\wcs{{\www\hspace{-.1em}\cdot\hspace{-.1em}\Ss}}

\def\cP{{\hspace{-.1em}c\hspace{-.05em}\mathrm{P}}}

\def\cov{{\hspace{-.1em}cov}}

\def\YM{{\mathrm{YM}}}
\def\MYM{{\mathrm{MYM}}}
\def\ym{{\textrm{Y}\hspace{-.1em}\textrm{M}}}

\def\Dc{{\hspace{-.1em}\textrm{D}c}}

\def\cG{{c^{~}_{\hspace{-.1em}g\hspace{-.1em}r}}}

\def\cGz{{c^{~}_{\hspace{-.1em}g\hspace{-.1em}r}}}
\def\cGo{{{c}^{~}_{\hspace{-.1em}g\hspace{-.1em}r}}}
\def\hcG{{{c}^{~}_{\hspace{-.1em}g\hspace{-.1em}r}}}
\def\DSp{{\D^{~}_{\hspace{-.1em}\Sp}}}
\def\dSp{{\slash{d}^{~}_{\hspace{-.1em}\Sp}}}
\def\dSg{{\slash{d}^{~}_{\SG}}}
\def\dSgp{{\slash{d}^{+}_{\SG}}}
\def\dSgm{{\slash{d}^{-}_{\SG}}}
\def\dSgpm{{\slash{d}^{\pm}_{\SG}}}
\def\hdSg{{\slash{\hat{d}}^{~}_{\SG}}}

\def\dR{{de\hspace{.1em}Rham }}

\def\cs{{\ccc\hspace{-.1em}\sss}}
\def\SM{{\hspace{-.1em}\setminus\hspace{-.1em}}}

\def\gdg{\bm{g}^{-1}_\SU\hspace{.1em}d\bm{g}_\SU^{~}}

\def\HD{{\widehat{\mathrm{H}}}}
\def\TxT{{2\hspace{-.1em}\times\hspace{-.1em}2}}

\def\sigmag{{{\sigma^{~}_{\hspace{-.1em}g}}}}

\def\sigmaeta{{{\sigma^{~}_{\hspace{-.1em}\eta}}}}
\def\sigmath{{{\sigma^{~}_{\hspace{-.1em}\theta}}}}

\def\KP{{K\hspace{-.2em}P}}

\def\dr{{d\hspace{-.1em}R}}
\def\Mp{{\M_{\hspace{-.1em}p}}}

\begin{document}
\title{Topological indices of general relativity and Yang--Mills theory in four-dimensional space-time}
\author{Yoshimasa Kurihara\footnote{yoshimasa.kurihara@kek.jp}
\\
{\it\small High Energy Accelerator Organization (KEK), 
Tsukuba, Ibaraki 305-0801, Japan}
}
%\date{}
\maketitle
\begin{abstract} % abstract
This report investigates general relativity and the Yang--Mills theory in four-dimensional space-time using a common mathematical framework, the Chern--Weil theory for principal bundles.
The whole theory is described owing to the fibre bundle with the $GL(4)$ symmetry by twisting several principal bundles with the gauge symmetry.

In addition to the principal connection, we introduce the Hodge-dual connection into the Lagrangian to make gauge fields have dynamics independent from the Bianchi identity.
We show that the duplex superstructure appears in the bundle when a $\Z_2$-grading operator exists in the total space of the bundle in general.
The Dirac operator appears in the secondary superspace using the one-dimensional Clifford algebra, and it provides topological indices from the Atiyah--Singer index theorem.

Though the topological index is usually discussed in the elliptic-type manifold, this report treats it in the hyperbolic type space-time manifold using a novel method, namely the $\theta$-metric space.
The $\theta$-metric treats Euclidean and Minkowski spaces simultaneously and defines the topological index in the Minkowski space-time.
\end{abstract}
%\begin{keyword}
%\kwd{Yang--Mills theory}
%\kwd{general relativity}
%\kwd{Atiyah--Singer index theorem}
%\kwd{cobordism}
%\end{keyword}
\maketitle
\tableofcontents
%%% ----------------------------------------------------------------------
% Introduction
%%% ----------------------------------------------------------------------
\section{Introduction}
A gauge invariance was found and named by Weyl in Maxwell equations at first.
An abelian group of $U(1)$ symmetry in the electromagnetic theory was extended to non-abelian groups of the Yang--Mills theory\cite{PhysRev.96.191}, which is now a basis of the standard theory of particle physics.
A mathematical work concerning the Yang--Mills theory is, e.g., a study of the instanton solutions owing to the harmonic analysis of the (anti-)self-dual Yang--Mills gauge action.
The Atiyah--Singer index theorem\cite{Atiyah:1968mp,Atiyah:1968rj,Atiyah:1967ih,Atiyah:1971rm,Atiyah:1970ws} and its extension to manifolds with boundaries by Atiyah, Patodi, and Singer\cite{atiyah1975spectral, atiyah_patodi_singer_1975, atiyah_patodi_singer_1976} have provided powerful tools to treat the Yang--Mills theory for physicists. 
In 1977, a pioneering work by Schwarz on the Yang--Mills theory\cite{SCHWARZ1977172} appeared.
Since then, dynamic interactions between physics and mathematics have been continuing.
After an epoch-making work by Seiberg and Witten\cite{Witten:1994cg} in 1994, a new era of the interplay between theoretical physics and topology, especially the Donaldson theory\cite{donaldson1983}, was opened, and it continues to date.
At the same time, in theoretical physics, the Atiyah--Singer index theorem of the Dirac operator has provided a deep insight into anomalies of quantum field theory\cite{PhysRevLett.42.1195,ALVAREZGAUME1984449,AlvarezGaume:1984dr}.

Topological studies of the Yang--Mills theory have treated the theory in the four-dimensional manifold, mainly with the \emph{Euclidean} metric.
The primary objective of this report is to clarify topological invariants in the four-dimensional Yang--Mills theory in the \emph{Lorentzian} metric space with non-vanishing curvature.
We are interested in topological invariants in solutions of the Einstein equation and the Yang--Mills equation in the compact and oriented space-time manifold.
This study treats general relativity and the Yang--Mills theory as the Chern--Weil theory with various principal bundles and utilizes the topological method like the Atiyah--Singer index theorem and cobordism.
We cannot use the index theory for the Minkowski space as for Euclidean space since the Dirac operator is the hyperbolic type in theory for the Minkowski space.
This study introduces a novel complex-metric space that provides a one-parameter family of homotopically equivalent metrics, including Euclidean metric space at one end and the Lorentzian metric space at the other.
We note that both spaces are cobordant, and characteristic classes are cobordism invariant.

We organise this report as follows:
After an introductory section, section 2 provides mathematical preliminaries and conventions of mathematical symbols utilised throughout this study.
The $\Z_2$-grading superspace commonly appears in principal bundles, and the novel $\theta$-metric space is introduced in this section.
Section 3 defines various principal bundles in which general relativity and the Yang--Mills theory are developed.
The successive section provides the Lagrangian of the entire Yang--Mills theory and equations of motion for all physical fields in theory.
Section 5 discusses topological invariants appearing in the Yang--Mills theory, and section 6 summarises this study.
The appendix discusses the existing conditions of the dual-connection for the given dual-curvature for $U(1)$ and $SU(N)$ gauge groups.
We note that though the existence of the Hodge-dual curvature is trivial, it is not the case for the dual connection.

This report uses the following physical units:
a speed of light is set to unity $c=1$, but a gravitational constant $\kappa=8\pi G_{\hspace{-.2em}N}/c^2$ ($G_{\hspace{-.2em}N}$ is the Newtonian constant of gravitation) and reduced Plank-constant $\hbar=h/2\pi$ are written explicitly.
In these units, physical dimensions of fundamental constants are $[\hbar\hspace{.1em}\kappa]=L^2=T^2$ and $[\hbar/\kappa]=E^2=M^2$, where $L$, $T$, $E$ and $M$ are, respectively, length, time, energy and mass dimensions.

%%% ----------------------------------------------------------------------
% Section 1 : Mathematical preliminary
%%% ----------------------------------------------------------------------
\section{Mathematical preliminaries}
This section summarises the mathematical preliminaries utilised throughout this study.
A smooth four-dimensional (pseudo-)Riemannian manifold is introduced as a model of the Universe, in which the Yang--Mills theory is developed.
We treat both $SO(1,3)$ and $SO(4)$ symmetries simultaneously using the one-parameter family of metric tensors, namely the $\theta$-metric space.

The Hodge-dual operator has a unique role in four-dimensional manifolds because the operator is endomorphism for two-form objects, such as the curvature form in a four-dimensional manifold.
The space of two-form objects is split into two subspaces and makes a $\Z_2$-grading superspace.
The Dirac operator is naturally introduced in the superspace and induces characteristic index in theories.
We observe that characteristic indices commonly appear in each principal $G$-bundle.

The Atiyah--Singer index theorem asserts that characteristic (analytic) indices owing to the Dirac operator are equivalent to topological indices.
The index theorem is maintained for the elliptic operator in Eudlidean space.
On the other hand, the Yang--Mills theory in physics is formulated in the Minkowski space, and the equation of motion is a hyperbolic type.
We define characteristic indices in the Minkowski space owing to cobordism with the $\theta$-metric.
Characteristic indices in the Minkowski space are mathematically well-defined in this formalism and supply a method to treat indices for physical objects.
%
%Inertial bundle
%
\subsection{Inertial bundle}
We introduce Riemannian manifold $(\MM,\bm{g})$ as a model of the Universe, where $\MM$ is a smooth and oriented four-dimensional manifold, and $\bm{g}$ is a metric tensor in $\MM$ with a negative signature such that det$[\bm{g}]<0$.
In an open  neighbourhood $U_p\subset\MM$, orthonormal bases in $T_{\hspace{-.1em}p}\MM$ and $T^*_{\hspace{-.1em}p}\MM$ are respectively introduced as $\partial/\partial x^\mu$ and $dx^\mu$.
We use the abbreviation $\partial_\mu:=\partial/\partial x^\mu$ throughout this report.
Two trivial vector-bundles $\TMM:=\bigcup_p T_{\hspace{-.1em}p}\MM$ and  $\TsMM:=\bigcup_p T^*_{\hspace{-.1em}p}\MM$ are referred to as  a tangent and cotangent bundles in $\MM$, respectively.

In space-time manifold $\MM$, the Levi-Civita connection is defined using the metric tensor as
\begin{align*}
\Gamma^\lambda_{~\mu\nu}:=\frac{1}{2}g^{\lambda\sigma}\left(
\partial_\mu g_{\nu\sigma}+\partial_\nu g_{\mu\sigma}-\partial_\sigma g_{\mu\nu}
\right),
\end{align*}
where $g_{\mu\nu}:=[\bm{g}]_{\mu\nu}$.
An inertial system, in which the coefficients of the Levi-Civita connection vanishes such that $\Gamma^\lambda_{~\mu\nu}=0$, exists at any point in $\MM$. 
An inertial system at point $p\in\MM$ is denoted $\M_p$, namely a local inertial manifold at $p$.
% Definition
\begin{definition}[Inertial bundle]\label{LorentzBndl}
An inertial bundle is a tuple such that:
\[
\left(\M,\pi_{\mathrm I},\MM,SO(1,3)\right).
\]
\end{definition}
\noindent
A four-dimensional \emph{rotation} is isometry $\Lambda_p:\Mp\rightarrow\Mp$ concerning a local $SO(1,3)$ group.
A total space $\M$ is an inertial manifold that is a trivial bundle of quotient spaces; 
\begin{align}
\M:=\bigcup_{p\in\MM}\Mp/\Lambda_p.\label{Lmani}
\end{align}
A projection map is defined as
\begin{align*}
\pi_{\mathrm I}:\M\otimes\MM\rightarrow\MM:\M_p\mapsto{p}.
\end{align*}

%%%
The existence of smooth map $\pi_{\mathrm I}$ and its inverse globally in $\MM$ is referred to as \emph{Einstein's equivalence principle} in physics.
An orthonormal basis on $\TM$ is represented like $\partial_a$.
As for suffixes of vectors in $\M$, Roman letters are used for components of the trivial basis throughout this study; Greek letters are used for them in $\MM$.
The metric tensor in $\M_p$ is denoted as $\bm\eta$ and is represented using the trivial basis as ${\bm{\eta}}=\mathrm{diag}(1,-1,-1,-1)$.
Metric tensor $\bm\eta$ and Levi-Civita tensor (complete anti-symmetric tensor) $\bm\epsilon$, whose component is $[\bm\epsilon]_{0123}=\epsilon_{0123}=+1$, are constant tensors in $\M_p$.

A pull-back\footnote{
A pull-back of a map $\bullet$ is denoted as $\bullet^\sharp$ in this study.
} 
of bundle map $\pi_{\mathrm I}$ induces one-form object $\eee\in\Gamma(\M,\TsM)$ represented using the trivial basis as
\begin{align*}
\pi_{\mathrm I}^\sharp:
\Omega^1(\TsMM)\otimes\MM\rightarrow\Gamma(\M,\TsM):
dx^\mu\mapsto [\eee]^a:=\Varepsilon^a_\mu(p) dx^\mu,
\end{align*}
where $\Omega^q(\TsMM)$ represents a space of differential $q$-forms in $\TsMM$.
The Einstein convention for repeated indices (one in up and one in down) is exploited throughout this study.
$\Varepsilon^a_\mu\in{C^\infty(\MM)}$ is a smooth function defined in $\MM$, namely a \emph{vierbein}.
The vierbein maps a vector in $\MM$ to that in $\M$ as $\bbb=\Varepsilon\bm\aaa$, where $\bbb\in\TsM$ and $\bm\aaa\in\TsMM$.
We use a Fraktur letter to represent $q$-form objects defined in the cotangent bundle in this report.
(A Fraktur letter is also used to show Lie algebra.)
Vierbain inverse $[\Varepsilon^{-1}]_a^\mu=\Varepsilon_a^\mu$ is also called the vierbein.
One-form object
$\eee=\Varepsilon d\bm{x}$ ($\eee^a=\Varepsilon^a_\mu dx^\mu$ by the trivial basis) is referred to as the vierbein form.
The vierbein form provides an orthonormal basis in $\TsM$; it is a dual basis of $\partial_a$ such as $\eee^a\partial_b=\Varepsilon^a_\mu\Varepsilon^\nu_b dx^\mu\partial_\nu=\delta^a_b$ and is a rank-one tensor (vector) in $\TM$.
The standard bases $\partial_a$ and $\eee^a$ in trivial frame bundles $\TM$ and $\TsM$ are referred to as the trivial bases in this study.
%Any $p$-form object $\aaa\in{V}^q(\TsM)$ can be equated to those in ${V}^q(\TsMM)$ under a bundle map $\pi_{\mathrm I}$, and their space is denoted as $\Wedge^q= V^q(\TsM)\cong V^q(\TsMM)$.
The four-dimensional volume form is represented using virbein forms as 
\begin{align}
\vvv&:=\frac{1}{4!}\epsilon_{\bcdots\bcdots}\hspace{.1em}\eee^\bcdot\wedge\eee^\bcdot\wedge\eee^\bcdot\wedge\eee^\bcdot
=\mathrm{det}[\Varepsilon]\hspace{.2em}dx^0\wedge dx^1\wedge dx^2\wedge dx^3.
\label{VolForm}
\end{align}
Dummy \emph{Roman} indices are often abbreviated to a small circle $\bcdot$ (or $\star$) when the dummy-index pair of the Einstein convention is obvious as above.
When multiple circles appear in an expression, the pairing must be on a left-to-right order at upper and lower indices.
Metric tensors  $\bm{g}$ and $\bm{\eta}$ are related each other like
\begin{align*}
[\bm{g}]_{\mu\nu}&=\left[\bm{\Varepsilon}^t\bm{\eta}
\hspace{.1em}\bm{\Varepsilon}\right]_{\mu\nu}
=\eta_\bcdots\Varepsilon^\bcdot_\mu\Varepsilon^\bcdot_\nu,
\end{align*}
yielding
\begin{align*}
\textrm{det}[\bm{g}]&=\textrm{det}[\bm{\eta}]\textrm{det}[\bm{\Varepsilon}]^2.
\end{align*}
Thus, we obtain that
\begin{align*}
\vvv&=\sqrt{\frac{\mathrm{det}[\bm{g}]}{\mathrm{det}[\bm{\eta}]}}\hspace{.2em}dx^0{\wedge}dx^1{\wedge}dx^2{\wedge}dx^3
=\sqrt{-\mathrm{det}[\bm{g}]}\hspace{.2em}dx^0{\wedge}dx^1{\wedge}dx^2{\wedge}dx^3.
\end{align*}
Similarly the two-dimensional surface form is defined as
\begin{align}
\SSS_{ab}:=\frac{1}{2}\epsilon_{ab\bcdots}\eee^\bcdot\wedge\eee^\bcdot.\label{SufeForm}
\end{align}

Connection one-form $\www$ concerning the $SO(1,3)$ group, namely the \emph{spin-connection form},  is introduced; a $SO(1,3)$-covariant differential operator $d_\www$ is defined using the spin-connection form as
\begin{align}
d_\www\eee^a:=d\eee^a+\cG\www^a_{~\bcdot}\wedge\eee^\bcdot,\label{dwww}
\end{align}
where $\www^a_{~b}=\omega^{~a\bcdot}_{\mu}\hspace{.1em}\eta_{\bcdot b}\hspace{.1em}dx^\mu$, and $\omega^{~ac}_{\mu}$ is a component of the spin-connection form using the trivial basis.
Raising and lowering indices are done using a metric tensor.
Two-form object 
\begin{align}
\TTT^a:=d_\www\eee^a\in V^1(\TM)\otimes\Omega^2\label{torsionFM}
\end{align}
is referred to as a \emph{torsion form}.

Local $SO(1,3)$ group action $\Gso:\TsM\rightarrow\TsM$ is known as the Lorentz transformation.
The Lorentz transformation of the vierbein form and the spin-connection form are
\begin{align*}
\Gso&:\eee\mapsto\Gso(\eee)=\eee'=\LLambda\eee,\\
\Gso&:\www\mapsto\Gso(\www)=\www'=
\LLambda\www\LLambda^{-1}+{\cG}^{\hspace{-.5em}-1}\LLambda\hspace{.1em}d\hspace{-.1em}\LLambda^{\hspace{-.2em}-1}.
\end{align*}
Thus, the Lorentz transformation transforms the covariant differential (\ref{dwww}) like $\Gso(d_\www\eee)=d_{\www'}\eee'$.
A \emph{Lorentz curvature} is defined owing to the structure equation as
\begin{align}
\RRR^{ab}&:=d\www^{ab}+\cG\www^a_{~\bcdot}\wedge\www^{\bcdot b}
\in V^2(\TM)\otimes\Wedge^2,\label{RRR}
\end{align}
which is a two-form valued rank-$2$ tensor  represented using the trivial basis as
\begin{align*}
\RRR^{ab}&=
\sum_{c<d}R^{ab}_{\hspace{.7em}cd}\hspace{.1em}\eee^{c}\wedge\eee^{d}=
\frac{1}{2}R^{ab}_{\hspace{.7em}\bcdots}\hspace{.1em}\eee^\bcdot\wedge\eee^\bcdot.
\end{align*}
Tensor coefficient $R^{ab}_{\hspace{.7em}cd}$ is referred to as the Riemann-curvature tensor.
Ricci-curvature tensor and scalar curvature are defined, respectively, owing to the Riemann-curvature tensor as
\begin{align}
R^{ab}:=\Ri^{\bcdot a}_{\hspace{.7em}\bcdot \star}\eta^{b\star}~~\textrm{and }~~
R:=\Ri^{\bcdot\star}_{\hspace{.7em}\bcdot\star}.\label{RicciR}
\end{align}
The first and second Bianchi identities are
\begin{align*}
d_\www(d_\www\eee^a)=\cG\hspace{.2em}
\eta_{\bcdots}\RRR^{a\bcdot}\wedge\eee^{\bcdot} ~~\mathrm{and}~~
d_\www\RRR^{ab}=0.
\end{align*}

%
% Section, connection, and curvature
%
\subsection{Section, connection, and curvature}\label{app1}
This section introduces \emph{sections}, \emph{connections}, and \emph{curvatures} of a principal bundle.
In fibre bundle $(E,\pi,M,G)$, section is vector space $V$ in total space $E$ and denoted as $\Gamma(M,V(E))$.
We introduce section $s\in\Gamma(M,\Wedge^0(E))$ in the bundle.
Group operator ${G}$ acts on the section from the left like $G(s)={G}\hspace{.1em}s$.
A section is a function (a field in physics) defined on the base manifold $M$.
A connection is a Lie algebra valued one-form $\AAA_G\in\Wedge^1\otimes\ggg$ that keeping differential operator
\begin{align}
d^{~}_G\hspace{.1em}s:=ds+c^{~}_G\hspace{.2em}\AAA^{~}_G\hspace{.1em}s,\label{connec}
\end{align}
namely the \emph{covariant differential}, covariant under a group action ${G}\in End(M,\Wedge^p(E))$, where $\ggg$ is a Lie algebra of the structure group.
Corresponding curvature two-form $\FFF_G\in\Wedge^2\otimes\ggg$ is defined owing to the structure equation as
\begin{align}
\FFF^{~}_G&:=d\AAA^{~}_G+c^{~}_G\hspace{.2em}\AAA^{~}_G\wedge\AAA^{~}_G.\label{adjointG}
\end{align}
We introduce a coupling constant $c^{~}_G\in\C$ in the covariant differential, though it does appear not in the standard mathematics textbook.
The coupling constant can be absorbed in a definition of a connection and a curvature by scaling map $c^{~}_G\hspace{.1em}\AAA_G\mapsto\AAA_G$ and $c^{~}_G\hspace{.1em}\FFF_G\mapsto\FFF_G$, when a single bundle is considered. Nevertheless, a coupling constant is explicitly written in this study because when more than one bundle is twisted over the same base manifold, the coupling constant of each bundle provides relative strength of the interaction in physics.

Covariant differential for $p$-form object $\aaa\in\Wedge^p$ is defined as
\begin{align}
d^{~}_G\aaa:=d\aaa+c^{~}_G\hspace{.1em}[\AAA^{~}_G,\aaa]_\wedge,\label{connec2}
\end{align}
where $[\AAA^{~}_G,\aaa]_\wedge=\AAA^{~}_G\wedge\aaa-(-1)^{p}\hspace{.1em}\aaa\wedge\AAA^{~}_G$.
%We state this fact as a remark as follows:
% Remark
\begin{remark}\label{remA1}
When connection form $\AAA_G$ is transferred under group operation $G$ as
\begin{align*}
G(\AAA^{~}_G)=c^{-1}_G{G}{\bm{1}}^{~}_Gd{G}^{-1}+{G}\AAA^{~}_G{G}^{-1},
\end{align*}
operator $d_G$ is homomorphism $\Gamma(M,\Wedge^p(E))\rightarrow\Gamma(M,\Wedge^{p+1}(E))$ to preserve a relation $G(d_G\sss)={G}d_G\sss$, where ${\bm{1}}_G$ is a unit operator of $G$.
\end{remark}
\begin{proof}
Simple calculations yield
\begin{align*}
d{G}&=d({G}{G}^{-1}{G})=d{G}+{G}d({G}^{-1}){G}
+d{G}%,\\&
\Rightarrow {G}d({G}^{-1}){G}=-d{G};
\end{align*}
thus,
\begin{align*}
G\left(d^{~}_G\sss\right)&=d\left({G}\sss\right)+
c^{~}_G\hspace{.1em}\left(c^{-1}{G}{\bm{1}}^{~}_Gd{G}^{-1}+{G}\AAA^{~}_G{G}^{-1}\right)
\wedge\left({G}\sss\right).
\end{align*}
At the same time,
\begin{align*}
G\left(d^{~}_G\sss\right)&=d{G}\sss
+c^{~}_G\hspace{.2em}G\left(\AAA^{~}_G\right)\wedge{G}\sss.
\end{align*}
Hence, the remark is maintained.
\end{proof}
%%%

When connection one-form $\AAA^{~}_G$ is given, curvature two-form $\FFF^{~}_G$ is provided owing to the structure equation (\ref{adjointG}).
When Lie-algebra valued two-form $\FFF^{~}_G$ is given, is the solution of the structure equation a connection?
% Remark
\begin{remark}\label{remA3}
Suppose one-form $\hat\AAA^{~}_G$ exists as a solution of the differential equation $\hat\FFF^{~}_G=d\hat\AAA^{~}_G+c^{~}_G\hspace{.1em}\hat\AAA^{~}_G\wedge\hat\AAA^{~}_G$  for given $\hat\FFF^{~}_G\in\Omega^2(\TsE)\otimes\ggg$, where $\hat\FFF^{~}_G$ is transferred as an adjoint representation $G(\hat\FFF^{~}_G)=c^{-1}_G{G}{\bm{1}}^{~}_Gd{G}^{-1}+{G}\hat\FFF^{~}_G{G}^{-1}$, one-form $\hat\AAA^{~}_G\in\Omega^1(\TsE)\otimes\ggg$ is also transferred as an adjoint representation such that:
\begin{align}
G(\hat\AAA^{~}_G)=c^{-1}_G{G}{\bm{1}}^{~}_Gd{G}^{-1}+
{G}\hspace{.1em}\hat\AAA^{~}_G\hspace{.1em}{G}^{-1}.\label{A31}
\end{align}
\end{remark}
\begin{proof}
Suppose the transformation of $\hat\AAA_G$ can be expressed as 
\begin{align*}
G(\hat\AAA_G)={G}\hat\AAA_G{G}^{-1}+c^{-1}_Gf_{{G}}\hspace{.1em}{\bm{1}}_G,
\end{align*}
where $f_{{G}}$ is a function of ${G}$, ${dG}$, ${G}^{-1}$ and/or ${dG}^{-1}$, and it does not include $\hat\AAA_G$.
Two-form $\hat\FFF_G$ is transfered as
\begin{align*}
G(\hat\FFF_G)&=d(G(\hat\AAA_G))+
c^{~}_G\hspace{.1em}G(\hat\AAA_G)\wedge G(\hat\AAA_G),\\
&=d({G}\hat\AAA_G{G}^{-1})+
c^{~}_G\hspace{.1em}{G}(\hat\AAA_G\wedge\hat\AAA_G){G}^{-1}
+c^{-1}_Gdf_{{G}}{\bm{1}}_G
\\&~
+c^{-1}_Gf_{{G}}^2{\bm{1}}_G
+\{f_{{G}}{\bm{1}}_G,{G}\hat\AAA_G{G}^{-1}\},
\end{align*}
where $\{a,b\}:=ab+ba$.
On the other hand, it is obtained that
\begin{align*}
G(\hat\FFF_G)&={G}\hat\FFF_G{G}^{-1}
={G}d\hat\AAA_G{G}^{-1}+
c^{~}_G\hspace{.1em}{G}(\hat\AAA_G\wedge\hat\AAA_G){G}^{-1},
\end{align*}
owing to the assumption.
The first term on the right-hand side can be expressed as
\begin{align*}
{G}d\hat\AAA_G{G}^{-1}
&=d({G}\hat\AAA_G{G}^{-1})+
\left\{{G}{\bm{1}}_Gd{G}^{-1},{G}\hat\AAA_G{G}^{-1}\right\}.
\end{align*}
Therefore, following equations are obtained by comparing above two results as
\begin{align*}
\left\{f_{{G}}{\bm{1}}_G,
{G}\hat\AAA_G{G}^{-1}\right\}=
\left\{{G}{\bm{1}}_Gd{G}^{-1},
{G}\hat\AAA_G{G}^{-1}\right\},&~~~
df_{{G}}+f_{{G}}^2=0.
\end{align*}
The solution of the first equation is $f_{{G}}={G}\hspace{.1em}d{G}^{-1}$, which is consistent with the second equation such as
\begin{align*}
f_{{G}}^2&=
\left({G}\hspace{.1em}d{G}^{-1}\right)\wedge
\left({G}\hspace{.1em}d{G}^{-1}\right)=-
\left(d{G}\hspace{.1em}{G}^{-1}\right)\wedge
\left({G}\hspace{.1em}d{G}^{-1}\right)
=-d\left({G}\hspace{.1em}d{G}^{-1}\right),\\&\Rightarrow
df_{{G}}+f_{{G}}^2=0
\end{align*}
Hence, the transformation of $\hat\AAA_G$ is obtained as (\ref{A31}).
\end{proof}
%%%

%
% $\theta$-metric space
%
\subsection{$\theta$-metric space}\label{q-metric}
We treat both Euclidean ($SO(4)$) and Lorentzian ($SO(1,3)$) space-time simultaneously in this study.
In the quantum field theory, we utilise an analytic continuation for a propagator function concerning the time coordinate.
An imaginary number replaces the time coordinate, and a point in Euclidean space is mapped to that in Minkowski space.
Here, we propose an alternative method: two metric spaces are continuously connected using a one-parameter family of metrics, namely the \emph{$\theta$-metric}.

Four-dimensional spaces with Eudlidean- or Lorentzian metric have accidental isomorphism such that:
\begin{align}
\begin{array}{ccc}
\sss\ooo(4)&=&\sss\uuu(2)\oplus\sss\uuu(2),\\
\sss\ooo(1,3)&=&\sss\lll(2,\C)\oplus\overline{\sss\lll(2,\C)}.
\end{array}\label{ismor}
\end{align}
Positional vector $\bm{v}_\xi$ pointing to position $\bm{\xi}=(x^0,x^1,x^2,x^3)^T$ in Riemannian manifold $(M,\bm{g})$ is represented like $\bm{v}_\xi=\bm\Sigma\cdot\bm\xi$ using bases 
\begin{align}
\bm\Sigma:=(\bm{1}_2,\kappa_g\hspace{.1em}\bm\sigma), \label{Sigma}
\end{align}
where $\bm{1}_2$ is a $(\TxT)$ identity matrix and  $\bm\sigma=(\sigma^1,\sigma^3,\sigma^3)$ are Pauli matrices.
Discrete function $\kappa_g$ is provided as 
\begin{align}
\kappa_g:=\frac{i}{\sqrt{\sigmag}},~~\textrm{where}~~\sigmag:=\mathrm{det}[\bm{g}], \label{fg}
\end{align}
yielding $\kappa_g=1$ for the Minkowski metric and $\kappa_g=i$ for the Eudlidean metric.
A vector norm squared is provided as $\mathrm{det}[\bm\Sigma\cdot\bm\xi] = g_\bcdots\xi^\bcdot \xi^\bcdot$.
In reality, $\bm{v}_\xi$ has representations such that:
\begin{align*}
  \bm{v}_\xi&=\left\{
  \begin{array}{cl}
    \left(
      \begin{array}{rr}
        ~~ x^0+i x^3 & i x^1 + x^2\\
         i x^1- x^2 &x^0 - i x^3\\
      \end{array}
     \right)&\textrm{for Eudlidean space}~~(\kappa_g=i),\\ \\
     \left(
      \begin{array}{rr}
         x^0-x^3 & -x^1 +i x^2\\
         -x^1-i x^2 &x^0+ x^3\\
      \end{array}
    \right)&\textrm{for the Minkowski space}~~(\kappa_g=1).
  \end{array}
\right.
\end{align*}
We note that a map $\bm{\xi}\rightarrow \bm{v}_\xi$  for Eudlidean space provides homeomorphism from $S^3$ to $SU(2)$.

We propose a novel method to treat Euclidean and Minkowski spaces simultaneously as two boundaries of the same higher-dimensional space.
$SU(2)$ basis (\ref{Sigma}) is extended to a one-parameter family of bases by replacing discrete parameter $\kappa_g$ with continuous function $\kappa(\theta)$.
We require the following properties for function $\kappa(\theta)$:
\begin{enumerate}
\item $\kappa(\theta)$ is a smooth function of continuous variable $\theta\in[0,1]$, and $\sigmag=-1$ at $\theta=0$ and $\sigmag=1$ at $\theta=1$ with $|\kappa(\theta)|=1$.
\item Complex-valued metric tensor $\bm\eta_\theta$ is introduced as a function of $\theta$ with the trivial basis.
\item $\kappa_g$ is replaced by continuous function $\kappa(\theta)$ according to $\sigmag\mapsto\sigmath:=\mathrm{det}[\bm{\eta}_\theta]$.
\end{enumerate}
We propose the following complex metric tensor and $\kappa(\theta)$ that fulfil above-mentioned requirements:
\begin{subequations}
\begin{align}
\left[\bm\eta^{~}_\theta\right]_{ab}&:=\mathrm{diag}\left(e^{i\pi\theta},-1,-1,-1\right),\label{Cmetric}\\
\kappa(\theta)&:=i\left(e^{-i\pi(\theta+1)}\right)^{1/2}.\label{ftheta}
\end{align}
\end{subequations}
We note that an inverse of the metric tensor is
\begin{align}
\left[\bm\eta^{-1}_\theta\right]^{ab}&:=\mathrm{diag}
\left(e^{-i\pi\theta},-1,-1,-1\right),\label{Cmetrici}
\end{align}
yielding
\begin{align*}
\bm\eta^{-1}_\theta\cdot\bm\eta^{~}_\theta=
\left[\bm\eta^{-1}_\theta\right]^{a\bcdot}\left[\bm\eta^{~}_\theta\right]_{\bcdot b}
=\delta^a_b,~~\textrm{for any}~\theta.
\end{align*}
We write, hereafter, elements of the metric tensor and its inverse as $\left[\bm\eta^{~}_\theta\right]_{ab}=\eta^{~}_{\theta\hspace{.1em}ab}$  and $\left[\bm\eta^{-1}_\theta\right]^{ab}=\eta_\theta^{\hspace{.3em}ab}$, respectively.
The Levi-Civita tensor in the $\theta$-metric space is defined as the completely anti-symmetric tensor with
\[
\epsilon^{~}_{0123}:=1,~~\textrm{and}~~
\epsilon^{0123}_{\theta}:=\frac{1}{4!}
\eta_\theta^{\hspace{.2em}0\bcdot}\eta_\theta^{\hspace{.2em}1\bcdot}
\eta_\theta^{\hspace{.2em}2\bcdot}\eta_\theta^{\hspace{.2em}3\bcdot}
\epsilon_{\bcdots\bcdots}=\mathrm{det}\left[\bm{\eta}^{-1}_\theta\right].
\]
The third requirement is maintained as 
\begin{align}
\kappa_g=\frac{i}{\sqrt{\sigmag}}&\mapsto\frac{i}{\sqrt{\sigmath}}
=i\sqrt{e^{-i\pi(\theta+1)}}=\kappa(\theta).\label{kappag}
\end{align}
We set a branch cut of multivalued function $\bullet^{1/2}$ along the negative real-axis; thus, function $\kappa(\theta)$ behaves around $\theta=0$ and $\theta=1$ like
\[
\kappa(\theta)\big|_{\theta\rightarrow1\pm0}
=i,~~\textrm{and}~~
\kappa(\theta)\big|_{\theta\rightarrow\pm0}
=
\left\{
\begin{array}{rl}
-1&(\theta\rightarrow+0)\\
1&(\theta\rightarrow-0)
\end{array}
.\right.
\]
Figure \ref{figF} shows the behaviour of function $\kappa(\theta)$ with $\theta\in[-2,2]$.
A manifold equipping the $\theta$-metric is denoted  as $(M_\theta,\bm{g}_\theta)$.
A line element in $\Tsm_\theta$ is provided as
\begin{align*}
ds^2:=\eta^{~}_{\theta\hspace{.1em}\bcdots}\hspace{.1em}\eee^\bcdot\otimes\eee^\bcdot
=[\bm{g}_\theta]_{\mu\nu}\hspace{.1em}dx^\mu{\otimes}dx^\nu,~~&\textrm{where}~~
[\bm{g}_\theta]_{\mu\nu}:=\eta^{~}_{\theta\hspace{.1em}\bcdots}\hspace{.1em}
\Varepsilon^\bcdot_\mu\Varepsilon^\bcdot_\nu,
\end{align*}
using a trivial bases.
%We note that the chosen metric and $\kappa(\theta)$ are not unique solutions to fulfil the above conditions. %\ref{themetric}
% Figure
\begin{figure}[t] 
\centering
\includegraphics[width={10cm}]{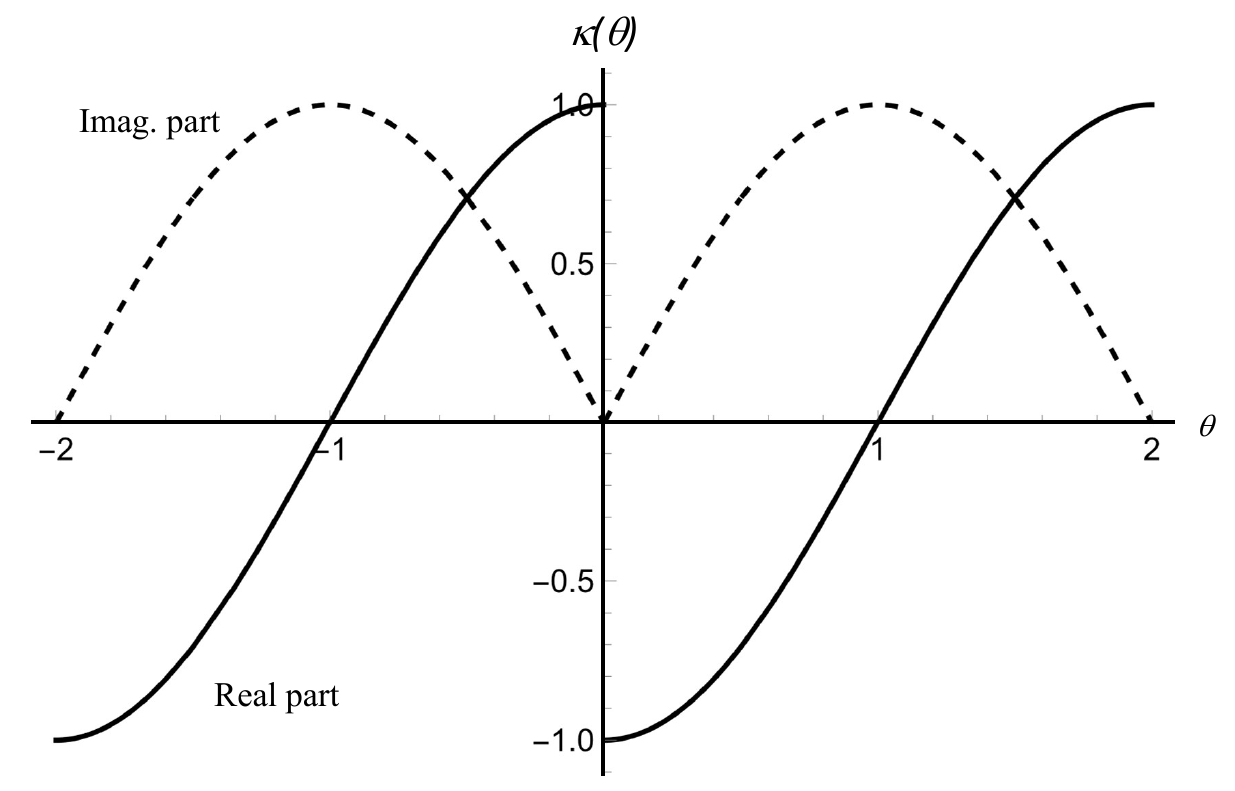}%
 \caption{
$\kappa(\theta)$ is drawn in $-2<\theta<2$.
Solid- and dashed-lines show real- and imaginary-parts of the function, respectively.
A branch cut of $\kappa(\theta)$ is set along the negative real-axis.} 
\label{figF} 
\end{figure}
%%%

A \emph{rotation} operator in the $\theta$-metric space with the trivial bases is 
\begin{align}
\left[\bm\Lambda_\theta(\phi)\right]^a_{~b}=\left(
  \begin{array}{cccc}
\cosh{\left(\kappa(\theta)\hspace{.1em}\phi\right)}&
e^{ i \pi\theta/2}\sinh{\left(\kappa(\theta)\hspace{.1em}\phi\right)}&0&0\\
e^{-i \pi\theta/2}\sinh{\left(\kappa(\theta)\hspace{.1em}\phi\right)}&
\cosh{\left(\kappa(\theta)\hspace{.1em}\phi\right)}&0&0\\
0&0&1&0\\
0&0&0&1
  \end{array}
\right),\label{rotM}
\end{align}
where $\phi\in\R$ is a \emph{rotational} angle.
It is a unitary operator on both boundary manifolds such that:
\begin{align*}
\bm\Lambda_\theta(-\phi)\big|_{\theta\rightarrow+0}=
\bm\Lambda_\theta(\phi)^{-1}\big|_{\theta\rightarrow+0}&=
\bm\Lambda_\theta(\phi)^\dagger\big|_{\theta\rightarrow-0},\\&=
\left(
  \begin{array}{ccc}
\cosh{\left(\phi\right)}&\sinh{\left(\phi\right)}&\bm{0}\\
\sinh{\left(\phi\right)}&\cosh{\left(\phi\right)}&\bm{0}\\
\bm{0}&\bm{0}&\bm{1}
  \end{array}
\right),\\
\bm\Lambda_\theta(-\phi)\big|_{\theta\rightarrow1}=
\bm\Lambda_\theta(\phi)^{-1}\big|_{\theta\rightarrow1}&=
\bm\Lambda_\theta(\phi)^\dagger\big|_{\theta\rightarrow1},\\&=
\left(
  \begin{array}{ccc}
\hspace{.9em}\cos{\left(\phi\right)}&\sin{\left(\phi\right)}&\bm{0}\\
-\sin{\left(\phi\right)}&\cos{\left(\phi\right)}&\bm{0}\\
\bm{0}&\bm{0}&\bm{1}
  \end{array}
\right),
\end{align*}
and provides a rotation in a $(x^0$-$x^1)$-plain when $\theta=1$ and a boost along $x^1$-axis when $\theta=+0$.
Moreover, it is isometric for any fixed $\theta$ in $0\le\theta\le1$ that is confirmed by direct calculations for vector $\bm{v}\in{V^1(TM}_{\hspace{-.2em}\theta})$ such that:
\begin{align*}
|\bm{v}|^2&:=\eta^{~}_{\theta\hspace{.1em}\bcdots} v^\bcdot v^\bcdot=
\eta^{~}_{\theta\hspace{.1em}\bcdots} 
(\bm\Lambda_\theta.\bm{v})^\bcdot 
(\bm\Lambda_\theta.\bm{v})^\bcdot
=|\bm{v'}|^2
\end{align*}
for any $\phi$.
The single rotation operator is represented by operator (\ref{rotM}) without any loss of generality, and a rotation around the $x^0$-axis is trivially preserving a norm of vectors; thus, the rotation invariance of a vector norm is maintained.

%
% duplex super-structure
%

\subsection{Duplex superspace}\label{HodgeZ2Dirac}
When a two-state discrete operator exists on a vector space, the $\Z_2$-grading structure, namely \emph{superstructure}, primarily appears.
Based on this primary space, the secondary $\Z_2$-grading superstructure is constructed using the one-dimensional Clifford algebra in the $(\TxT)$-matrix representation.

The Hodge-dual operator is an example of a two-state discrete operator.
The Yang--Mills theory utilises the Hodge-dual operator for its formulation.
General relativity also has the same structure concerning the Hodge-dual operator.
This section first provides a general construction of this duplex superspace in a vector space on a smooth manifold; then, introduces an example of the duplex superspace owing to the Hodge-dual operator.

%
% General structure
%
\subsubsection{General structure}\label{Gen-Z2-structure}
Suppose $(E,\pi,M,G)$ is a principal bundle, where base space $M$ is a $n$-dimensional Riemannian manifold with a metric tensor $\bm{g}$.
Fibre $V_p=\pi^{-1}(p)$ is a vector space at $p\in M$ and trivial bundle $V:=\bigcup_pV_p$ is defined over $M$.
% 04/Aug./2022
Involution $\hat{h}:V\rightarrow V$, namely the \emph{parity operator} in this study, acts on a vector $v\in{V}$ such that:
\begin{align}
\hat{h}\bcdot\hat{h}(v):=v;\label{hhv2sigmav}
\end{align}
thus, operator $\hat{h}$ has eigenvalue $\pm1$.

A projection operator and its eigenvectors are, respectively, provided as
\begin{align}
P^\pm_{v}:=\frac{1}{2}\left(\bm{1}
\pm\hat{h}\right)~~\textrm{and}~~
v^\pm:=P_v^\pm v\in V^\pm,\label{prohecth}
\end{align}
where $\bm{1}$ is an identity operator.
A dual object concerning the parity operator is denoted as $\hat{\bullet}:=\hat{h}({\bullet})$, e.g., $\hat{v}:=\hat{h}(v)$.
A relation between $(v,\hat{v})$ and $(v^+,v^-)$ is 
\begin{align*}
v=v^++v^-~~\textrm{and}~~\hat{v}=v^+-v^-.
\end{align*}

%
% Primary $\Z_2$-grading structure
%
\paragraph{\textbf{Primary superspace:}}\label{PrimZ2}
Owing to an image of the projection operator, the vector space splits into two subspaces such that:
\begin{align*}
\textrm{Im}(P^+_v:V\rightarrow V)&=\textrm{Ker}(P^-_v:V\rightarrow V)~=:~V^+,\\
\textrm{Im}(P^-_v:V\rightarrow V)&=\textrm{Ker}(P^+_v:V\rightarrow V)~=:~V^-.
\end{align*}
Splitting $V=V^+\oplus V^-$ immediately follows from the definition (\ref{prohecth}).
Operator $\hat{h}$ acts on $v^\pm$ as 
\begin{align*}
\hat{h}(v^\pm)&=\frac{1}{2}\left(
\hat{h}\pm\hat{h}\bcdot\hat{h}
\right)v
=\pm\frac{1}{2}\left(
\bm{1}\pm\hat{h}
\right)v=\pm v^\pm;
\end{align*}
thus, $v^\pm$ are eigenvectors of the parity operator with eigenvalues $\pm1$.
Complementary operator $\bar{P}_v^\pm:=\bm{1}-P_v^\pm$
acts on $v$ such that $\bar{P}^\pm v^\pm=v^\mp$ since
\begin{align*}
P_v^\pm\hspace{.1em}\bar{P}_v^\pm&=P_v^\pm-P_v^\pm\hspace{.1em}P_v^\pm~=~0,\\
P_v^\mp\hspace{.1em}\bar{P}_v^\pm&=P_v^\mp-P_v^\mp\hspace{.1em}P_v^\pm~=~P_v^\mp.
\end{align*}
% remark
%\begin{remark}\label{Endpm}
%Suppose $V$ is a vector space in a smooth manifold.
We note that endomorphisms on vector space $V$ induce the $\Z_2$-grading space such that:
\begin{align}
End^+(V)&=Hom(V^+,V^+)\oplus Hom(V^-,V^-),\nonumber\\
End^-(V)&=Hom(V^+,V^-)\oplus Hom(V^-,V^+).\label{endm}
\end{align}
%\end{remark}
%\noindent
%%%%%%%%%%
E.g., $\hat{h},P^\pm\in End^+(V)$ and $\bar{P}^\pm\in End^-(V)$. 
A product of two endomorphisms yields an algebra, namely \emph{superalgebra}, such that:
\begin{align}
{\cal O}^\pm\bcdot\hspace{.1em}{\cal O}^\pm\cong{\cal O}^+~~\textrm{and}~~
{\cal O}^\pm\bcdot\hspace{.1em}{\cal O}^\mp\cong{\cal O}^-,\label{superalgOO}
\end{align}
where ${\cal O}^\pm\in{End}^\pm(V)$.
A $\Z_2$-grading space equipping superalgebra is referred to as the \emph{superspace} in this study.
The parity operator generally induces the superspace in bundle $E$, namely the \emph{primary superspace} of $E$ concerning parity operator $\hat{h}$.

For given $v\in\Omega^{n/2}$, a pair of vectors 
\begin{align}
\KP(v):=(v,\hat{v}) 
\end{align}
is referred to as the \emph{Kramers pair} in this study.
%Principal connection $\AAA_G^{~}$ is given by (\ref{connec}) and provides principal curvature $\FFF_G^{~}$ through the structure equation (\ref{adjointG}).
The Kramers pair of the principal curvature is introduced owing to $\FFF_G^{~}$ and its dual $\hat{\FFF}_G^{~}$ as $\KP(\FFF^{~}_G)=(\FFF_G^{~},\hat{\FFF}_G^{~})$ in four-dimension.
A dual connection of ${\AAA}_G^{~}$ is defined as a solution of the structure equation,
\begin{align}
\hat\FFF_G^{~}=d\hat\AAA_G^{~}+
\hat{c}^{~}_G\hspace{.2em}\hat\AAA_G^{~}\wedge\hat\AAA_G^{~},\label{steqAG}
\end{align}
for given $\hat\FFF_G^{~}$ and $\hat{c}^{~}_G$.
Thus, it is a connection of the bundle owing to \textbf{Remark \ref{remA3}}.
We assume the existence of $\hat\AAA_G^{~}$ here. 
An existence condition of the dual connection is discussed in \textbf{Appendix \ref{AppDC}}.
We use a short-hand representation of the parity operator on the connection one-form such that:
\begin{align*}
\hat{h}:\Omega^1\rightarrow\Omega^1:\AAA_G^{~}\mapsto
\hat{h}(\AAA_G^{~})=\hat\AAA_G^{~}.
\end{align*}
Generally, connections $\AAA_G^{~}$ and $\hat\AAA_G^{~}$ have different Chern classes to each other; thus, the fibre bundle with $\hat\AAA_G^{~}$ is different from the one with $\AAA_G^{~}$.
The bundle having the connection $\hat\AAA_G^{~}$ is denoted as $\hat{E}$.

We introduce complex scalar function $\phi$ as an example of a section with a unit eigenvalue for the squared parity operator.
A structure group operator acts on $\phi$ as the fundamental representation. 
Total space $E$ ($\hat{E}$) has section field $\phi$ ($\hat\phi$), which is  respectively given as a solution of an equation of motion:
\begin{subequations}
\begin{align}
(id_{\AAA}-\mu)\phi&:=
(id+i{c}^{~}_G\hspace{.1em}\AAA^{~}_G-\mu)\phi=0,\label{eomphi0}\\
(id_{\hat\AAA}-\mu)\hat\phi&:=
(id+i\hat{c}^{~}_G\hspace{.1em}\hat\AAA^{~}_G-\hat\mu)\hat\phi=0,\label{eomphihat0}
\end{align}
\end{subequations}
where $\mu$ and $\hat\mu$ are called  \emph{particle  mass} in physics. 

%%%
%
% Secondary $\Z_2$-grading structure
%
\paragraph{\textbf{Secondary superspace:}}\label{SecondZ2}
We introduce another superspace over the primary superspace utilising the $(\TxT)$-matrix space.
Kramers pairs for given $v\in V$ and $\phi\in\Phi$ are extended to the $(\TxT)$-matrix space as
\begin{align*}
\bm{v}_2:=\left(
\begin{array}{cc}
	v & 0\\
 	0 & \hat{v}
\end{array}
\right)~~\textrm{and}~~
\bm\Phi_2:=\left(
\begin{array}{c}
\phi\\ \hat\phi
\end{array}
\right).
\end{align*}
A $\Z_2$-grading unitary-representation of one-dimensional Clifford algebra $\gamma$ is introduced together with chiral operator $\Gamma$ in the ($\TxT$)-matrix representation such that:
\begin{align}
\gamma:=\left(
\begin{array}{cc}
	0 & \bm{1}\\
 	\bm{1} & 0
\end{array}
\right)~~\textrm{and}~~
\Gamma:=\left(
\begin{array}{cr}
	\bm{1} &   0\\
 	0  & -\bm{1}
\end{array}
\right),\label{2x2Clifford}
\end{align}
yielding $\gamma\Gamma=\Gamma\gamma$ and $\gamma^{2}=\Gamma^{2}=\bm{1}_2$, where $\bm{1}_2:=$diag$(\bm{1},\bm{1})$.
The parity operator in the ($\TxT$)-matrix representation is given as
\begin{align}
\hat{\bm{h}}_2:=
\gamma\Gamma\left(
\begin{array}{cc}
	\hat{h} & 0\\
 	0 & \hat{h}
\end{array}
\right)=
\left(
\begin{array}{cc}
	0& -\hat{h}\\
 	\hat{h} &0
\end{array}
\right)=-\hat{\bm{h}}_2^{-1},\label{2x2hhat}
\end{align}
and acts on vectors as
\begin{align}
\hat{\bm{h}}_2\bm{v}_2\hat{\bm{h}}_2^{-1}=
\bm{v}_2\big|_{v\leftrightarrow\hat{v}}
~~\textrm{and}~~
\hat{\bm{h}}_2\bm{\Phi}_2=
\left(
\begin{array}{r}
	-\hat\phi\\\phi
\end{array}
\right).\label{phihat}
\end{align}
Field $\bm{\Phi}_2$ is called a \emph{spinor} in the secondary superspace.

In the $(\TxT)$-matrix space, we introduce the $\Z_2$-grading superalgebra by assigning an even (odd) parity to diagonal (anti-diagonal) matrices, respectively.
Any $(\TxT)$-matrices can be represented as a sum of diagonal- and anti diagonal-matrices.
Therefore, a space of $(\TxT)$-matrix is split into two subspaces as ${\bm{M}}_2={\bm{M}}_2^+\oplus{\bm{M}}_2^-$, where ${\bm{M}}_2^+$ and ${\bm{M}}_2^-$ are spaces of diagonal and anti-diagonal matrices, respectively. 
A matrix multiplication operator provides the superalgebra for $\bm{m}^\pm,\bm{n}^\pm\in{\bm{M}}_2^\pm$ such that:
\begin{align}
{\bm{m}}^\pm\cdot\bm{n}^\pm\in{\bm{M}}_2^+ &~~\textrm{and }~~
{\bm{m}}^\pm\cdot{\bm{n}}^\mp\in{\bm{M}}_2^-.\label{superalgebra}
\end{align}
The $\Z_2$-grading superspace, namely the \emph{secondary superspace}, is induced in the ($\TxT$)-matrix representation as follows:
Clifford algebra $\gamma$ flips a parity of matrices such as
\begin{align*}
\gamma:{\bm{M}}_2^\pm\rightarrow{\bm{M}}_2^\mp:
\bm{m}^\pm\mapsto\gamma\bm{m}^\pm\in{\bm{M}}_2^\mp.
\end{align*}
In the secondary superspace, parity and projection operators are represented as
\begin{align*}
\hat{\bm{h}}_2:=\hat{h}\bm{1}_2~~\textrm{and}~~
{\bm{P}}^\pm_2:=\frac{1}{2}
\left(\bm{1}_2\pm\hat{\bm{h}}_2\right).
\end{align*}

The Dirac operator is generally induced in the secondary superspace using the one-dimensional Clifford algebra. 
In this study, we define the Dirac operator as follows\cite{berline2003heat}:
% Definition
\begin{definition}[Dirac operator]\label{DiracOp}
Suppose $\bm{V}^\pm$ is a superspace and $\bm{v}^\pm\in\bm{V}^\pm$.
The Dirac operator $\D$ is a first-order differential operator flipping a parity such that:
\begin{align*}
\D:\bm{V}^\pm\rightarrow \bm{V}^\mp:
\bm{v}^\pm\mapsto\D\bm{v}^\pm\in\bm{V}^\mp.
\end{align*} 
$\D^2$ is referred to as a generalized Laplacian. 
\end{definition}
\noindent
%%%
E.g., a covariant differential in the secondary superspace induces the Dirac operator as follows:
The covariant differential is defined in the secondary superspace as
\begin{align*}
D^{~}_\AAA&:=\left(
\begin{array}{cc}
	d^{~}_{\AAA} &  0\\
 	    0  & \hat{d}^{~}_{\AAA}
\end{array}
\right)=
\left(
\begin{array}{cc}
	d+c^{~}_G\hspace{.1em}[\AAA^{~}_G,~]_\wedge &  0\\
 	    0  & d+\hat{c}^{~}_G\hspace{.1em}[\hat\AAA^{~}_G,~]_\wedge
\end{array}
\right).
\end{align*} 
The Dirac operator in the secondary superspace is defined as
\begin{align*}
\slash{D}^{~}_\AAA&:=\gamma D^{~}_\AAA=
\left(
\begin{array}{cc}
	0 & d^{~}_{\hat\AAA} \\
 d^{~}_{\AAA}&  0
\end{array}
\right),
\end{align*}
which fulfils the \textbf{Definition \ref{DiracOp}}.
The Bianchi identities are represented using the Dirac operator such that:
\begin{align*}
\slash{D}^T_\AAA\slash{D}^{~}_\AAA=D^{~}_\AAA D^{~}_\AAA=
\left(
\begin{array}{cc}
	c^{~}_G\hspace{.1em}\FFF_G^{~} &  0\\
 	    0  & \hat{c}^{~}_G\hspace{.1em}\hat\FFF_G^{~}
\end{array}
\right)~~\textrm{and}~~
\slash{D}^{~}_\AAA\left(
\begin{array}{cc}
	\FFF_G^{~} &  0\\
 	    0  & \hat\FFF_G^{~}
\end{array}
\right)=0.
\end{align*}
An equation of motion for the section field is also given as the Dirac operator as $\slash{D}^{~}_\AAA\bm\Phi_2=0$.
The Dirac conjugate of the section field is defined as 
\begin{align*}
\bar{\bm\Phi}^{~}_2:=\bm\Phi_2^\dagger\gamma=(\hat\phi^*,\phi^*),
~~&\textrm{yielding}~~
\bar{\bm\Phi}^{~}_2\slash{D}^{~}_\AAA\hspace{.1em}\bm\Phi^{~}_2=
{\phi}^*d^{~}_{\AAA}\phi+
{\hat\phi}^*\hat{d}^{~}_{\AAA}\hat\phi.
\end{align*}
 
A supertrace\cite{berline2003heat} is defined in the secondary superpsace as
\begin{align}
Str[\bm{m}]:=\mathrm{Tr}[\Gamma\bm{m}]=\left\{
\begin{array}{cl}
\mathrm{Tr}[m^{~}_{11}]-\mathrm{Tr}[m^{~}_{22}] & 
\bm{m}\in\bm{M}_2^+\\
0 & \bm{m}\in\bm{M}_2^-
\end{array}
\right.,\label{SupTr}
\end{align}
where 
\begin{align*}
\bm{m}:=
\left(
\begin{array}{cc}
	m^{~}_{11} &  0\\
 	    0  & m^{~}_{22}
\end{array}
\right)\in\bm{M}_2^+,&~~\textrm{or}~~
\bm{m}:=\left(
\begin{array}{cc}
	0 &  m^{~}_{12}\\
 	m^{~}_{21} & 0  
\end{array}
\right)\in\bm{M}_2^-.
\end{align*}
For the parity-odd curvature defined as 
\begin{align*}
\slash{\FFF}:=\gamma\hspace{.1em}\left(
\begin{array}{cc}
	\FFF_G^{~} &  0\\
 	    0  & \hat\FFF_G^{~}
\end{array}
\right)=\left(
\begin{array}{cc}
 0  & \hat\FFF_G^{~} \\
 	\FFF_G^{~} &  0
\end{array}
\right),
\end{align*}
and the supertrace of the curvature squared yields 
\begin{align*}
Str[\slash{\FFF}\hspace{.1em}\wedge\hspace{-.2em}\slash{\FFF}]&=
Str\left[\left(
\begin{array}{cc}
	\hat\FFF_G^{~}\wedge\FFF_G^{~} &  0\\
 	    0  & \FFF_G^{~}\wedge\hat\FFF_G^{~}
\end{array}
\right)\right],\\&=
\Tr\left[\hat\FFF_G^{~}\wedge\FFF_G^{~}\right]-\Tr\left[\FFF_G^{~}\wedge\hat\FFF_G^{~}\right]
=2\hspace{.1em}\Tr\left[\hat\FFF_G^{~}\wedge\FFF_G^{~}\right].
\end{align*}
We note that
\[
\left[\FFF^{~}_G\wedge\hat\FFF^{~}_G-\hat\FFF^{~}_G\wedge\FFF^{~}_G\right]^I:=
if^I_{~JK}\left[\FFF_G^J,\hat\FFF_G^K\right]=
2\left[\FFF_G\wedge\hat\FFF_G\right]^I,
\]
where $f^I_{~JK}$ is a structure constant\footnote{
Capital Roman letters represent indices of the structure group, and the Einstein convention is also applied to them.
} of structure group $G$.

%
% Principal Hodge-dual operator
%
\subsubsection{Hodge-dual as a parity operator}\label{242}
The Hodge-dual operator maps a $p$-form object to an $(n\hspace{-.1em}-\hspace{-.1em}p)$-form object in an $n$-dimensional smooth Riemannian manifold.
This section discusses the duplex superspace with respect to the Hodge-dual operator as the parity operator in the Lorentz bundle.
This section treats the Lorentz bundle with the $\theta$-metric.

%
% Primary uperspace
%
\paragraph{\textbf{Primary superspace:}}\label{hodge}
% Definition
\begin{definition}[Hodge-dual operator]\label{ctm}
Suppose $M$ is an $n$-dimensional oriented and smooth manifold and $\aaa,\bbb\in\Wedge^p(\Tsm)$ are $p$-form objects.
The Hodge-dual operator, denoted as $\HD$, is defined to give
\begin{align*}
\aaa\wedge\HD(\bbb):=g(\aaa,\bbb)\vvv.
\end{align*}
\end{definition}
\noindent
When $\aaa,\bbb$ are represented using tensor coefficients concerning orthogonal basis $\eee^\bullet$ such that: 
\begin{align*}
\aaa&:=\frac{1}{p!}{a}_{i_1\cdots i_p}\eee^{i_1}\wedge\cdots\wedge\eee^{i_p}
\hspace{.2em}\in\Wedge^p(\Tsm)~~\textrm{for}~~(0\leq{p\in\Z}\leq{n}),\\\
\bbb&:=\frac{1}{p!}{b}_{i_1\cdots i_p}\eee^{i_1}\wedge\cdots\wedge\eee^{i_p}
\hspace{.2em}\in\Wedge^p(\Tsm)~~\textrm{for}~~(0\leq{p\in\Z}\leq{n}),
\end{align*}
a bilinear form concerning the metric tensor is defined as
\begin{align*}
g_\theta(a,b)=\frac{\sqrt{\sigmath}}{p!}
\eta_\theta^{~i_1j_1}\cdots\eta_\theta^{~i_pj_p}
a^{~}_{i_1{\cdots}i_p}b^{~}_{j_{1}{\cdots}j_p}
\end{align*}
in this study.
The Hodge-dual operator fulfilling \textbf{Definition \ref{ctm}} has a component representation of
\begin{align*}
&~\HD:\Omega^p\rightarrow\Omega^{n-p}:\bbb\mapsto\hat{\bbb}
=\frac{\sqrt{\sigmath}}{p!(n-p)!}{b}_{i_1{\cdots}i_{p}}
\left[\bm{\epsilon}_\theta\right]^{i_1{\cdots}i_{p}}_{\hspace{2.4em}i_{p+1}{\cdots}i_{n}}
\eee^{i_{p+1}}\wedge\cdots\wedge\eee^{i_{n}},
\end{align*}
where 
\begin{align*}
\left[\bm{\epsilon}_\theta\right]^{i_1{\cdots}i_{p}}_{\hspace{2.4em}i_{p+1}{\cdots}i_{n}}
:=\frac{1}{p!}\eta_{\theta}^{~i_1j_1}\cdots\eta_{\theta}^{~i_pj_p}
\epsilon_{j_1\cdots j_{p}i_{p+1}{\cdots}i_{n}},
\end{align*}
and $\bm\epsilon_\theta$ is a $n$-dimensional completely anti-symmetric tensor with $[\bm\epsilon_\theta]_{01\cdots n-1}=+1$. 
The operator acts on tensor coefficients of $\bbb$  as
\begin{align*}
\HD:{b}_{i_1{\cdots}i_p}\mapsto\hat{b}_{i_{p+1}{\cdots}i_{n}}
:=&\frac{\sqrt{\sigmath}}{p!}{b}_{i_{1}{\cdots}i_{p}}
\left[\bm{\epsilon}_\theta\right]^{i_{1}{\cdots}i_{p}}_{\hspace{2.4em}i_{p+1}{\cdots}i_{n}},
\intertext{with}
\hat{\bbb}=&\frac{1}{(n-p)!}\hat{b}_{i_1{\cdots}i_{n-p}}
\eee^{i_1}\wedge\cdots\wedge\eee^{i_{n-p}}.
\end{align*}
%%%
E.g., the two-dimensional surface form in the four-dimensional space-time is expressed using the Hodge-dual operator such that 
\begin{align*}
\HD({\eee^a\wedge\eee^b})=\sqrt{\sigmath}\eta_{\theta}^{~a\bcdot}\eta_{\theta}^{~b\bcdot}\SSS_\bcdots.
\end{align*}
Another example of a form object obtained by the Hodge-dual is the volume form.
According to the definitions of the Hodge-dual operator and the volume form defined by (\ref{VolForm}), we obtain that
\begin{align}
\HD(1)=\sqrt{\sigmath}\eee^0\wedge\cdots\wedge\eee^{n-1}
=\sqrt{\sigmath}\hspace{.2em}\vvv=:\vvv_\theta,~~&\textrm{and}~~~
\HD\bcdot\HD(\aaa)=(-1)^{p(n-p)}\aaa\in\Omega^p.\label{HHa}
\end{align}
Therefore, the Hodge-dual operator is the parity operator when $p(n-p)$ is an even number.

An $L^2$ inner-product of $\aaa,\bbb$ is defined using the Hodge-dual operator as
\begin{align}
\langle{\aaa,\bbb}\rangle_\theta&:=\int_{M_\theta}\aaa\wedge\HD(\bbb)=
\int_{M_\theta}g_\theta(a,b)\hspace{.2em}\HD(1)=\langle{\bbb,\aaa}\rangle_\theta
.\label{bilinear}
\end{align}

A pseud-norm of a $p$-form object is defined as a square-root of 
\begin{align}
\|\aaa\|_\theta^2:=\langle\aaa,\aaa\rangle_\theta=
\int_{M_\theta}g_\theta(a,a)\hspace{.2em}\HD(1).\label{L2Norm}
\end{align} 
This pseud-norm squared is positive definite only at $\theta=0$ (the Eudlidean metric).
%The trace operator is provided by (\ref{bilinear}) as
%\begin{align*}
%\Tr_{\hspace{.1em}\Tsm}\hspace{.1em}[\aaa\wedge\hat\bbb]:=(a,b).
%\end{align*}

When a dimension of the manifold is an even number and $\aaa$ is an $(n/2)$-form, operator $\HD$ is endomorphism. (For the four-dimensional space-time, see, e.g., Refs.\cite{Atiyah425},\cite{Woit:2021bmb}.)
% remark
\begin{remark}\label{vpvm}
Operator $\HD$ has two-eigenvalues $\pm1$.
A space of two-form objects splits into two subspaces, namely a positive- and negative-parity spaces, concerning eigenvalues of operator $\HD$ like 
\begin{align}
V(\Omega^{n/2})=V^+(\Omega^{n/2})\oplus V^-(\Omega^{n/2}).\label{superspace} 
\end{align}
\end{remark}
\begin{proof}
Suppose $\aaa$ is an eigenvector of operator $\HD$. 
Due to (\ref{HHa}) with $p=n/2$, $\HD$ has two-eigenvalues of $\pm1$.

An orthonormal basis of $\Omega^1(\Tsm)$ is represented as $\{\eee^0,\cdots,\eee^{n-1}\}$.
An orthonormal basis of $(n/2)$-form objects is constructed using this basis such that
\begin{align*}
\tilde\SSS^I:=\eee^{i_{c(1}}\wedge\cdots\wedge\eee^{i_{n/2)}},~~\textrm{with}~~I=1,\cdots,m:=\frac{n!}{((n/2)!)^2},
\end{align*}
where $c(1\cdots n/2)$ runs over all possible combination of $(n/2)$-integers from $0$ to $n-1$.
We note that $m$ is even integer for any even number $n$.
Owing to the definition of the Hodge-dual operator, we can always take $(m/2)$-pairs of bases $(\tilde\SSS^I,\tilde\SSS^J)$ fulfiling
\begin{align*}
\HD(\tilde\SSS^I)=(-1)^{n/2}\hspace{.1em}\tilde\SSS^J~~\textrm{and}~~
\HD(\tilde\SSS^J)=\tilde\SSS^I.
\end{align*}
Owing to pairs $(\tilde\SSS^I,\tilde\SSS^J)$, new bases are constructed as
\begin{align}
\tilde\SSS^\pm:=\tilde\SSS^I\pm(-1)^{n/2}\hspace{.1em}\tilde\SSS^J,\label{SISJ}
\end{align}
yielding 
\begin{align*}
\HD\left(\tilde\SSS^+\right)&=
(-1)^{n/2}\left(\tilde\SSS^J+\tilde\SSS^I\right)=
+(-1)^{n/2}\tilde\SSS^+\in V^+(\Omega^{n/2}),\\
\HD\left(\tilde\SSS^-\right)&=
(-1)^{n/2}\left(\tilde\SSS^J-\tilde\SSS^I\right)=
-(-1)^{n/2}\tilde\SSS^-\in V^-(\Omega^{n/2}).
\end{align*}
The $(n/2)$-form objects $\tilde\SSS^\pm$ are linearly independent from one another; any two-form objects can be represented using a linear combination of these forms.
Therefore, the remark is maintained. 
\end{proof}
\noindent
%%%
A space of endomorphisms in $V(\Omega^{n/2})$ splits into $\Z_2$-grading subspaces like  (\ref{endm}), and superalgebra is induced as (\ref{superalgOO}).
E.g., the Hodge-dual operator on $(n/2)$-form objects in $n$-dimensional space belongs to $End^+(\Omega^{n/2})$.

So far in this section, we have discussed the superspace in $n$-dimensional manifold.
Hereafter, we restrict a space-time manifold to four dimensional; thus, $(-1)^{n/2}=1$.

Projection operator $P^\pm_\textrm{H}$ is defined as
\begin{align}
P^\pm_\textrm{H}:=\frac{1}{2}\left(1\pm\HD\right),\label{projHodge}
\end{align}
which projects $2$-form objects onto one of two eigenspaces such that $P^\pm_\textrm{H}\aaa=\aaa^\pm\in V^\pm(\Omega^{2})$.
The $2$-form object belonging to $V^+$ ($V^-$) is referred to as the self-dual (SD) (the anti-self-dual (ASD)) forms, respectively.
Two-form objects $\aaa^\pm\in V^\pm$ are eigenvectors of the Hodge-dual operator owing to (\ref{HHa}) and (\ref{projHodge}) such that:
\begin{align}
\HD(\aaa^\pm)=\pm\aaa^\pm.\label{Hapm}
\end{align}

In the four-dimensional $\theta$-metric space, the Hodge-dual operator induces the superspace in the space of curvature two-forms: SD and ASD curvatures are
\begin{align*}
\FFF_G^\pm:=P^\pm_\textrm{H}\FFF_G^{~}~~\textrm{and}~~
\HD(\FFF_G^\pm)=\pm\FFF_G^\pm
\in V^\pm\left(\Omega^2\right)\otimes\ggg.
\end{align*}
Owing to isomorphism (\ref{ismor}), a space of one-form objects in $\Tsm_{\hspace{-.2em}\theta}$ is also split into two subspaces:
\begin{align*}
V(\Omega^1)&=
V^+(\Omega^1)\oplus V^-(\Omega^1).
\end{align*}
The SD and ASD curvatures are represented using corresponding connections like
\begin{align}
\FFF_{G}^\pm&=d\AAA_{G}^\pm+c^\pm_G\hspace{.2em}
\AAA_{G}^\pm\wedge\AAA_{G}^\pm~\in~V^\pm(\Omega^2)\otimes\ggg.\label{A+A+}
\end{align}
One-form objects in $V^\pm(\Omega^1)$ are spanned by the basis 
\begin{align}
\eee^{03}_\pm:=\eee^0\pm \kappa(\theta)\hspace{.1em}\eee^3~~\textrm{and}~~
\eee^{12}_\pm:=\eee^1\pm i\hspace{.1em}\eee^2.\label{2bases}
\end{align}
Four base-vectors $\{\eee^{03}_\pm,\eee^{12}_\pm\}$ are independent from one another; any one-form objects can be expanded these basis like $\aaa=\alpha_1\eee^{03}_++\alpha_2\eee^{03}_-+\alpha_3\eee^{12}_++\alpha_4\eee^{12}_-$.
Connection one-form $\AAA_G^\pm$ with structure group $G$ is provided as 
\begin{align}
\AAA_G^\pm:=
[\AAA_G^\pm]^I\hspace{.2em}\tau^{~}_I=(a^I_{03}\hspace{.1em}\eee_+^{03}
+\hspace{.1em}a^I_{12}\hspace{.1em}\eee^{12}_\pm)\tau^{~}_I
\in V^\pm(\Omega^1)\otimes\ggg,\label{AAAdif}
\end{align}
where $\tau^{~}_I\in\ggg$ is a fundamental representation of $\ggg$. 
This connection provides
\begin{align*}
-ic^{\pm}_G\hspace{.1em}
\AAA_G^\pm\wedge\AAA_G^\pm&=
\frac{c^\pm_{G}}{2}f^I_{\hspace{.1em}JK}[\AAA_G^\pm]^J\wedge[\AAA_G^\pm]^K\hspace{.2em}\tau^{~}_I
\in V^\pm(\Omega^2)\otimes\ggg
\end{align*}
for any $\theta\in[0,1]$; thus, bases $\{\eee^{03}_\pm,\eee^{12}_\pm\}$ and  $\{\eee^{03}_\pm,\eee^{12}_\mp\}$ give $\AAA_\SU^+$ and $\AAA_\SU^-$, respectively.
We note that $\AAA_G^+\wedge\AAA_G^-=\AAA_G^-\wedge\AAA_G^+\neq0$.
When $\AAA_G^\pm$ fulfils holomorphic relations
\begin{align}
\left(\partial_0-\kappa(\theta)^{-1}\partial_3\right)a^I_\bullet=0,&~~\textrm{and}~~
\left(\partial_1\pm i\partial_2\right)a^I_\bullet=0,\label{holomorphism}
\end{align}
for both $a^I_\bullet=a^I_{03}$ and $a^I_{12}$, they provide $d\AAA_G^\pm\in V^\pm(\Omega^2)\otimes\ggg$, and (\ref{A+A+}) is maintained with any $\theta\in[0,1]$.

For structure groups $SO(1,3)$ and $SO(4)$, a space of the spin-connection is also split into two subspaces.  
A curvature is obtained owing to the structure equation (\ref{RRR}) in one of the subspaces of connections as a consequence of isomorphism (\ref{ismor}).
More precisely, each curvature is obtained owing to the structure equation:
\begin{align}
\left[\RRR^\pm\right]^{ab}&:=d\left[\www^\pm\right]^{ab}+\cG^{\hspace{-.6em}\pm}
\hspace{.2em}[\bm\eta^{~}_\theta]_\bcdots\left[\www^\pm\right]^{a\bcdot}\wedge\left[\www^\pm\right]^{\bcdot b},\label{StEqpm}
\intertext{where}
\left[\www^\pm\right]^{ab}&:=\left(
\begin{array}{cccc}
 0& \Omega^{1}\hspace{.1em}\eee^{03}_+
     +\Omega^{2}\hspace{.1em}\eee^{12}_\pm&
   -\Omega^{3}\hspace{.1em}\eee^{03}_+
     -\Omega^{4}\hspace{.1em}\eee^{12}_\pm&
    \Omega^{5}\hspace{.1em}\eee^{03}_+
      +\Omega^{6}\hspace{.1em}\eee^{12}_\pm\\
    ~&0&
     \hspace{.8em}\Omega^{5}\hspace{.1em}\eee^{03}_+
    +\Omega^{6}\hspace{.1em}\eee^{12}_\pm&
     \Omega^{3}\hspace{.2em}\eee^{03}_+
    +\Omega^{4}\hspace{.2em}\eee^{12}_\pm\\
    ~&~&0&
     \Omega^{1}\hspace{.2em}\eee^{03}_+
    +\Omega^{2}\hspace{.2em}\eee^{12}_\pm\\
    ~&~&~&0
\end{array}
\right).%\nonumber\\
\label{wwwpm}
\end{align}
$\Omega^i$ ($i=1,\cdots,6$) are functions of a space-time point independent from one another.
In (\ref{wwwpm}) an omitted part is obvious owing to anti-symmetry of the spin-connection.
In addition to that, when tensor coefficients fulfil holomorphic conditions (\ref{holomorphism}) for  all $\Omega^\bullet$, it provides $d[\www^\pm]\in V^\pm(\Omega^2)$ for any $\theta\in[0,1]$; thus, the $\Z_2$-grading structure-equation (\ref{StEqpm}) is maintained in the $\theta$-metric space.
We note that $\www^\pm\wedge\www^\pm\in V^\pm(\Omega^2)$ are anti-symmetric tensor.
On the other hand $\www^\pm\wedge\www^\mp$ are not anti-symmetric nor symmetric tensor.
After anti-symmetrise those tensors, it is provided that $AS[\www^+\wedge\www^-]=AS[\www^-\wedge\www^+]\neq0$, where $AS[M]=(M-M^T)/2$ for square matrix $M$. 

%
% Kramers pair of connection and curvature
%
\paragraph{\textbf{Kramers pair of connection and curvature:}}
For the given principal bundle $(E,\pi,M,G)$, we introduced principal connection $\AAA_G^{~}$ and curvature $\FFF_G^{~}$.
The dual curvature is provided using the Hodge-dual operator.
On the other hand, dual connection $\hat\AAA_G^{~}$ is provided as a solution of the structure equation (\ref{steqAG}).
The SD- and ASD-curvatures are provided owing to the Kramers pair as
\begin{align*}
\FFF_G^\pm:=\frac{1}{2}\left(\FFF_G^{~}\pm\hat\FFF_G^{~}\right)
\in V^\pm(\Omega^2)\otimes\ggg.
\end{align*}
Corresponding connections are provided as solutions of structure equations (\ref{A+A+}); they are represented as (\ref{AAAdif}) owing to the bases (\ref{2bases}).
Due to non-linearity of the structure equation, the SD- and ASD-connections are not obtained using a linear combination of the Kramers pair of connections in general.
Therefore, a relation between $(c^{~}_G,\hat{c}^{~}_G)$ and $(c^+_{G},c^-_{G})$ is not trivial.

% Remark
\begin{remark}\label{2bundles}
Suppose $\M$ is smooth oriented four-dimensional manifold with the structure group $SU(N)$, connection $\AAA^{~}_G$ and its dual $\hat\AAA^{~}_G$ belong to the same bundle.
\end{remark}
\begin{proof}
The principal curvature and its dual are represented using the trivial basis as
\begin{align*}
&{\FFF}^I_G\wedge{\FFF}^J_G=
\frac{1}{4}\f^I_\bcdots\f^J_\bcdots
\eee^\bcdot\wedge\eee^\bcdot\wedge\eee^\bcdot\wedge\eee^\bcdot=
\frac{\sigmath}{4}\left[\epsilon_\theta\right]^\bcdott
\f^I_\bcdots\f^J_\bcdots\hspace{.2em}\vvv,\\
&\hat\FFF^I_G\wedge\hat\FFF^J_G=
\frac{\sigmath}{4}\left[\epsilon_\theta\right]^\bcdott
\hat\f^I_\bcdots\hat\f^J_\bcdots\hspace{.2em}\vvv=
\frac{\sigmath}{4}\left[\epsilon_\theta\right]^\bcdott
\f^I_\bcdots\f^J_\bcdots\hspace{.2em}\vvv,\\
&\Rightarrow{\FFF}^I_G\wedge{\FFF}^J_G=\hat\FFF^I_G\wedge\hat\FFF^J_G.
\end{align*}
We note that this relation holds also in the $\theta$-metric space.
Thus, we obtain that
\begin{align*}
\Tr_G[\FFF^{~}_G]^2&=\left(\sum_I{\FFF}^I_G\hspace{.1em}\Tr\left[\tau^I\right]\right)\wedge
\left(\sum_J{\FFF}^J_G\hspace{.1em}\Tr\left[\tau^J\right]\right),\\&=
\sum_{I,J}{\FFF}^I_G\wedge{\FFF}^J_G
\hspace{.1em}\Tr\left[\tau^I\right]\hspace{.1em}\Tr\left[\tau^J\right]=\Tr_G[\hat\FFF^{~}_G]^2,\\
%%%
\Tr_G[\FFF^{~}_G\wedge\FFF^{~}_G]&=
\sum_{I,J}{\FFF}^I_G\wedge{\FFF}^J_G\hspace{.1em}\Tr\left[\tau^I\cdot\tau^J\right]=
\hspace{.1em}\Tr_G[\hat\FFF^{~}_G\wedge\hat\FFF^{~}_G].
\end{align*}
Therefore, two bundles have the same second Chern class as 
\[
c_2(\FFF^{~}_G)=\frac{1}{8\pi^2}\int(\Tr_G[\FFF^{~}_G\wedge\FFF^{~}_G]-
\Tr_G[\FFF^{~}_G]^2)=c_2(\hat\FFF^{~}_G).
\]

In a four-dimensional manifold, a fibre bundle with $SU(N)$ structure group is uniquely characterized by $c^{~}_2(\FFF^{~}_G)$.
When the structure group is $SU(N)$, a matrix trace of Lie-algebra $\bm\tau^{~}_\SU\in\sss\uuu(N)$ yields
\begin{align*}
\Tr\left[\tau_\SU^I\right]=0,~~&\textrm{and}~~~
\Tr\left[\tau_\SU^I\cdot\tau_\SU^J\right]=\left\{
\begin{array}{cc}
2&(I=J)\\
0&(I\neq J)
\end{array}
\right..
\end{align*}
Thus, we obtain that:
\begin{align*}
\begin{array}{ccccl}
\Tr^{~}_\SU[\FFF^{~}_\SU]^2&=&\Tr^{~}_\SU[\hat\FFF^{~}_\SU]^2&=&0,\\
\Tr^{~}_\SU[\FFF^{~}_\SU\wedge\FFF^{~}_\SU]&=&\Tr^{~}_\SU[\hat\FFF^{~}_\SU\wedge\hat\FFF^{~}_\SU]
&=&2\sum_{I}{\FFF_\SU}^I\hspace{-.1em}\wedge{\FFF_\SU}^I.
\end{array}
\end{align*}
Therefore, the remark is maintained. 
\end{proof}
\noindent
%%%
When the structure group is $U(1)$, we obtain that $\Tr^{~}_{\Uo}[\FFF^{~}_{\Uo}]^2=\Tr^{~}_{\Uo}[\FFF^{~}_{\Uo}\wedge\FFF^{~}_{\Uo}]$; thus, $c^{~}_2(\FFF^{~}_{\Uo})=c^{~}_2(\hat\FFF^{~}_{\Uo})=0$.
%Therefore, the bundle is characterized only by $c^{~}_1(\FFF^{~}_G)$.
The Kramers pair of principal connections belongs to the same bundle for the $SU(N)$ structure group; thus, coupling constants $c^{~}_\SU$ and $\hat{c}^{~}_\SU$ must have the same value. 

On the other hand,  $\FFF_G^{+}\wedge\FFF_G^{+}\neq\FFF_G^{-}\wedge\FFF_G^{-}$ in general; thus, $\FFF_G^{+}$ and $\FFF_G^{-}$ belong to the different bundles since they have different second Chern classes: $c^{~}_2(\FFF^{+}_G)\neq c^{~}_2(\FFF^{-}_G)$.
Similar calculations for the Lorentz curvature show $c^{~}_2(\RRR)=c^{~}_2(\hat\RRR)$ and $c^{~}_2(\RRR^+){\neq}c^{~}_2(\RRR^-)$.  
When SD curvature $\FFF_G^+$ is provided owing to connection $\AAA^{+}_G$, dual curvature
$\hat\FFF_G^+=\FFF_G^+$ is also provided using the same connection $\hat\AAA^{+}_G=\AAA^{+}_G$ with coupling constants ${c}^{+}_{G}={c}^{-}_G$.
For the ASD case, connection $-\AAA^{-}_G$ does not provide the dual curvature except for the $U(1)$ structure group owing to non-linearity of the structure equation.
A simple  solution is setting ${c}^{-}_G=-{c}^{+}_G$ and $\hat\AAA^-_G=-\AAA^{-}_G$ simultaneously such that:
\begin{align*}
\hat\AAA^-_G&:=-\AAA^-_G,\\
\left[\hat\FFF^-_G\right]^I
&:=d\hat\AAA^-_G-\frac{{c}^{+}_G}{2}\hat\AAA^-_G\wedge\AAA^-_G.
\end{align*}
We note that coupling constants ${c}^{+}_G$ and ${c}^{-}_G$ can be different to each other because they belong to different bundles. 
Another solution is defining a dual connection and  a dual curvature for the given $\AAA^-_G$ as:
\begin{align*}
\left[\hat\AAA^-_G\right]^I\tau^{~}_I
&:=-\left[\AAA^-_G\right]^{N^{~}_G-I}\tau^{~}_I,\\
\left[\hat\FFF^-_G\right]^I\tau^{~}_I
&:=\left(d\left[\hat\AAA^-_G\right]^I+
\frac{{c}^{+}_G}{2}f^I_{~JK}
\left[\hat\AAA^-_G\right]^J\wedge\left[\hat\AAA^-_G\right]^K\right)
\tau^{~}_{N^{~}_G-I}.
\end{align*}
where $N^{~}_G:=N^2-1$ is a degree of freedom for the $SU(N)$ group.
Direct calculation shows $\HD(\FFF^-_G)=-\FFF^-_G$ for both definitions.

%
% Secondary $\Z_2$-grading structure
%
\paragraph{\textbf{Secondary superspace:}}
For the Kramaers pairs of a connection and a curvature, $(\TxT)$-matrix representations are provided as
\begin{align*}
\bm\AAA^{~}_2&:=\left(
\begin{array}{cc}
	\AAA^{~}_G &  0\\
 	    0  & \hat\AAA^{~}_G
\end{array}
\right)&\in&\hspace{.2em}{\bm{M}}^+,\\
\bm\FFF^{~}_2&:=
d{\bm\AAA^{~}_2}+c^{~}_G\hspace{.2em}{\bm\AAA^{~}_2}\wedge{\bm\AAA^{~}_2}=
\left(
\begin{array}{cc}
	\FFF^{~}_G &  0\\
 	    0  & \hat\FFF^{~}_G
\end{array}
\right)&\in&\hspace{.2em}{\bm{M}}^+.%\label{2x2conmcurv}
\end{align*}
A ($\TxT$)-matrix representation of the covariant differential concerning the structure group $G$ is provided as
\begin{align}
D^{~}_{\hspace{-.1em}G}:=\left(
\begin{array}{cc}
	d_{G} &  0\\
 	    0  & \hat{d}_{G}
\end{array}
\right)=\left(
\begin{array}{cc}
	\bm{1}_G\hspace{.1em}d+c^{~}_G\hspace{.1em}[\AAA^{~}_G,~]_\wedge &  0\\
 	    0  & \bm{1}_G\hspace{.1em}d+{c}^{~}_G\hspace{.1em}[\hat\AAA^{~}_G,~]_\wedge
\end{array}
\right).\label{2x2Dirac}
\end{align}
The Dirac operator in the secondary superpsace is defined as
\begin{align}
\D^{~}_{\hspace{-.1em}G}&:=\gamma D^{~}_{\hspace{-.1em}G}=\left(
\begin{array}{cc}
	0 &\hat{d}^{~}_G \\
 	d_G^{~} & 0
\end{array}
\right)&\in{\bm{M}}^-.\label{Dirac2x2}
\end{align}
This operator is a first-order differential operator with an odd-parity; thus, it flips a parity of an operand matrix due to superalgebra (\ref{superalgebra}) and fulfils the definition of the Dirac operator. 
The Bianchi identities in this representation are given as
\begin{align*}
\D^{T}_{\hspace{-.1em}G}\D^{~}_{\hspace{-.1em}G}=
D^{~}_{\hspace{-.1em}G}D^{~}_{\hspace{-.1em}G}=
-i{{c}^{~}_G}\hspace{.2em}{\bm\FFF^{~}_2},~~\textrm{and}~~
\D^{~}_{\hspace{-.1em}G}\hspace{.1em}{\bm\FFF^{~}_2}=0.
\end{align*}

The Hodge-dual operator also has a ($\TxT$)-matrix representation:
\begin{align*}
\bm\HD_2&:=\Gamma\gamma\HD
=\left(
\begin{array}{cc}
	0 &  -\HD\\
 	\HD  & 0
\end{array}
\right)=-\bm\HD_2^{-1}&\in{\bm{M}}^-.
\end{align*}
This operator acts on the connection and the curvature as
\begin{align*}
\bm\HD_2\hspace{.2em}\bm\XXX_2\hspace{.2em}\bm\HD_2^{-1}
=\bm\XXX_2\big|_{\XXX^{~}_G\leftrightarrow\hat\XXX^{~}_G},
\end{align*}
where $\XXX^{~}_G$ is $\AAA^{~}_G$ or $\FFF^{~}_G$.
We note that ${\AAA}^\pm_G$ is not eigenvector of the Hodge-dual operator because the Hodge-dual of a one-form object is a three-form object in a four-dimensional manifold.
Here, we denote the action of operator $\HD$ for a connection one-form with short-hand notation of
\begin{align}
%\HD(\AAA^{~}_G)=\hat\AAA^{~}_G~~\textrm{and}~~
%\HD(\hat\AAA^{~}_G)=\AAA^{~}_G.\label{HAAA}
\hat\AAA^{~}_G:\AAA^{~}_G\mapsto\FFF^{~}_G\mapsto\hat\FFF^{~}_G\mapsto\hat\AAA^{~}_G
\label{HAAA}
\end{align}
Operator $\HD$ in $\bm\HD_2$ for $\KP(\AAA^{~}_G)$ is a short-hand notation understood as through (\ref{HAAA}).
The Hodge-dual operator is homomorphic concerning the Dirac operator; e.g., the second Bianchi identity fulfils
\begin{align*}
\HD\left(\D^{~}_{\hspace{-.1em}G}\hspace{.1em}\bm\FFF_2\right)=
\bm\HD^{~}_2\left(\D^{~}_{\hspace{-.1em}G}\hspace{.1em}\bm\FFF^{~}_2\right)\bm\HD_2^{-1}=
\D^{~}_{\hspace{-.1em}G}\hspace{.1em}
\bm\FFF^{~}_2\big|_{\FFF^{~}_G\leftrightarrow\hat\FFF^{~}_G}=0.
\end{align*}
The structure equation is also homomorphic such that:  
\begin{align*}
\HD\left(
d\bm\AAA^{~}_2+{c}^{~}_G\hspace{.2em}\bm\AAA^{~}_2\wedge\bm\AAA^{~}_2\right)=
\left(d\bm\AAA^{~}_2+{c}^{~}_G\hspace{.2em}\bm\AAA^{~}_2\wedge\bm\AAA^{~}_2
\right)\big|_{\AAA^{~}_G\leftrightarrow\hat\AAA^{~}_G}=\HD\left(\bm\FFF^{~}_2\right).
\end{align*}

We introduced the section field in primary- and secondary superspace.
The section field is provided as a solution to the equation of motion and belongs to a fundamental representation of the structure group.
We investigate the $\Z_2$-grading structure of the equation.
Suppose the equation of motion consists of the Dirac operator in the second superspace; thus, the equation is the first-order differential operator in the primary superspace.
In addition to that, the equation respects the structure group of the principal bundle.
We introduce the following simple equation as an example:
\begin{align}
i\hspace{.1em}d^{~}_G\phi-\mu\phi=0~~\textrm{and}~~
i\hspace{.1em}\hat{d}^{~}_G\hat\phi-\hat\mu\hat\phi=0.\label{eomphi}
\end{align}
The Hodge-dual operator acts on scalar fields as (\ref{phihat}).
Constants $\mu$ and $\hat\mu\in\R$ are called \emph{mass} in physics.
In the secondary superspace, equations (\ref{eomphi}) has a representation such that:
\begin{align*}
\left(i\D^{~}_{\hspace{-.1em}G}-\gamma\bm\mu_2\right)\bm\Phi_2=0,~~\textrm{where}~~
\bm\mu_2:=\left(
\begin{array}{cc}
 \mu&   0\\
	0      &  \hat\mu
\end{array}
\right).
\end{align*}
The Klein--Gordon equation is provided in the secondary superspace as 
\begin{align*}
\left(i\D^{~}_{\hspace{-.1em}G}-\gamma\bm\mu_2\right)^T
\left(i\D^{~}_{\hspace{-.1em}G}-\gamma\bm\mu_2\right)\bm\Phi_2=0
\end{align*}

In summary, we showed a duplex $\Z_2$-grading structure is generally induced by the parity operator in a four-dimensional manifold:
The primary structure is induced utilizing the Hodge-dual operator on two-form objects and extended to one-form objects owing to isomorphism (\ref{ismor}).
The secondary structure is constructed using the $(\TxT)$-matrix representation and superalgebra using matrix multiplication.
The Dirac operator is defined owing to a superalgebra of $(\TxT)$-matrices in the secondary superspace.

%
% Index theorem and cobordism
%
\subsection{Index theorem and cobordism}\label{IndexCobordism}
The Atiyah-Singer index theorem for a closed manifold without boundaries appeared in 1968, and it was extended to compact manifolds with boundaries by Atiyah, Patodi and Singer in 1975.
In theoretical physics, the Atiyah--Singer index theorem of the Dirac operator provides a deep insight into anomalies of the gauge field theory.
The Atiyah--Patodi--Singer index theorem was recently applied to the domain-wall fermion after appropriately modifying the boundary condition\cite{PhysRevD.96.125004,Fukaya:2019ytk,Fukaya:2019qlf}.
Originally the Atiyah-Singer index theorem has been provided in a Euclidean manifold for the elliptic type Dirac operator: B\"{a}r and Strohmaier\cite{wrro123059, bar2018} extended the Atiyah--Patodi--Singer type index theorem to globally hyperbolic space-times in 2017.

The Atiyah--Singer index theorem asserts that characteristic (analytic) indices owing to the Dirac operator are equivalent to topological indices.
More precisely, the cohomology of the manifold determines an index of the Dirac operator in a compact and oriented even-dimensional Riemannian manifold.
We note that the Dirac operator must be an elliptic operator for the theorem's validity; thus, the Minkowski manifold does not maintain the theorem as it is.
We propose to define a topological index in the Minkowski manifold using that in Eudlidean space cobordant to the Minkowski manifold.
The Minkowski space is homotopy equivalence to Eudlidean space owing to the $\theta$-metric.
%
% Index of Dirac operator
%
\subsubsection{Index of Dirac operator}\label{AppIndex}
\paragraph{\textbf{Characteristics:}}
First, we introduced the Chern characteristic and $\hat{A}$-genus appearing in the Atiyah--Singer index theorem.
They are defined in several contexts; definitions in linear algebra and Chern--Weil theory are provided here.
For the Chern characteristic, a definition in a superbundle is also given.
% Definition
\begin{definition}[Chern characteristic]\label{ChChSer}
The Chern characteristics are defined as follows:
\begin{enumerate}
\item \textbf{Chern characteristic in Lie algebra:}
Suppose $\{\lambda_1,\cdots,\lambda_n\}$ are eigenvalues of $F\in\ggg\lll(n,\C)$ with $F+F^\dagger=0$; namely a skew-adjoint matrix.
Chern characteristic $ch(F)$ is defined such that:
\begin{align*}
ch(F):=\sum_{j=1}^n e^{i\lambda_j/2\pi}.
\end{align*}
\item \textbf{Chern characteristic in fibre-bundle:}
When $\FFF\in\Omega^2\otimes\ggg$ is a curvature two-form of a fibre-bundle with structure group $G$ and its Lie algebra $\ggg$, Chern characteristic $ch(\FFF)\in\WWW(\FFF)$ is defined such that:
\begin{align*}
ch(\FFF):=e^{i\FFF/2\pi},
\end{align*}
where $\WWW(\FFF)$ is an invariant polynomial of the curvature two-forms in Weil algebra.
\item \textbf{Chern characteristic in super bundle:}
When $\FFF=\FFF^++\FFF^-\in\Omega^2\otimes\ggg$ is a curvature two-form of a $\Z_2$-grading superbundle $\S$ with structure group $G$ and its Lie algebra $\ggg$, Chern characteristic $ch(\S)\in\WWW(\FFF)$ is defined such that:
\begin{align*}
ch(\S):=ch(\FFF^+)-ch(\FFF^-)=e^{i\FFF^+/2\pi}-e^{i\FFF^-/2\pi}.
\end{align*}
\end{enumerate}
\end{definition}
\noindent
%%%
The first several terms of the Chern characteristic is provided such that:
\begin{align*}
ch(\FFF)&=N^{~}_G+c_1(\FFF)+\frac{1}{2}\left(c_1^2(\FFF)-2c^{~}_2(\FFF)\right)+
\frac{1}{6}\left(c_1^3(\FFF)-3c_1(\FFF)\hspace{.2em}c^{~}_2(\FFF)+3c_3(\FFF)\right)+\cdots,
\end{align*}
where $c_j(\FFF)$ is the $j$'th Chern class, and $N^{~}_G$ is a dimension of the structure group.
The first three terms may appear for the four-dimensional space.

% Definition
\begin{definition}[$\hat{A}$-genus]\label{hAgen}
$\hat{A}$-genera are defined as follows:
\begin{enumerate}
\item \textbf{$\hat{A}$-genus in Lie algebra:}
Suppose $\{{\pm}i\lambda_1,\cdots,{\pm}i\lambda_k\}$ are eigenvalues of $F\in\ggg\lll(2k,\R)$ with $F+F^T=0$; namely askew-symmetric matrix.
$\hat{A}$-genus is defined such that:
\begin{align*}
\hat{A}(F):=\prod_{j=1}^k\frac{-\lambda_j/4\pi}{\sinh{\left(-\lambda_j/4\pi\right)}}.
\end{align*}
\item \textbf{$\hat{A}$-genus in fibre-bundle:}
When $\FFF\in\Omega^2\otimes\ggg$ is a curvature two-form of a bundle with structure group $G$ and its Lie algebra $\ggg$, $\hat{A}$-genus $\hat{A}(\FFF)\in\WWW(\FFF)$ is defined such that:
\begin{align*}
\hat{A}(M):=\mathrm{det}^{1/2}\left(\frac{-i\FFF/4\pi}{\sinh{\left(-i\FFF/4\pi\right)}}\right).
\end{align*}
\end{enumerate}
\end{definition}
\noindent
$\hat{A}$-genus is sometimes referred to as the Dirac-genus in physics\cite{bertlmann2000anomalies}.
%%%
The first several terms of $\hat{A}$-genus in manifold $M$ is provided such that:
\begin{align*}
\hat{A}(M)=1-\frac{1}{24}p_1(M)&+\frac{1}{5760}
\left(7\hspace{.1em}p^2_1(M)-4\hspace{.1em}p_2(M)\right)\\
&+\frac{1}{967680}\left(
-31\hspace{.1em}p_1^3(M)+44\hspace{.1em}p_1(M)\hspace{.1em}p_2(M)
-16\hspace{.1em}p_3(M)\right)
+\cdots,
\end{align*}
where $p_j(M)$ is the $j^\textrm{th}$ Pontrjagin class in an even-dimensional manifold $M$.
The first three terms may appear for the four-dimensional space.

\paragraph{\textbf{Spin manifold:}}
Suppose a spin structure is defined in $n$-dimensional smooth and oriented manifold $M$, where $n$ is even integer. 
It is denoted by  $Spin(M)$ and  Clifford algebra of $Spin(M)$ is denoted as $\Cl(M)$.
Manifold $M$ is lifted to a spin manifold owing to Clifford module $\S:=\Cl(M)\otimes\C$, and spinor group is doubly covering $M$ and induces spinor bundle of $\left(\S,\pi_\Sp,M, Spin(M)\right)$.
Clifford algebra is introduced in base manifold $M$:
A representation of $\Cl(M)$ is denoted as $\gamma_\Sp^a$ $(a=0,\cdots,n-1)$ yielding
\begin{align}
\left\{\gamma_\Sp^a,\gamma_\Sp^b\right\}&:=
\gamma_\Sp^a\times_\Sp\gamma_\Sp^b+\gamma_\Sp^b\times_\Sp\gamma_\Sp^a~=~
2\eta^{ab}_{~}\hspace{.2em}{\bm{1}}_\Sp,\label{cf1}
\end{align}
where ${\bm{1}}_\Sp$ is a unit Clifford algebra, and $\times_\Sp$ represents Clifford product.
Hereafter. Clifford product $\times_\Sp$ is omitted for simplicity.
Although each component can be set as Hermitian self-conjugate or anti-self-conjugate, all of them cannot be set entirely self-conjugate nor anti-self-conjugate, when manifold  $M$ is non-compact.
This study utilizes a representation such that 
\begin{align}
(\gamma_\Sp^0)^\dagger=\gamma_\Sp^0,
&~~(\gamma_\Sp^i)^\dagger=-\gamma_\Sp^i,~\textrm{for}~~i=1,\cdots,n-1\label{cf2}.
\end{align}
We note that  a relation  $(\gamma_\Sp^a)^\dagger=\gamma_\Sp^0\gamma_\Sp^a\gamma_\Sp^0$ follows from (\ref{cf1}) and (\ref{cf2}) for any $a{\leq}n-1$.
Representation \textrm{(\ref{cf2})} is always possible independent from the chosen basis.
The chiral operator and the projection operator are defined, respectively, as 
\[
\Gamma_{\hspace{-.1em}\Sp}:=i
\gamma_{\hspace{-.1em}\Sp}^0\cdots
\gamma_{\hspace{-.1em}\Sp}^{n-1},~~\textrm{and}~~
P^\pm_\Sp:=\frac{\bm{1}_\Sp\pm\Gamma_\Sp}{2};%\label{GammaP} 
\]
they are independent from chosen basis owing to their definition.
The chiral operator fulfils relations:
\begin{align}
(\GammaSp)^2=\bm{1}_\Sp,&~~~\GammaSp\gamma_\Sp^a=
-\gamma_\Sp^a\GammaSp,\label{g5ga}
\end{align}
for any $a{\leq}n-1$.

\paragraph{\textbf{Index theorem:}}
The elliptic-type Dirac-operator in a compact space has an analytic index equivalent to the topological index owing to the Atiyah-- Singer theorem.
We introduce the index of the Dirac operator as follows:
Suppose $\D:(M,\S)\rightarrow(M,\S)$ is a self-adjoint Dirac operator.
Clifford algebra $\gamma_\Sp^a$ induces a $\Z_2$-grading space in the spinor module as $\S=\S^+\oplus\S^-$ owing the chiral operator $\Gamma_\Sp$ such that:
\begin{align*}
\S^\pm:=\{\psi\in\S|\Gamma_\Sp\psi=\pm\psi\},
\end{align*}
where $\psi\in\S$ is a spinor field defined as a section in $\Tm$.
Restriction of $\D$ acting on $\S^\pm$ are defined as $\D^\pm:=\D P_\Sp^\pm$.
Spinor fields are divided into two parts as $\psi^\pm:=P_\Sp^\pm\psi\in\S^\pm$.
An index of the Dirac operator and a super trace are introduced as follows: 
% Definition
\begin{definition}[Index of the Dirac operator\cite{berline2003heat}]\label{IDO}
An index of Dirac operator $\d$ is defined as
\begin{align*}
ind(\D)&:=\mathrm{dim}\hspace{.2em}\mathrm{Ker}\left(\D^+\right)
-\mathrm{dim}\hspace{.2em}\mathrm{Ker}\left(\D^-\right),\\&=
\mathrm{dim}\hspace{.2em}\mathrm{Ker}\left(\D^+\right)
+\mathrm{dim}\hspace{.2em}\mathrm{Coker}\left(\D^+\right).
\end{align*}
\end{definition}
%%%
\noindent
The heat kernel, denoted as $\mathcal{H}_s$, is the unique operator that satisfies the equation such that:
\begin{align*}
\left(\frac{\partial}{\partial s}+\D^2\right)\mathcal{H}_s&=0.
\end{align*}
The McKean--Singer theorem\cite{10.4310/jdg/1214427880} states a relation between a heat kernel and an index of the Dirac operator as follows:
% Remark
\begin{remark}[McKean--Singer]\label{McS}
Suppose $\D$ is a Dirac operator on a compact space and $e^{-s\hspace{.1em}\D^\mp\D^\pm}\in\mathcal{H}_s$ . 
The index of $\D$ is provided as
\begin{align*}
ind(\D)=Str\left(e^{-s\hspace{.1em}\D^2}\right)&=
\Tr\left[e^{-s\hspace{.1em}\D^-\D^+}\right]-\Tr\left[e^{-s\hspace{.1em}\D^+\D^-}\right]
=\Tr\left[\Gamma\hspace{.2em}e^{-s\hspace{.1em}\D^2}\right].
\end{align*}
\end{remark}
\begin{proof}
When $H^\pm_\lambda$ are eigenspaces of $\D^\mp\D^\pm$ with an eigenvalue $\lambda$, and $n^\pm_\lambda=\mathrm{dim}H^\pm_\lambda$, it is provided that:
\begin{align}
\Tr\left[e^{-s\hspace{.1em}\D^-\D^+}\right]-\Tr\left[e^{-s\hspace{.1em}\D^+\D^-}\right]
=\sum_{\lambda\geq0}\left(n^+_\lambda-n^-_\lambda\right)e^{-s\lambda}.\label{nmn}
\end{align}
When $\lambda\neq0$, operator $\lambda^{-1}\circ\D^-\circ\D^+:H^+_\lambda\rightarrow H^+_\lambda$ is an identity map; thus, $n^+_\lambda-n^-_\lambda=0$.
Therefore, it is  is maintained that 
\[(\ref{nmn})=n^+_0-n^-_0=ind(\D).
\]

Moreover, we obtain that 
\begin{align*}
\Gamma_\Sp\D\D&=\D\D\Gamma_\Sp=\d\d(P_\Sp^+-P_\Sp^-)=
(\D^++\D^-)(\D^+-\D^-),\\
&=\D^-\D^+ - \D^+\D^-,
\end{align*}
using  relations $\Gamma_\Sp=P_\Sp^+-P_\Sp^-$, $\D=\D^++\D^- $ and 
$
\D^\pm\D^\pm=\D P_\Sp^\pm\D P_\Sp^\pm=\D\D P_\Sp^\mp P_\Sp^\pm=0.
$
Thus, it is provided that
\begin{align*}
\Gamma(\D\D)^n&=
\left(\D^-\D^+\right)^n - \left(\D^+\D^-\right)^n.
\end{align*}
Therefore, the second equality is maintained.
\end{proof}
%%%

% Remark
\begin{remark}[Atiyah--Singer\cite{Atiyah:1968mp}]\label{ASI}
Suppose $M$ is a compact and oriented even-dimensional spin manifold with no boundary, $\S$ is a Clifford module.
\begin{enumerate}
\item For Dirac operator $\D$ in $\Gamma(M,\S)$, an index of $\D$ is provided such that:
\begin{align}
ind(\D)=\int_M\hat{A}(M),\label{ASIT1}
\end{align}
when $\D$ is the Fredholm operator.
\item Suppose $\W$ is a vector bundle in $M$.
For twisting Dirac operator $\D_\W$ in $\Gamma(M,\W\otimes\S)$, an index of $\D_\W$ is provided such that:
\begin{align}
ind(\D_\W)=\int_M\hat{A}(M)ch(\W),\label{ASIT2}
\end{align}
when $\D_\W$ is the Fredholm operator.
\end{enumerate}
\end{remark}
\noindent
A proof of the  theorem is given in literatures, e.g., chapter 4 in Ref.\cite{doi:10.1063/1.1705276}.

Suppose $\RRR$ is a curvature of four-dimensional manifold $M_4$ and  $\FFF^I\tau^{~}_I$ is a curvature of principal bundle $\W$ with structure group $G$  whose Lie algebra is $\tau^{~}_I\in\ggg$. 
The index is provided as
\begin{align}
ind(\D_\W)&=\frac{1}{(2\pi)^2}\int_{M_4}\left(
\frac{n}{48}\Tr_{\Tm}[\RRR\wedge\RRR]
-\frac{1}{2}\Tr_G[\FFF\wedge\FFF]
\right),\nonumber\\
&=-\frac{1}{(4\pi)^2}\int_{M_4}\left(
\frac{n}{12}\eta^{~}_\bcdots\eta^{~}_\stars\RRR^{\bcdot\star}\wedge\RRR^{\bcdot\star}
+i\sum_{I=1}^{n}f^I_{~JK}\FFF^J\wedge\FFF^K
\right).
\label{ASIT1in4}
\end{align}
We note that $\Tr_{\Tm}[\RRR]=\eta_\bcdots\RRR^\bcdots=0$ due to anti-symmetry of $\RRR$ with respect to its tensor indices.
The index in the $\theta$-metric space is provided by replacing $\bm\eta\rightarrow\bm\eta^{~}_\theta$.

%
% cobordism
%
\subsubsection{Cobordism}\label{Appcob}
Ren\'{e} Tom has proposed a method to categorise compact and orientable $C^\infty$-manifolds into the equivalence classes named ``\emph{cobordant}\cite{tom1,tom2}'' in 1953.
When a disjoint union of two $n$-dimensional manifolds is the boundary of a compact $(n+1)$-dimensional manifold, they are cobordant to each other; the cobordant is an equivalence relation and quotient space 
\begin{align*}
\Omega_C^n:=\{n\textrm{-dimensional~closed~manifolds}\}/\hspace{-.2em}\sim 
\end{align*}
forms an abelian group, where $\sim$ is an equivalent relation with respect to cobordism.
We present a definition and elementary remarks of cobordism in this section.
For details of the cobordism theory, see, e.g., Chapter 7 in Ref.\cite{hirsch1976differential}.
% Definition
\begin{definition}[cobordant]\label{cob}
Suppose $M$ and $N$ are compact and oriented $n$-dimensional smooth manifolds, and $\omega_\bullet$ denotes an orientation of oriented manifold $\bullet$.
$(M,\omega_M)$ and $(N,\omega_N)$ are called ``cobordant'' to each other, if there exists a compact and oriented $(n\hspace{-.1em}+\hspace{-.2em}1)$-dimensional smooth manifold $(W,\omega_W)$ such that
\begin{align*}
(\partial{W},\hspace{.2em}\partial{\omega_W})&=(M,-\omega_M)\amalg(N,\omega_N),
\end{align*}
and $W$ is called a cobordism from $M$ to $N$.
When $M$ and $N$ are cobordant to each other owing to $W$, it is denoted as $M\sim_WN$ in this study.
Cobordant induces an equivalent class among compact and oriented $n$-dimensional smooth manifolds, namely, a cobordant class denoted by $\Omega_{C}^n$.
\end{definition}
%%%

% Remark
\begin{remark}
Cobordant class $\Omega_{C}^n$ forms the abelian group concerning the disjoint-union operation.
\end{remark}
\begin{proof}
When map $M_0\rightarrow M_1$ is diffeomorphic, $M_0$ and $M
_1$ are cobordant; thus, $M_0$ and $M_1$ are cobordant with respect to closed period $\theta=[0,1]$; $M_0\sim_{M_0\otimes[0,1]}M_1$ such that
$
\partial{(M_\theta\otimes[0,1])}=(M_0\otimes0)\amalg(M_1\otimes1)
$
with appropriate orientations.
It fulfils the associative and commutative laws trivially.
Any compact manifold $V=\partial{W}$ is the zero element owing to $(M\amalg V)\sim_{W}M$.
The inverse of $(M,\omega_M)$ is given by $(M,-\omega_M)$.
Therefore, the remark is maintained.
\end{proof}
%%%

% Remark
\begin{remark}[Pontrjagin]\label{prg}~\\
Characteristic numbers are cobordism invariant.
\end{remark}
\begin{proof}
A proof is provided for the Pontrjagin number; it is sufficient to show that the Pontrjagin numbers on $M$ vanish when $M=\partial W$.
There exists one-dimensional trivial bundle $\xi$ such that  $TW\bigl|_M=TM\oplus\xi$; thus, for any inclusion map $\iota:M\hookrightarrow{W}$, relation $p_i(M)=\iota^\#(p_i(W))$ is fulfilled.
Therefore, for any invariant polynomial $f\in\WWW(M)$ to define characteristic numbers, it is provided such that:
\begin{align*}
f\left(p_1,p_2,\cdots\right)[M]=\int_Mf\left(p\left(M\right)\right)=
\int_{\partial M}\iota^\#\left(f\left(p\left(W\right)\right)\right)
=\int_Wdf\left(p\left(W\right)\right)=0,
\end{align*}
owing to the Stokes theorem; thus, the lemma is maintained.
Here, $[\bullet]$ denotes a cobordism class of manifold $M$.
\end{proof}
%%%

%
% $\theta$-metric space and index theorem
%
\subsection{$\theta$-metric space and index theorem}\label{themetric}
The Atiyah--Singer topological index is defined for the elliptic-type Dirac operator.
On the other hand, the Dirac operator in the Minkowski space has a hyperbolic type; thus, the theorem is not fulfilled as its original shape.
Gibbons has discussed an analytic continuation of functions defined in Eudlidean space to that in the Minkowski space in a study of a Fermion number violation in curved space-time\cite{GIBBONS1979431,GIBBONS198098}; he pointed out in Ref.\cite{GIBBONS1979431} that passing from Eudlidean to the Minkowski space is problematic as only non-singular manifolds admit an analytic continuation to real Riemannian space.
On the other hand, B\"{a}r and Strohmaier have provided the index theorem for a globally hyperbolic manifold with boundaries in a mathematically rigorous manner and applied the theorem to the total charge generation due to the gravitational and gauge field background\cite{bar2016}.

This section discusses the Atiyah--Singer topological index in the space-time manifold with no boundary using the cobordism theory with the $\theta$-metric.
Suppose $\M_\theta$ is a four-dimensional compact and oriented manifold with no boundary associating a standard Euclidean topological phase.
Two manifolds $\M_0:=\M_\theta\big|_{\theta\rightarrow0}$ and $\M_1:=\M_\theta\big|_{\theta\rightarrow1}$ are cobordant with respect to five-dimensional manifold $W_\theta:={\M_\theta}\otimes(\theta\in[0,1])$; manifold $W_\theta$ is cobordism from Lorentzian manifold $\M_0$ to Euclidean manifold $\M_1$ (see the upper half of Figure \ref{figCobor}).
% Figure
\begin{figure}[tb] 
\centering
\includegraphics[width={11cm}]{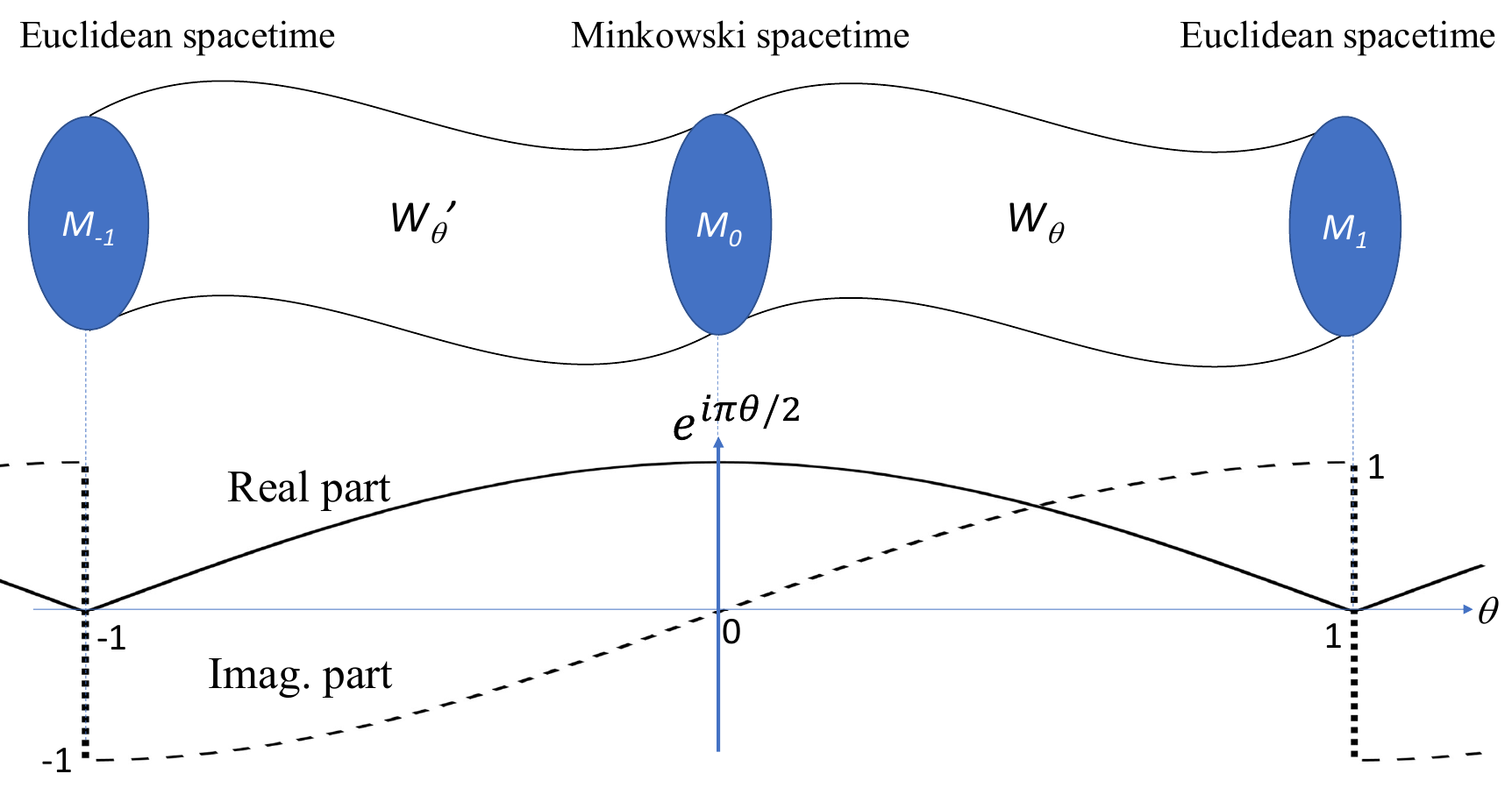}%
 \caption{
Schematic drowning of five-dimensional manifolds $W_\theta$ and $W'_\theta$, and its boundary manifolds.
Four-dimensional boundary manifolds $\M_1$ and $\M_{-1}$ have a Euclidean metric, and $\M_{-1}$ has a Minkowski metric.
Real (solid line) and imaginary (dashed line) parts of $e^{i\pi\theta/2}$ are show in the figure.
} 
\label{figCobor} 
\end{figure}
%%%

A representation of Clifford algebra in $Spin(\M_\theta)$ is defined to fulfil
algebra 
\begin{align}
\{\gamma_\theta^a,\gamma_\theta^b\}=2\eta_\theta^{ab}\bm{1}_\Sp\label{CliffordTheta}
\end{align}
with any values of $\theta$ as follows:
\begin{align}
\gamma^0_\theta:=e^{i\pi\theta/2}\hspace{.1em}\gamma^0_\Sp,&~~\textrm{and}~~\gamma^j_\theta:=\gamma^j_\Sp
~\textrm{for}~~j=1,2,3.\label{gamma-theta}
\end{align}
Here, $\bm{\gamma}_\theta$ is equivalent to $\bm{\gamma}_\Sp$ at $\theta=0$, and $\bm{\gamma}_E:=\bm{\gamma}_\theta|_{\theta\rightarrow1-0}$ is the Clifford algebra in Eudlidean space
(See lower half of Figure \ref{figCobor}).
Clifford algebra $\bm{\gamma}_E$ is Hermitian anti-self-conjugate owing to (\ref{cf2}).
Accordingly, the chiral- and projection-operators are defined as 
\begin{align*}
\Gamma_{\hspace{-.1em}\theta}:=\Gamma_{\hspace{-.1em}\Sp}
\big|_{\gamma_\Sp\rightarrow\gamma_\theta}~~\textrm{and}~~
P^\pm_\theta:=P^\pm_\Sp\big|_{\Gamma_\Sp\rightarrow\Gamma_\theta}, 
\end{align*}
respectively.

A $(4\hspace{-.1em}\times\hspace{-.1em}4)$-matrix representation using Clifford algebra $\bm{\gamma}_\theta$ is introduced in the $\theta$-metric space.
A point ${\bm\xi}:=(x^0,x^1,x^2,x^3)\in\M_\theta$ is represented using a $(4\hspace{-.1em}\times\hspace{-.1em}4)$-matrix representation as
\begin{align*}
{\bm{q}}&:={\bm\xi}\cdot\bm{\gamma}_{\theta}=
\sum_{a=0}^3\xi^a\gamma_{\theta}^a,\\
&=
\left(
   \begin{array}{cccc}
0&0&e^{i\pi\theta/2}\hspace{.1em}x^0-x^3 & -x^1+ix^2\\ 
0&0& -x^1-ix^2&e^{i\pi\theta/2}\hspace{.1em}x^0+x^3 \\
e^{i\pi\theta/2}\hspace{.1em}x^0+x^3 & x^1-ix^2&0&0\\
 x^1+ix^2&e^{i\pi\theta/2}\hspace{.1em}x^0-x^3 &0&0\\
  \end{array}
  \right),
\end{align*}
where the chiral representation of $\bm{\gamma}_\Sp$ is used.
A generator of a \emph{rotational} transformation in the $\theta$-metric space is provided as
\begin{align*}
\Lambda^{~}_\theta:=\sum^{~}_{a<b}c^{~}_{ab}\hspace{.1em}\gamma^a_\theta\gamma^b_\theta,
\end{align*}
where $c^{~}_{ab}\in\R$ is six real-parameters of \emph{rotational angles}; determinant of $\Lambda_\theta$ normalized to unity. 
This \emph{rotation} preserves a determinant as
\begin{align*}
\mathrm{det}\left[\Lambda^{~}_\theta{\bm{q}}\Lambda_\theta^{-1}\right]&=\mathrm{det}[{\bm{q}}]
=\left(e^{i\pi\theta}t^2-x^2-y^2-z^2\right)^2
=\left(\eta^{~}_{\theta\hspace{.1em}\bcdots}\xi^\bcdot\xi^\bcdot\right)^2,
\end{align*}
and it is real valued at $\theta=0$ and $\theta=1$.

An index of the Dirac operator in the Minkowski space is defined owing to cobordism in the $\theta$-metric space:
% definition
\begin{definition}\label{EMcobor}
Suppose $\D_{\W_\theta}$ is the twisting Dirac operator in $\Gamma(\M_\theta,\W_\theta\otimes\S_\theta)$ and the spin-connection $\www$ and principal connection $\AAA^{~}_G=[\Aa^I_G]_\bcdot\eee^\bcdot\tau^{~}_I$ with structure group $G$ exist in $\W_\theta\big|_{\theta\neq0}$.
$\DSp$ has topological indices provided by a cohomology formula such as
\begin{align*}
ind(i\D_{0}):=ind(i\D_{1})=\int_{\M_1}\hat{A}(\M_1)ch(\W_1),
\end{align*}
where $\W_1:=\left.\W_\theta\right|_{\theta\rightarrow1-0}$.
\end{definition}
% Theorem
\begin{theorem}
\textbf{Definition \ref{EMcobor}} is well-defined.
\end{theorem}
\begin{proof}
Owing to the assumption, $\M_\theta$ is a four-dimensional compact and oriented manifold with no boundary, and $\M_1$ is Eudlidean manifold.
Differential operator $i\D_1$ is the Hermitian anti-self-conjugate Dirac operator concerning Clifford algebra $\bm{\gamma}^{~}_E$ given above.
A Laplacian concerning ${\D}_{\W_\theta}$ is provided such that:
\begin{align*}
\Delta_\theta:={\D}_{\W_\theta}{\D}_{\W_\theta}&=
\gamma^a_\theta\gamma^b_\theta\nabla_a\nabla_b
=\left(2\eta_\theta^{ab}-\gamma^b_\theta\gamma^a_\theta\right)\nabla_a\nabla_b,\\
\Longrightarrow
{\D}_{\W_\theta}{\D}_{\W_\theta}&=
e^{i\pi\theta}(\nabla_0)^2-(\nabla_1)^2-(\nabla_2)^2-(\nabla_3)^2,
\intertext{where }
\nabla_a&:=\partial_a-i\frac{1}{2}\cG
\eta_{\theta\hspace{.1em}\bcdots}
\eta_{\theta\hspace{.1em}\stars}\hspace{.1em}
\Varepsilon^\mu_a\omega^{~\bcdot\star}_{\mu}\hspace{.1em}{\Ss}^{\bcdot\star}_{\hspace{-.1em}\theta}
-i{c^{~}_G}\hspace{.1em}
\left[\Aa^I_G\right]_a\tau^{~}_I,
\intertext{
where $\Ss^\bullets_\theta$ is a generator of a $Spin(1,3)$ group in the $\theta$-metric space and represented as
}
\Ss^{ab}_\theta&:=i[\gamma_\theta^a,\gamma_\theta^b]/4 
\end{align*}
using the trivial basis.
Thus, a Laplacian is elliptic operator when $\theta=1$ and Atiyah--Singer index is well-defined in $\M_1$.
On the other hand, $\M_0$ and $\M_1$ are cobordant such that $\M_0\sim_{[0,1]\otimes\M_\theta}\M_1$ and connections exists in $\W_\theta$ with $\theta\in(0,1]$; thus, $\M_{+0}$ and $\M_1$ are homotopically equivalent.
Therefore, Pontrjagin and Chern classes are preserved owing to \textbf{Remark \ref{prg}}.
The Atiyah--Singer index in hyperbolic space is defined as a corresponding index in $\M_1$; $ind(i\D_0):=ind(i\D_1)$.
Therefore, the Atiyah--Singer index in $\TsM_1$ is well-defined.
\end{proof}
\noindent
%%%
The same result is obtained owing to cobordism $\M_{-1}\sim_{W'_\theta}\M_0$ where $W'_\theta:=[-1,0]\otimes\M_\theta$ and gives $ind(i\D_{0}):=ind(i\D_{-1})=ind(i\D_{1})$.
When $\gamma_E^0=\gamma_\theta^0|_{\theta\rightarrow1+0}$ is replaced by $-\gamma_E^0=\gamma_\theta^0|_{\theta\rightarrow1-0}$ in Eudlidean space-time, the Clifford algebra (\ref{CliffordTheta}) does not change.

%%% ----------------------------------------------------------------------
% Section 2 : Principal bundle
%%% ----------------------------------------------------------------------
\section{Principal bundles}
This section introduces several principal bundles necessary to construct the Yang--Mills theory.
Structure groups, namely the gauge-group, of principal bundles govern phenomena induced by a connection and a curvature of the Yang--Mills theory through an equation of motion.
The theory of principal bundles with the structural gauge group provides the Yang--Mills theory's backbone.
Physically, matter fields are spinor sections of the principal gauge bundle, and forces, including a gravitational force, are represented owing to connections and curvatures.

This section discusses principal bundles mainly with the Lorentzian metric for simplicity; the determinant of the metric is explicitly indicated if necessary.
Extension to the $\theta$-metric is realized by replacing $\bm\eta\mapsto\bm\eta_\theta$,  $\bm\epsilon\mapsto\bm\epsilon_\theta$ and $\sigmaeta\mapsto\sigmath$ according to (\ref{kappag}).
We note that the metric tensor and its inverse are hidden in the index-raising and -lowering operations.

%
% principal bundles of space-time manifold
%
\subsection{Principal bundles of space-time manifold}
A principal bundle is a tuple such that $(E,\pi,M,G)$, where $E$ and $M$ are total and base spaces, respectively, and $\pi:E\rightarrow M$ is a projection (bundle) map, and $G$ is an associating structure group.
A group operator acts simply transitively  from the right\footnote{A group operator of the co-Poincar\'{e} group is defined exceptionally acting from the left.}.
A duplex superspace appears in the Yang--Mills theory and general relativity in four-dimensions as introduced in section \ref{242}.

%
% co-Poincare bundle
%
\subsubsection{Co-Poincar\'{e} bundle}
The Yang--Mills theory has the Poincar\'{e} symmetry in addition to the gauge-symmetry.
On the other hand, general relativity is not invariant under the four-dimensional translation.
The author introduces a modified translation operator to preserve in general relativity\cite{doi:10.1063/1.4990708,Kurihara_2020}.
This section discusses a principal bundle in the space-time manifold with the structure co-Poincar\'{e} group.

We extend a structure group of the inertial bundle to Poincar\'{e} group:
\begin{align*}
I\hspace{-.1em}S\hspace{-.1em}O(1,3)=SO(1,3)\ltimes T^4.
\end{align*}
Representations of Lie algebra $\iii\sss\ooo(1,3)$ of Poincar\'{e} group is obtained using the trivial basis as follows\cite{fre2012gravity}:
\begin{align}
\left[P_{a},P_{b}\right]&=0,\nonumber\\
\left[J_{ab},P_{c}\right]&=-\eta_{ac}P_{b}+\eta_{bc}P_{a},\label{JJ}\\
\left[J_{ab},J_{cd}\right]&=
-\eta_{ac}J_{bd}+\eta_{bc}J_{ad}
-\eta_{bd}J_{ac}+\eta_{ad}J_{bc}.\nonumber
\end{align}
where $P_a$ and $J_{ab}$ are generators of the $T^4$ group and the $SO(1,3)$ group, respectively.

The Einstein--Hilbert gravitational Lagrangian does not respect the Poinca\'{e} symmetry\cite{0264-9381-29-13-133001}.
We introduce the co-Poincar\'{e} symmetry extending the Poinca\'{e} symmetry here.
The co-Poincar\'{e} symmetry, which is denoted as $G_\cP$, is the symmetry in which the translation operator  is replaced by the co-translation operator.
A generator of the co-translation is defined as $P_{ab}:=P_a\iota_b$, where $\iota_b$ is a contraction with respect to trivial frame field $\partial_b$.
Lie algebra of the co-Poincar\'{e} group are provided as
\begin{align*}
\left[P_{ab},P_{cd}\right]&=0,\\
\left[J_{ab},P_{cd}\right]&=-\eta_{ac}P_{bd}+\eta_{bc}P_{ad},
\end{align*}
and $\left[J_{ab},J_{cd}\right]$ in (\ref{JJ}).
The structure constants of the co-Poincar\'{e} group can be obtained from above Lie algebra through a relation, $[\Theta_I,\Theta_J]:=\FF_{~IJ}^{K}\hspace{.1em}\Theta_K$, where $\left[\Theta_I\right]_{ab;cd}:=\left(\left[\Theta_{1}\right]_{ab},\left[\Theta_{2}\right]_{cd}\right)=\left(P_{ab},J_{cd}\right)$.
Lie algebra of the co-Poincar\'{e} symmetry is denoted as $\ggg_\cP$. 
Each component of the structure constant is provided by direct calculation using the trivial basis as
\begin{align}
\left\{
\begin{array}{l}
\FF_{~11}^{1}=\FF_{~11}^{2}=\FF_{~12}^{2}=
\FF_{~21}^{2}=\FF_{~22}^{1}=0,\\
\left[\FF_{~12}^{1}\right]_{ab;cd}^{ef}=
-\left[\FF_{~21}^{1}\right]_{cd;ab}^{ef}=
\eta_{ac}\delta^e_b\delta^f_d-\eta_{bc}\delta^e_a\delta^f_d,\\
\left[\FF_{~22}^{2}\right]_{ab;cd}^{ef}=
-\eta_{ac}\delta_b^e\delta_d^f+\eta_{bc}\delta_a^e\delta_d^f
-\eta_{bd}\delta_a^e\delta_c^f+\eta_{ad}\delta_b^e\delta_c^f.
\end{array}
\right.\label{FgrLie}
\end{align}

Connection form $\AAA_\cP$ and curvature form $\FFF_\cP$ are, respectively, introduced as Lie-algebra valued one- and two-forms concerning co-Poincar\'{e} group, and they are expressed using the trivial basis as
\begin{align*}
\AAA_\cP&=J_\bcdots\otimes\www^\bcdots+P_\bcdots\otimes\hspace{1.2em}\SSS^\bcdots
&\in\ggg_\cP\otimes\Wedge^1(\TsM),\\
\FFF_\cP&=
J_\bcdots\otimes\RRR^\bcdots+P_\bcdots\otimes d_\www\SSS^\bcdots
&\in\ggg_\cP\otimes\Wedge^2(\TsM).
\end{align*}
We note that $P^\bcdots\otimes d_\www\SSS_\bcdots\in\Omega^2$ due to $\iota_\bullet:\Omega^p\rightarrow\Omega^{p-1}$, e.g., $\iota_a\eee^b=\delta^b_a$.
A principal co-Poincar\'{e} bundle is defined as follows:
% Definition
\begin{definition}[principal co-Poincar\'{e} bundle]\label{CoPoinBndl}
A principal co-Poincar\'{e} bundle is a tuple $\left(\M,\pi_{\mathrm I},\MM,G_\cP\right)$,
where a co-Poincar\'{e} group operator acts on the the total space from the left.
\end{definition}
%%%

%
% superspace in general relativity
%
\subsubsection{Duplex superspace in general relativity}\label{CPbundle}
The Hodge-dual operator works as the parity operator and induces a duplex superspace in $\Omega^2(\TsM)$ over the four-dimensional manifold as shown in section \ref{242}.
The superspace consisting of curvature two-forms of the inertial bundle is crucial in physics because it is related to the structure of the space-time itself.
Various representations of general relativity in four-dimensions are summarized in Ref.\cite{krasnov2020formulations}.
\begin{definition}[Hodge-dual Lorentz curvature]
A Hodge-dual of the curvature of the inertial bundle is defined owing to the Hodge-dual operator  as
% Definition
\begin{align*}
\HD:\RRR^{ab}\mapsto\hat\RRR^{ab}
:=\frac{1}{2}\hat{R}^{ab}_{\hspace{.8em}\bcdots}
\hspace{.1em}\eee^\bcdot\wedge\eee^\bcdot,
\end{align*}
where
\begin{align*}
\hat{R}^{ab}_{\hspace{.8em}cd}&:=\frac{\sqrt{\sigmaeta}}{2}\hspace{.1em}{R}^{ab}_{\hspace{.8em}\stars}
\hspace{.1em}\eta^{\bcdot\star}\eta^{\bcdot\star}\hspace{.1em}\epsilon_{\bcdots cd}.
\end{align*}
\end{definition}
\noindent
%%%
The Hodge-dual operator squared yields eigenvalue equation $\HD\bcdot\HD(\RRR^{ab})=\RRR^{ab}$; thus, the operator has eigenvalues of $\pm1$.
Owing to \textbf{Remark \ref{vpvm}}, a space of the two-form objects is split into two subspaces as 
\begin{align*}
\RRR^\pm:=P_\textrm{H}^\pm(\RRR)\in V^\pm ~~\textrm{yielding}~~ 
\HD(\RRR^\pm)=\pm\RRR^\pm.
\end{align*}
Surface form $\SSS\in V^2(\TM)\otimes\Omega^2(\TsM)\otimes\sss\ooo(1,3)$ is expressed using tensor coefficients as 
\begin{align*}
\SSS^{ab}=\frac{1}{2}S^{ab}_{\hspace{.7em}\bcdots}\hspace{.2em}
\eee^\bcdot\wedge\eee^\bcdot,~~\textrm{where}~~
S^{ab}_{\hspace{.7em}cd}=\epsilon^{ab}_{\hspace{.7em}cd}.
\end{align*}
Any two-form objects are represented using the parity bases as shown in (\ref{SISJ}); thus, the surface form can be expressed using a sum of two eigenstates of the Hodge-dual operator as $\SSS=\SSS^++\SSS^-$.

Curvature and surface forms in the inertial bundle make the Kramers pairs $\KP(\RRR)$ and $\KP(\SSS)$.
For the Kramers pair of the spin-connection $\KP(\www)$, $\hat\www$ is provided as a solution of the structure equation:
\begin{align*}
\hat\RRR&=
d\hat\www_\Sp+\cG\hspace{.1em}\hat\www\wedge\hat\www,
\end{align*}
for given $\hat\RRR$.

We construct the secondary superspace in the inertial bundle.
The covariant differential in the Lorentzian manifold induces the Dirac operator such that:
\begin{align*}
\D^{~}_\www&:=\gamma\d^{~}_\www
=\left(
\begin{array}{cc}
	0      &  d+{\cGo}[\hat\www,]_\wedge\\
 	d+{\cGz}[\www,]_\wedge &   0
\end{array}
\right),
\end{align*}
where $\d^{~}_\www:=$diag$(d^{~}_{\www},\hat{d}^{~}_{\www})$.
The Lorentz curvature is represented in the secondary superspace as
\begin{align*}
\bm\RRR^{~}_2&:=\left(
\begin{array}{cc}
	\RRR &   0\\
    0    & \hat\RRR
\end{array}
\right)=\left(
\begin{array}{cc}
 	d\left[\www\right]^{ab}+{\cGz}
\left[\www\right]^a_{~\bcdot}\wedge
\left[\www\right]^{\bcdot b} &   0\\
	0      &  d\left[\hat\www\right]^{ab}+{\cGo}
\left[\hat\www\right]^a_{~\bcdot}\wedge
\left[\hat\www\right]^{\bcdot b}
\end{array}
\right).
\end{align*}
The second Bianchi identity is represented as
\begin{align*}
\D^{~}_\www\bm\RRR^{~}_2&=\left(
\begin{array}{cc}
	0      &  d\hat\RRR+{\cGo}[\hat\www,\hat\RRR]_\wedge\\
 	d\RRR+{\cGz}[\www,\RRR]_\wedge &   0
\end{array}
\right)=\bm{0}
\end{align*}
in the secondary superspace.

Dual connection and curvature are also provided in the co-Poincar\'{e} bundle as
\begin{align*}
\hat\AAA^{~}_\cP&=J_\bcdots\otimes\hat\www^\bcdots+P_\bcdots\otimes\hspace{1.2em}\hat\SSS^\bcdots
&\in\hspace{.2em}\ggg^{~}_\cP\otimes\Wedge^1(\TsM),\\
\hat\FFF^{~}_\cP&=
J_\bcdots\otimes\hat\RRR^\bcdots+P_\bcdots\otimes 
\hat{d}^{~}_{\www}\hat\SSS^\bcdots
&\in\hspace{.2em}\ggg^{~}_\cP\otimes\Wedge^2(\TsM).
\end{align*}
The gravitational form is introduced in the secondary superspace as follows:
% Definition
\begin{definition}[gravitational form]\label{Gform}
The gravitational form is defined owing to the co-Poincar\'{e} curvature in the secondary superspace such that:
\begin{align*}
\bm\FFF^{~}_{Gr}&:=
\gamma\left(
\begin{array}{cc}
	\FFF^{~}_\cP & 0\\
 	0 & \hat\FFF^{~}_\cP
\end{array}
\right)=
\left(
\begin{array}{cc}
 	0 & \hat\FFF^{~}_\cP \\
	\FFF^{~}_\cP & 0
\end{array}
\right).
\end{align*}
\end{definition}
\noindent
%%%
The Einstein--Hilbert gravitational Lagrangian is represented using the gravitational form. 

%
% principal bundles in Yang--Mills theory
%
\subsection{Principal bundles in Yang--Mills theory}
The Yang--Mills theory is constructed on the inertial manifold.
A spinor section describing a physical matter field is introduced in the spin manifold over $\M$.
The matter field is an element of a compact Lie group, namely a gauge group.
A connection and curvature of the principal bundle are also physical objects which mediate a \emph{force} among matter fields.

Two structure groups, $Spin(M)$ and $SU(N)$, are introduced; the former concerns symmetry in the Lorentzian manifold, the latter is an internal symmetry group and governing property of a gauge interaction ({gauge force}).
These bundles induce a superspace, and they are twisted into each other.
This section introduces principal bundles, their connections and curvatures appearing in the Yang--Mills theory.

%
% spinor bundle
%
\subsubsection{Spinor bundle}\label{spinerbundle}
A spin structure is introduced in $\M$ through spinor group $Spin(V)$, where $V$ is a vector space over $\M$.
$Spin(V)$ is also denoted as $Spin(1,3)$ when $V$ has the $SO(1,3)$ symmetry.
It is known that a $Spin(1,3)$ group doubly covers a $SO(1,3)$ group such that $Spin(1,3)\cong SO(1,3)\times_{\Z_2}U(1)$ under covering map $\tau_\cov:Spin(1,3)\rightarrow SO(1,3)\otimes\{\pm1\}$.
A Clifford algebra of $Spin(1,3)$ is denoted as $\Cl(1,3)$.
When a spin structure is defined globally in $\M$, it is denoted $Spin(\M)$, and the inertial manifold is extended to a spin manifold.
% definition
\begin{definition}[Principal spinor bundle]
The principal spinor bundle is a tuple as follows: 
\begin{align*}
\left(\S_\Sp,\pi_\Sp,\M, Spin(1,3)\right).
\end{align*}
A total space is Clifford module $\S_\Sp=\Cl(1,3)\otimes\C$.
Projection map $\pi_\Sp$ is provided owing to the covering map as follows:
\begin{align*}
\pi_\Sp:=\tau_\cov/\{\pm1\}:
\S_\Sp\otimes \M\rightarrow\M/\{\pm1\}:\left.\psi\right|_p\hspace{.1em}\mapsto p\in\M,
\end{align*}
where $\psi\in\Gamma\left(\M,\S_\Sp\right)$.
The structure group of the inertial bundle, $SO(1,3)$, is lifted to $Spin(1,3)$ owing to projection map $\pi_\Sp$; a spin structure is induced globally in $\M$.
\end{definition}
\noindent
%%%
We assume the existence of a spin structure globally; a spinor is a section in $\S_\Sp$.
The Lorentz transformation of a spinor is provided as follows:
\begin{align*}
\Gso:\S_\Sp\otimes\M\rightarrow \S_\Sp:
\psi(x)\mapsto S(\LLambda)\psi(x).
\end{align*}
Vector $\bm{\gamma}_\Sp=\gamma_\Sp^a\partial_a\in\TM\otimes\S_\Sp$ is a Clifford algebra valued vector in $\TM$ transformed as an adjoint representation under the Lorentz transformation such that:
\begin{align*}
\Gso(\bm{\gamma}_\Sp)&:=
S(\LLambda)\bm{\gamma}_\Sp S(\LLambda)^{-1}=\bm{\gamma}_\Sp\LLambda^{-1}.
\end{align*}
A spinor product of spinor and Clifford algebra are compatible with the Lorentz transformation as follows:
\begin{align*} 
\Gso:\bm{\gamma}_\Sp\psi(p)\mapsto
\left(S(\LLambda)\bm{\gamma}_\Sp S(\LLambda)^{-1}\right)
\left(S(\LLambda)\psi(\Lambda p)\right)
=S(\LLambda)\left(\bm{\gamma}_\Sp\psi(\Lambda p)\right)\cong\bm{\gamma}_\Sp\psi(p).
\end{align*}
Vector space $V_\psi$ has $\Z_2$-grading structure $V_\psi=V_\psi^+\oplus V_\psi^-$, where 
\begin{align*}
V_\psi^\pm&:=\left\{\psi^\pm\in V_\psi | \GammaSp\psi^\pm=\pm\psi^\pm\right\}.
\end{align*}
Any spinors belong to one of two chiral sets $V_\psi^\pm$ according to an eigenvalue of the chiral operator.
Chiral projection operator $P_\Sp^\pm:V_\psi\rightarrow V_\psi$ acts on a spinor as $P_\Sp^\pm\psi=\psi^\pm\in V_\psi^\pm$.

A representation of Clifford module $\S_\Sp(\Cl(1,3)\otimes\C)$ can be expressed using a direct sum of four isomorphic representations.
Suppose $\Cl_0$ is a subspace spanned by Clifford products of an even number of elements.
In this case, a representation of $\Cl_0$ is provided using a direct sum of two irreducible representations, $\S_\Sp(\Cl_0(1,3)\otimes\C)\cong\S_\Sp^+(\Cl_0(1,3)\otimes\C)\oplus\S_\Sp^-(\Cl_0(1,3)\otimes\C)$.
Two representations $\S_\Sp^\pm(\Cl_0(1,3)\otimes\C)$ are not isomorphic to each other.
Clifford module $\S_\Sp^\pm(\Cl_0(1,3)\otimes\C)$ can be defined globally in $\M$ owing to the covering map, which is simply denoted by $\S_\Sp^\pm(\M)$.

A connection form of the spinor bundle is obtained by equating it to the spin-connection form.
Moreover, a spinor is lifted to a section in $\TM$ owing to the coverage map.
Covariant differential $d_\Sp$ concerning $Spin(1,3)$ with connection $\AAA_\Sp$, namely a Clifford connection, acts on the spinor as\cite{fre2012gravity, moore1996lectures,Nakajima:2015rhw}
\begin{align}
d_\Sp\psi&:=d\psi+\cG\AAA_\Sp\psi,
\intertext{where}
\AAA_\Sp&:=-i\hspace{.1em}\www\cdot\Ss\hspace{.1em}
=-\frac{i}{2}\hspace{.1em}\eta_\bcdots\eta_\stars
\www^{\bcdot\star}\Ss^{\bcdot\star},~~\textrm{and}~~\Ss:=\Ss_\theta\big|_{\theta\rightarrow+0}.\label{Asp}
\end{align}
A coupling constant in the spinor bundle is set to the same as that of the inertial bundle.
Corresponding curvature two-form $\FFF_\Sp$ is defined as
\begin{align}
\FFF_\Sp&:=d\AAA_\Sp+\cG\hspace{.1em}\AAA_\Sp\wedge\AAA_\Sp.\label{Fsp}
\end{align}
The Bianchi identities are provided as
\begin{align}
\left(d_\Sp\right)^2=\cG\hspace{.1em}\FFF_\Sp=i\cG\hspace{.1em}\RRR\cdot\Ss~\mathrm{and}~~d_\Sp\FFF_\Sp=0,\label{spcurv}
\end{align}
where $\RRR\cdot{\Ss}=\RRR^\bcdots{\Ss}_\bcdots/2$ using the trivial basis.
Dirac operator 
\[
\dSp:\Gamma\left(\M,\Wedge^0(\S_\Sp)\right)\rightarrow
\Gamma\left(\M,\Wedge^0(\S_\Sp)\right)
\]
is defined as
\begin{subequations}
\begin{align}
\dSp(\psi)&:=\iota_{{\gamma}_\Sp}d_\Sp:
\Gamma\left(\M,\Wedge^0(\S_\Sp)\right)
\longrightarrow^{\hspace{-1.3em}^{d_\Sp}}\hspace{.3em}
\Gamma\left(\M,\Wedge^1(\S_\Sp)\right)
\longrightarrow^{\hspace{-1.5em}^{\iota_{{\gamma}_\Sp}}}\hspace{.4em}
\Gamma\left(\M,\Wedge^0(\S_\Sp)\right),\label{DSG}
\intertext{
which is expressed using the trivial basis as
}
\dSp(\psi)&=\left(\gamma_\Sp^\bcdot\partial_\bcdot
-\frac{i}{2}\cG\hspace{.1em}\gamma_\Sp^\bcdot
\Varepsilon_\bcdot^\mu\omega_\mu^{~\stars}{\Ss}_\stars
\right)\psi.
\end{align}
\end{subequations}
We note that
$
\iota_{{\gamma}_\Sp}d\psi=
(\iota_{{\gamma}_\Sp}\eee^\bcdot)(\partial_\bcdot\psi)
=\gamma_\Sp^\bcdot\partial_\bcdot\psi.
$
A Lorentz transformation of the Dirac operator is
\begin{align*}
\Gso\hspace{-.2em}\left(\dSp\right)&=
\left(S(\LLambda)\bm{\gamma}_\Sp S(\LLambda)^{-1}\right)
\left(S(\LLambda) d_\Sp S(\LLambda)^{-1}\right)
=S(\LLambda)\dSp S(\LLambda)^{-1}.
\end{align*}

For the Dirac spinor, the Dirac conjugate is defined as a map
\begin{align*}
\tau_{\Dc}:\Gamma\left(\M,\Wedge^0(\S_\Sp)\right)
\rightarrow\Gamma\left(\M,\Wedge^0(\S_\Sp)\right):
\psi\mapsto\bar\psi:=\psi^\dagger\gamma_\Sp^0.
\end{align*}
Hereafter, $\Gamma\left(\M,\Wedge^0(\S^{~}_\Sp)\right)$ and $\Gamma\left(\M,\Wedge^0(\S^\pm_\Sp)\right)$ are simply denoted as $\S^{~}_\Sp$ and $\S^\pm_\Sp$, respectively.
Chiral operator $\GammaSp$ acts on the Dirac conjugate spinor as $\bar\psi^\pm\GammaSp=\mp\bar\psi^\pm$ because
\begin{align*}
(\psi^{\pm})^\dagger=\left(\GammaSp\psi^\pm\right)^\dagger=
i\hspace{.1em}(\pm\psi^{\pm})^\dagger\gamma_\Sp^3\gamma_\Sp^2\gamma_\Sp^1\gamma_\Sp^0
=\mp i\hspace{.1em}\bar\psi^\pm\gamma_\Sp^3\gamma_\Sp^2\gamma_\Sp^1;%\label{barGamma1}
\end{align*}
yielding
\begin{align}
(\psi^{\pm})^\dagger\gamma_\Sp^0=
\left(\mp i\hspace{.1em}\bar\psi^\pm\gamma_\Sp^3\gamma_\Sp^2\gamma_\Sp^1\right)\gamma_\Sp^0
&\Rightarrow
\bar\psi^\pm=\mp\bar\psi^\pm\GammaSp.\label{barGamma}
\end{align}
%Here, (\ref{cf2}) is used.
Therefore, we obtain that $\bar\psi\in\S_\Sp^+$ for $\psi\in\S_\Sp^-$.
Terms $i\hspace{.1em}\bar\psi(\iota_{{\gamma}_\Sp}d)\psi$ and $\bar\psi\psi$ are $SO(1,3)$ invariant as follows: 
\begin{align*}
\Gso\hspace{-.2em}\left(i\hspace{.1em}\bar\psi(\iota_{{\gamma}_\Sp}d)\psi\right)
&=i\hspace{.1em}\psi^\dagger S(\LLambda)^{-1}
\left(S(\LLambda)\gamma_\Sp^0 S(\LLambda)_{~}^{-1}\right)
\left(S(\LLambda)(\iota_{{\gamma}_\Sp}d)S(\LLambda)^{-1}\right)
S(\LLambda)\psi,\\
&=i\hspace{.1em}\bar\psi(\iota_{{\gamma}_\Sp}d)\psi,
\end{align*}
where $\Gso\hspace{-.2em}\left(\iota_{{\gamma}_\Sp}d\right)=S(\LLambda)(\iota_{{\gamma}_\Sp}d)S(\LLambda)^{-1}$ is used.
Similar calculations provide the $SO(1,3)$ invariance for $\bar\psi\psi$.

% Remark
\begin{remark}\label{lemDOpm}
There exists super algebra in $End\left(\S_\Sp\right)=End^+\left(\S_\Sp\right)\oplus End^-\left(\S_\Sp\right)$.
\end{remark}
\begin{proof}
Suppose operator $\mathcal{O}_0$ ($\mathcal{O}_1$) includes an even (odd) number of $\gamma_\Sp$\hspace{-.15em}'s.
For $\psi^\pm\in\S_\Sp^\pm$, 
\begin{align*}
\GammaSp\psi^\pm=\pm\psi^\pm, ~~&\textrm{and}~~
\left\{
\begin{array}{rrl}
\GammaSp\left(\mathcal{O}_0\psi^\pm\right)
=&\mathcal{O}_0\GammaSp\psi^\pm
=&\pm\left(\mathcal{O}_0\psi^\pm\right),\\
\GammaSp\left(\mathcal{O}_1\psi^\pm\right)
=&-\mathcal{O}_1\GammaSp\psi^\pm
=&\mp\left(\mathcal{O}_1\psi^\pm\right)
\end{array}
\right.
\end{align*}
owing to (\ref{g5ga}). 
Thus, we obtain $\mathcal{O}_0\in{End}^+\left(\S_\Sp\right)$ and $\mathcal{O}_1\in{End}^-\left(\S_\Sp\right)$.
Moreover, it is obvious that
\begin{align*}
\mathcal{O}_0\mathcal{O}_0, \mathcal{O}_1\mathcal{O}_1\in{End}^+\left(\S_\Sp\right)
~~&\textrm{and}~~ 
\mathcal{O}_0\mathcal{O}_1, \mathcal{O}_1\mathcal{O}_0\in{End}^-\left(\S_\Sp\right).
\end{align*}
Therefore, the remark is maintained. 
\end{proof}
%%%
The Dirac operator $\dSp$ belongs to $\mathcal{O}_1$ and flips a parity of $\psi^\pm$.
We note that
\begin{align*}
\begin{array}{ccccc}
\bar\psi\left(i\hspace{.1em}\dSp\right)\psi&=&
\left(\bar\psi\GammaSp\right)
\left(\GammaSp\left(i\hspace{.1em}\dSp\right)\psi\right)&=&
\bar\psi^+\left(i\hspace{.1em}\dSp\right)\psi^+ +
\bar\psi^-\left(i\hspace{.1em}\dSp\right)\psi^-,\\
\bar\psi\psi&=&\left(\bar\psi\GammaSp\right)
\left(\GammaSp\psi\right)&=&
\bar\psi^+\psi^- + \bar\psi^-\psi^+.
\end{array}
\end{align*}

%
% gauge bundle
%
\subsubsection{Gauge bundle}
This section treats a principal bundle with a compact Lie group, namely the principal gauge bundle,  in the inertial manifold.
We exploit the $SU(N)$ group (including $U(1)$) as a structure group, that is called  the gauge group.
A scalar function (field) appearing in the gauge bundle as a section is called the Higgs field in physics.

% Definition
\begin{definition}[Principal gauge bundle]
A principal gauge bundle is defined as a tuple such that:
\begin{align*}
\left(\TsM^{\otimes N},\pi_{\SU},\M,SO(1,3)\otimes{SU(N)}\right),~&\textrm{where}~
\TsM^{\otimes N}:=\overbrace{\TsM\otimes\cdots\otimes\TsM}^{N}
\end{align*}
A space of complex-valued scalar fields, namely the \emph{Higgs field}, $\bm\phi=(\phi^1,\phi^2,\cdots,\phi^N)^T$ is introduced as a section $\bm\phi\in\bm\Phi=\Gamma\left(\M,\Omega^0(\TsM^{\otimes N})\otimes{SU(N)}\right)$ belonging to the fundamental representation of the $SU(N)$ symmetry.
\end{definition}
\noindent
%%%
$SU(N)$ group operator $\Gsu$  acts on the section as
\begin{align*}
\Gsu:\bm\Phi\rightarrow\bm\Phi:
\bm\phi\mapsto\Gsu\hspace{-.2em}\left(\bm\phi\right):=
\bm{g}_\SU^{~}\hspace{.1em}\bm\phi,
\end{align*}
that has an $(N\hspace{-.2em}\times\hspace{-.2em}N)$-matrix representation such as
\begin{align*}
\left[\Gsu\left(\bm\phi\right)\right]^I=
\left[\bm{g}_\SU^{~}\right]^I_{J}\phi^J,
\end{align*}
where $\bm{g}_\SU^{~}$ is a unitary matrix with det$[\bm{g}_\SU^{~}]=1$.
Lie algebra of $SU(N)$ group is 
\begin{align*}
\left[\tau^{~}_I,\tau^{~}_J\right]
:=i\sum_K\hspace{.1em}f^{~}_{IJK}\hspace{.2em}\tau^{~}_K,
\end{align*}
where $f_{\bullets\bullet}$ is a structure constant of $SU(N)$, and $\bm{\tau}=(\tau_1,\tau_2,\cdots,\tau_{N^2-1})\in\sss\uuu(N)$.
A gauge connection $\AAA^{~}_{\SU}:=\AAA_\SU^I\hspace{.2em}\tau^{~}_I\in\Omega^1\otimes\sss\uuu(N)$ is a Lie-algebra valued one-form object.
Covariant differential $d^{~}_\SUO$ on $p$-form object $\aaa\in\Omega^p(\TsM)$ concerning $SU(N)$ and $SO(3,1)$ is defined as
\begin{align}% Peskin--Schroeder p.487 eq. (15.24)
d^{~}_\SUO\hspace{.2em}\aaa\hspace{.1em}:=
\bm{1}_\SU\hspace{.1em}d_\www\aaa-i\hspace{.1em}{{c^{~}_\SU}}\hspace{.1em}
[\AAA^{~}_\SU,\aaa]_\wedge.\label{cdSU}
\end{align}
Real constant ${c^{~}_\SU}\in\R$ is a coupling constant of the gauge interaction in physics.
Connection $\AAA^{~}_\SU$ belongs to an adjoint representation of the gauge group as
\begin{align}
\Gsu\hspace{-.2em}\left(\AAA^{~}_\SU\right)&=
\bm{g}^{-1}_\SU\hspace{.2em}\AAA^{~}_\SU\hspace{.2em}\bm{g}_\SU^{~}
+i\hspace{.1em}{c^{-1}_\SU}\gdg,\label{pisu}
\end{align}
which ensures $SU(N)$ covariance of the covariant differential.
At the same time, $\AAA^{~}_\SU$ is a vector in $\TsM$, which is $SO(1,3)$-transformed as
\begin{align}
\Gso:\AAA^{~}_\SU\mapsto\Gso(\AAA^{~}_\SU)&=
\LLambda\AAA^{~}_\SU.\label{SUcd}
\end{align}
Corresponding gauge curvature two-form $\FFF_\SU$ is defined through a structure equation such that:
\begin{align}
\FFF_\SU=\FFF_\SU^I\hspace{.2em}\tau^{~}_I&:=d_\www\AAA^{~}_\SU
-i\hspace{.1em}{{c^{~}_\SU}}\hspace{.1em}\AAA^{~}_\SU\wedge\AAA^{~}_\SU,\nonumber \\
&=\left(d\AAA_\SU^I+\cG\hspace{.1em}\www\wedge\AAA_\SU^I
+\frac{{c^{~}_\SU}}{2}f^I_{~JK}\hspace{.1em}\AAA_\SU^J\wedge\AAA_\SU^K
\right)\tau^{~}_I,\label{SUCRV}
\end{align}
where $f^I_{~JK}=f^{~}_{IJK}$.
Curvature $\FFF^{~}_\SU$ is $SO(1,3)$ covariant  and is $SU(N)$-transformed as an adjoint representation such that:
\begin{align*}
\Gsu\hspace{-.2em}\left(\FFF^I_{\SU}\right)&=\bm{g}^{~}_\SU\hspace{.1em}\FFF^I_{\SU}\hspace{.1em}\bm{g}_\SU^{-1}+i\hspace{.1em}{c^{-1}_\SU}\gdg.
\end{align*}
The first and second Bianchi identities are
\begin{align}
d_{\SUO}^2\hspace{.2em}\aaa^a=
\cG\eta_{\bcdots}\RRR^{a\bcdot}\wedge\aaa^\bcdot
-i{c^{~}_\SU}\hspace{.1em}\FFF_{\SU}\wedge\aaa^a,
~~\textrm{and}~~
d^{~}_{\SUO}\hspace{.2em}\FFF_{\SU}=0,\label{BianchiSU}
\end{align}
where $\aaa\in\Omega^p(\TsM)$.
The gauge connection and curvature are represented using the trivial bases in $\TsM$, respectively, as
\begin{align}
\AAA^{~}_\SU&=\AAA_\SU^I\hspace{.2em}\tau^{~}_I=:\Aa^I_{a}\eee^a\hspace{.2em}\tau^{~}_I,~~\textrm{and}~~
\FFF_\SU=\FFF_\SU^I\hspace{.2em}\tau^{~}_I=:\frac{1}{2}\f^I_{ab}\hspace{.1em}\eee^a\hspace{-.1em}\wedge\eee^b\hspace{.2em}\tau^{~}_I;\label{Fan}
\end{align}
and the second Bianchi identity is represented as
\begin{align*}
\epsilon^{a\bcdots\bcdot}\left(
\partial_\bcdot\f^I_\bcdots+{c^{~}_\SU}\hspace{.1em}f^I_{~JK}\Aa^J_\bcdot\f^K_\bcdots
\right)&=0.
\end{align*}
In this expression, tensor coefficients of the gauge curvature is provided using those of the gauge connection such that: 
\begin{align}
\f^I_{ab}&=\partial_a\Aa^I_b-\partial_b\Aa^I_a+
{{c^{~}_\SU}}\hspace{.1em}f^I_{~JK}\Aa^J_a\Aa^K_b+\Aa^I_\bcdot\TT^\bcdot_{~ab},\label{FabISU}
\end{align}
where $\TT$ is a tensor coefficient of the torsion form defined as  $\TTT^\bullet=:\TT^\bullet_{\hspace{.5em}\bcdots}\eee^\bcdot\wedge\eee^\bcdot/2$.
When the space-time manifold is torsion-less, the gauge curvature has the same representation as it in the flat space-time.

%
% spinor-gauge bundle
%
\subsubsection{Spinor-gauge bundle}\label{sgb}
The gauge group introduced in the previous section also acts on the spinor section (field), which also has the $SU(N)$ symmetry.
The spinor field creates matter particles from the quantum mechanical point of view; the gauge curvature acts on them as forces among matters, and the gauge connection mediates forces between the spinor fields.
A spinor-gauge bundle is a Whitney sum of Spinor- and gauge bundles.

% Definition
\begin{definition}[Principal spinor-gauge bundle]\label{spinor-gauge-bundle}
A principal spinor-gauge bundle is a tuple such as
\begin{align*}
\left(\S_\Sp^{\otimes N},\pi_{\Sp}\oplus\pi_{\SU},\M, Spin(1,3)\otimes{SU(N)}\right)
~&\textrm{where}~
\S_\Sp^{\otimes N}:=\overbrace{\S_\Sp\otimes\cdots\otimes\S_\Sp}^{N}.
\end{align*}
The total space of the gauge bundle is lifted to the spin manifold.
\end{definition}
\noindent
%%%
A connection and a curvature are provided, respectively, as
\begin{align}
\AAA_\SG=\AAA_\Sp\otimes\bm{1}_\SU+\bm{1}_\Sp\otimes\AAA^{~}_\SU,~~\textrm{and}~~
\FFF_\SG=\FFF_\Sp\otimes\bm{1}_\SU+\bm{1}_\Sp\otimes\FFF_\SU,\label{SGcc}
\end{align}
Corresponding Bianchi identities are obtained by direct calculations as
\begin{align}
\left({d}_\SG\right)^2
=-i\cG\hspace{.1em}\FFF_\Sp\cdot{\Ss}\otimes\hspace{.2em}\bm{1}_\SU
+\bm{1}_\Sp\otimes\left(-i{{c^{~}_\SU}}\hspace{.1em}\hspace{.1em}\FFF_\SU\right),~
{d}_\Sp\FFF_\Sp={d}_\SU \FFF_\SU=0.\label{BianchiSG}
\end{align}

We introduce $N$ spinors such that $\bm{\psi}=(\psi^1,\cdots,\psi^N)^T$ with the $SU(N)$ symmetry.
A covariant differential on $\bm{\psi}\in\S_\Sp^{\otimes N}$ in the spinor-gauge bundle is defined as
\begin{align*}
d_\SG\bm{\psi}&:=
\bm{1}_\SU\otimes{d}\bm{\psi}+
\cG\bm{1}_\SU\otimes\AAA_\Sp\bm{\psi}
-i\hspace{.1em}{{c^{~}_\SU}}\hspace{.1em}\AAA^{~}_\SU\otimes\bm{\psi}.
\end{align*}
A Dirac operator concerning the spinor-gauge bundle is provided as
\begin{align*}
\ds_{\SG}&:\Gamma\left(\M,\Wedge^0\left(\S_\Sp^{\otimes{N}}\otimes{SU(N)}\right)\right)\rightarrow
\Gamma\left(\M,\Wedge^0\left(\S_\Sp^{\otimes{N}}\otimes{SU(N)}\right)\right)\\&:
\bm{\psi}\mapsto
\ds_{\SG}\bm{\psi}:=\iota_{{\gamma}_\Sp}(d-i\cG\hspace{.1em}\wcs)
\hspace{.2em}\bm{\psi}\otimes\bm{1}_\SU
-{i}{c^{~}_\SU}\hspace{.1em}\iota_{{\gamma}_\Sp}\AAA_{\SU}\otimes\bm{\psi},\\
&\hspace{5.5em}=\left(
\gamma_\Sp^\bcdot\partial_\bcdot
-\frac{i}{2}\cG\gamma_\Sp^\bcdot\Varepsilon_\bcdot^\mu
\omega_\mu^{~\stars}{\Ss}_\stars\right)\bm{\psi}
-i{c^{~}_\SU}\gamma_\Sp^\bcdot\Aa^I_\bcdot \tau^{~}_I\hspace{.2em}\bm{\psi}.
\end{align*}

%
% Hodge-dual bundle
%
\subsubsection{Hodge-dual connection and curvature}\label{cg-b}
This section introduces dual connections and curvatures concerning the Hodge-dual operator in the gauge, spinor and  spinor-gauge bundles.

%
% Dual gauge bundle
%
\paragraph{\textbf{Dual curvature in gauge bundle:}}
We defined dual curvature in the gauge bundle as follows:
% Definition
\begin{definition}[Hodge-dual geuge-curvature]
A dual curvature of the $SU(N)$ gauge-curvature is defined owing to the Hodge-dual operator  as
\begin{align}
\HD:\FFF^I_\SU\mapsto\hat\FFF^I_{\SU}
:=\frac{1}{2}\hat\f^{I}_\bcdots\hspace{.1em}\eee^\bcdot\wedge\eee^\bcdot,
~~&\textrm{where}~~
\hat\f^{I}_{ab}:=
\frac{\sqrt{\sigmaeta}}{2}\f^I_{\bcdots}\hspace{.1em}
\eta^{\bcdot\star}\eta^{\bcdot\star}\hspace{.1em}\epsilon_{\stars{ab}}.\label{Fhat}
\end{align}
\end{definition}
\noindent
%%%
We note that two-group operators, ${G}^{~}_\SO$ and ${G}^{~}_\SU$, are homomorphic to each other:
% Remark
\begin{remark}\label{pihatpiSO}
A following commutative diagram is maintained:
\begin{align*}
\begin{tabular}{ccc}
~&$\Gso$&~\\
$\hat\FFF_\SU$ &${\longrightarrow}$&$\hat\FFF_\SU'$\\ 
&~\\
$\hspace{2em}\downarrow$\hspace{.5em}$\Gsu$&~&$\hspace{2em}\downarrow$\hspace{.5em}$\Gsu$\\ 
&~\\
$\hat\FFF_{\SU'}$ &${\longrightarrow}$&$\hat\FFF_{\SU'}'$\\
~&$\Gso$&~
\end{tabular}
\end{align*}
\end{remark}
\begin{proof}
The $SO(1,3)$-transformation of a dual curvature is given as
\begin{align*}
\Gso\left(\hat\FFF^I_{\SU}\right)&=
\LLambda\hat\FFF^I_{\SU}\LLambda^{-1}=
\LLambda\bm\epsilon\bm\eta^{-1}\hspace{.1em}
\FFF^I_{\SU}\hspace{.1em}\bm\eta^{-1}\LLambda^{-1},\\
&=
\left(\LLambda\bm\epsilon\LLambda^{-1}\right)
\left(\LLambda\bm\eta^{-1}\LLambda^{-1}\right)
\left(\LLambda\FFF^I_{\SU}\LLambda^{-1}\right)
\left(\LLambda\bm\eta^{-1}\LLambda^{-1}\right),\\
&=
\bm\epsilon\bm\eta^{-1}
\left(\LLambda\FFF^I_{\SU}\LLambda^{-1}\right)
\bm\eta^{-1}.
\end{align*}
Therefore, it is maintained that 
\[
\Gso\circ\HD\left(\FFF^I_{\SU}\right)=\HD\circ\Gso\left(\FFF^I_{\SU}\right).
\] 
Similar calculations show that a dual curvature is transformed under the $SU(N)$ 
transformation as 
\[
\Gsu\hspace{-.2em}(\hat\FFF^I_{\SU})=\bm{g}^{-1}_\SU\hspace{.1em}\hat\FFF^I_{\SU}\hspace{.1em}\bm{g}_\SU^{~}+
i\hspace{.1em}c^{-1}_\SU\hspace{.1em}\gdg.
\]
Therefore, the remark is maintained.
\end{proof}
%%%
When the dual gauge-curvature is given, the dual gauge connection is provided as a solution of the structure equation.
Though the existence of the dual connection is not trivial, we assume it in this section.
\textbf{Appendix \ref{AppDC}} discusses a necessary condition to exist the dual connection. 

Dual gauge connection $\hat\AAA^I_{\SU}$ is defined as a solution of the structure equation:
\begin{align}
\hat\FFF^{~}_{\SU}&=
d_\www\hat\AAA^{~}_{\SU}
-i\hspace{.1em}{{c}^{~}_\SU}\hspace{.1em}\hat\AAA^{~}_{\SU}\wedge\hat\AAA^{~}_{\SU},\label{eq1}
\end{align}
which consists of $6N$ independent first-order differential equations.
We write the dual gauge connection by means of expression (\ref{HAAA}),
which is compatible with the definition of the dual connection as
\begin{align*}
\HD\left(\hat\FFF^I_{\SU}\right)&=\HD\left(d_\www\hat\AAA^I_{\SU}
+\frac{{c}^{~}_\SU}{2}f^I_{~JK}\hspace{.1em}
\hat\AAA^J_{\SU}\wedge\hat\AAA^K_{\SU}\right)=\FFF^I_{\SU}.
\end{align*}
A solution of equations (\ref{eq1}) for given $\hat\FFF^{~}_{\SU}$ is Lie algebra-valued one-form object such that $\hat\AAA^{~}_{\SU}=\hat\AAA^I_{\SU}\hspace{.2em}\tau^{~}_I=\hat\Aa^I_\bcdot \eee^\bcdot\hspace{.2em}\tau^{~}_I\in\Omega^1\otimes \sss\uuu(N)$.
A following remark ensures that $\hat\AAA^{~}_{\SU}$ is the $SU(N)$ connection in the inertial manifold $\M$. 
% Remark
\begin{remark}\label{rem29}
One-form $\hat\AAA^{~}_{\SU}$ is transformed under the structure group as
\begin{align*}
\Gsu(\hat\AAA^I_{\SU})&=
\bm{g}^{~}_\SU\hspace{.2em}\hat\AAA^{~}_{\SU}\hspace{.2em}\bm{g}_\SU^{-1}+
i\hspace{.1em}{c}^{-1}_\SU\hspace{.1em}\gdg;
\end{align*} 
thus, it is a connection of the dual gauge bundle.
\end{remark}
\begin{proof}
Tensor coefficients of the dual curvature defined in (\ref{Fhat}) are transformed as 
\[
\Gsu(\hat\f^I_{ab})=\bm{g}^{~}_\SU\hspace{.1em}\hat\f^I_{ab}\hspace{.2em}\bm{g}_\SU^{-1}+i\hspace{.1em}{c}^{-1}_\SU\hspace{.1em}\gdg,
\] 
owing to its definition.
One-form object $\hat\AAA^I_{\SU}$ is defined as a solution of equation (\ref{eq1}).
Hence, $\hat\AAA^I_{\SU}$ is a connection due to Remark \ref{remA3}.
\end{proof}
\noindent
%%%
A dual map and its adjoint maps  on a $SU(N)$ connection are  introduced through structure equations and dual maps in the curvature as follows:
\begin{align*}
\begin{tabular}{cccr}
$\AAA^{~}_\SU$ &$\hspace{.5em}\left({\longleftrightarrow}^{\hspace{-1.1em}\HD}\hspace{.55em}\right)$&$\hat\AAA^{~}_{\SU}$&\\ 
&~\\
$\hspace{2em}\updownarrow$\hspace{.5em}(\ref{SUCRV})&~&$\hspace{2em}
\updownarrow$\hspace{.5em}(\ref{eq1})&\textrm{(structure~equations).}\\ 
&~\\
$\FFF^{~}_\SU$ &$\hspace{.1em}{\longleftrightarrow}^{\hspace{-1.1em}\HD}
\hspace{.55em}$&$\hat\FFF^{~}_\SU$
\end{tabular}
\end{align*}
We use short-hand  representations such that  $\HD(\AAA^{~}_\SU)=\hat\AAA^{~}_\SU$ and $\HD\bcdot\HD(\AAA^{~}_\SU)=\AAA^{~}_\SU$.
A covariant differential in the dual bundle is provided as
\begin{align*}
\hat{d}^{~}_\SU\aaa&:=
\bm{1}^{~}_\SU\hspace{.1em}d\aaa-i\hspace{.1em}{{{c}^{~}_\SU}}\hspace{.1em}
[\hat\AAA^{~}_{\SU}\aaa]_\wedge
\end{align*}
for $\aaa\in\Omega^p$, and Bianchi identities in the dual bundle are
\begin{align}
\left(\hat{d}^{~}_{\SU}\right)^2=
-i{{c}^{~}_\SU}\hspace{.1em}\hat\FFF^{~}_{\SU},
~~\textrm{and}~~
\hat{d}^{~}_{\SU}\hat\FFF^{~}_{\SU}=0.\label{BianchihSU}
\end{align}

%
% Dual spnior bundle
%
\paragraph{\textbf{Dual curvature in spnior bundle:}}
We defined the curvature and the connection in the dual inertial bundle in section \ref{CPbundle}; the dual spinor-connection is obtained using the dual spin-connection and the $Spin(3,1)$ generator.
The dual spinor connection is defined as 
\begin{align*}
\hat\AAA_\Sp:=-\frac{i}{2}\hspace{.1em}\eta_\bcdots\eta_\stars
[\hat\www]^{\bcdot\star}\Ss^{\bcdot\star}.
\end{align*}
The dual spinor-curvature is provided for given $\hat\AAA_\Sp$ using the structure equation as
\begin{align*}
\hat\FFF_\Sp^{~}&:=d\hat\AAA_\Sp^{~}+
\hcG\hspace{.1em}\hat\AAA_\Sp^{~}\wedge\hat\AAA_\Sp^{~},
\end{align*}
where the dual curvature fulfils $\HD(\FFF_\Sp^{~})=\hat\FFF_\Sp^{~}$.

Spinor valued  section $\hat\psi\in\Gamma(\M,\S_\Sp\otimes{SU(N)})$ is introduced also in the dual space.
We define an action of the Hodge-dual operator on sections as 
\begin{align*}
\HD(\bm{\psi})={\bm{\hat\psi}}~~\textrm{and}~~\HD({\bm{\hat\psi}})={\bm{\psi}}.
\end{align*}
\textbf{Remark~\ref{lemDOpm}} is maintained also for ${\bm{\hat\psi}}$.
The Dirac conjugate of $\bm{\psi}$ and $\bm{\hat\psi}$ is denoted as $\bm{\bar{\hat{\psi}}}$ and $\bm{\bar{\hat{\psi}}}$, whose components are a Dirac conjugate of each element and is homomorphic with respect to dual and its adjoint operators as follows:
\begin{align*}
\HD(\bm{\bar{\hat{\psi}}})=\overline{\HD({\bm{\psi}})}=\bm{\bar{\hat{\psi}}},
~~\textrm{and}~~
\HD(\bm{\bar{\hat{\psi}}})=\overline{\HD(\bm{\hat{\psi}})}=\bm{\bar{\psi}}.
\end{align*}
Kramers pairs for the spinor sections are provided as
\begin{align*}
\KP\left(\psi\right):=
\left(\psi,\hat\psi\right)
~~\textrm{and}~~
\KP\left(\bar\psi\right):=
\left(\bar\psi,\bar{\hat\psi}\right).
\end{align*}

% Dual spnior-gauge bundle
%
\paragraph{\textbf{Dual curvature in spnior-gauge bundle:}}
The spinor- and gauge connections and their Hodge-dual  belong to the same bundle as \textbf{remark \ref{2bundles}}.
They are provided as  
\begin{align*}
\hat\AAA_\SG:=\hat\AAA_\Sp\otimes\bm{1}_\SU+\bm{1}_\Sp\otimes\hat\AAA_{\SU},~~\textrm{and}~~
\hat\FFF_\SG:=\hat\FFF_\Sp\otimes\bm{1}_\SU+\bm{1}_\Sp\otimes\hat\FFF_{\SU},
\end{align*}
and construct Kramers pairs $\KP(\AAA_\SG^{~})$ and $\KP(\FFF_\SG^{~})$. 

The covariant differential and the dual Dirac operator are provided as
\begin{align*}
\dSg\aaa&:=
d\aaa\hspace{.2em}\bm{1}_\Sp\otimes\bm{1}_\SU
+\hcG\left[\hat\AAA_\Sp,\aaa\right]_\wedge\otimes\bm{1}_\SU+
\bm{1}_\Sp\otimes\left(-i\hspace{.1em}{{c}^{~}_\SU}\hspace{.1em}[\hat\AAA^{~}_{\SU},\aaa]_\wedge\right),\\
\hdSg{\bm{\hat\psi}}&:=
\iota_{{\gamma}_\Sp}
(d-i\hcG\hspace{.1em}\hat\www\cdot{\Ss}){\bm{\hat\psi}}\otimes\bm{1}_\SU
-{i}{{c}^{~}_\SU}\hspace{.1em}\iota^{~}_{\gamma_\Sp}\hat\AAA^{~}_{\SU}\otimes{\bm{\hat\psi}},
\end{align*}
where ${\bm{\hat\psi}}=(\hat\psi^1,\hat\psi^2,\cdots,\hat\psi^N)^T$, and $\aaa\in\Omega^p$ and ${\bm{\hat\psi}}\in\S_\Sp^{\otimes N}\otimes{SU(N)}$.
The dual- and  adjoint operators are homomorphic as follows:
\begin{equation*}
\begin{array}{cclccl}
\HD\left({\dSg}\bm{\psi}\right)&=\hdSg{\bm{\hat\psi}},&
\HD\left(\hdSg{\bm{\hat\psi}}\right)&=\dSg\bm{\psi}.
\end{array}
\end{equation*}
The existence of the dual spinor that provides the dual curvature as a solution of the equation of motion is not trivial and is discussed in \textbf{Appendix \ref{AppDC}}.

%
% Yang--Mills bundle
%
\subsection{Yang--Mills bundle}\label{cb-b}
The gauge-spinor bundle is referred to as the \emph{Yang--Mills bundle} when it is represented using the $(\TxT)$-matrix.
The connection and the curvature of the Yang--Mills bundle are defined in the secondary superspace as
\begin{align*}
\bm\AAA^{~}_\ym:=\bm\AAA^{~}_\Sp+\bm\AAA^{~}_\SU,~~\textrm{where}~~
\bm\AAA^{~}_\Sp:=
\left(
    \begin{array}{cc}
      \AAA^{~}_{\Sp} & {0}\\
            {0}    &\hat\AAA^{~}_{\Sp}
    \end{array}
  \right),~~
\bm\AAA^{~}_\SU:=
\left(
    \begin{array}{cc}
      \AAA^{~}_{\SU} & {0}\\
            {0}    &\hat\AAA^{~}_{\SU}
    \end{array}
  \right).
\end{align*}
The Yang--Mills curvature is provided using the structure equation such that:
\begin{align}
\bm\FFF^{~}_\ym&:=
d\AAA^{~}_\ym
+\cG\hspace{.1em}\bm\AAA^{~}_{\Sp}\wedge\bm\AAA^{~}_{\Sp}\otimes\bm{1}_{\SU}+
\bm{1}_{\Sp\hspace{-.1em}}\otimes(-ic^{~}_\SU\hspace{.1em}
\bm\AAA^{~}_{\SU}\wedge\bm\AAA^{~}_{\SU}),\nonumber\\&=
\left(
    \begin{array}{cc}
      \FFF^{~}_\SG &          0\\
         0       & \hat\FFF^{~}_\SG
    \end{array}
  \right).\label{FYM}
\end{align}
The dual operator in the Yang--Mills bundle is introduced as 
\begin{align*}
\HD^{~}_{\hspace{-.1em}\ym}=
{\bm{1}_\SU\otimes\bm{1}_\Sp}\hspace{.1em}\Gamma\gamma\HD=
{\bm{1}_\SU\otimes\bm{1}_\Sp}
\left(
    \begin{array}{cc}
      0&-\HD \\
      \HD & 0
    \end{array}
  \right)=-\HD^{-1}_{\hspace{-.1em}\ym},
\end{align*}
that acts on the curvature as
\begin{align*}
\HD^{~}_{\hspace{-.1em}\ym}\left(\FFF^{~}_\ym\right):=
\HD^{~}_{\hspace{-.1em}\ym}\hspace{.1em}\FFF^{~}_\ym\hspace{.1em}
\HD^{-1}_{\hspace{-.1em}\ym}
=
\left(
    \begin{array}{cc}
      \hat\FFF^{~}_{\SG} & 0\\
      0& \FFF^{~}_{\SG} \\
    \end{array}
  \right).
\end{align*}
Dual operator $\HD^{~}_{\hspace{-.1em}\ym}$ acts on the Yang--Mills connection formally as defined in section \ref{242}.
A structure equation are homomorphic such that:
%%%%%%
\begin{align*}
\begin{tabular}{cccr}
$\AAA^{~}_\ym$ &$\hspace{.5em}\left(\overset{\HD^{~}_{\hspace{-.1em}\ym}}{\longleftrightarrow}\right)\hspace{.55em}$&$\HD^{~}_{\hspace{-.1em}\ym}\left(\AAA^{~}_\ym\right)$&\\ 
&~&\\
$\hspace{2em}\updownarrow$\hspace{.5em}(\ref{FYM})&~&$
\hspace{1em}
\updownarrow$\hspace{.5em}&\textrm{(structure~equations).}\\ 
&~\\
$\FFF^{~}_\ym$ &$\overset{\HD^{~}_{\hspace{-.1em}\ym}}{\longleftrightarrow}
\hspace{.55em}$&$\HD^{~}_{\hspace{-.1em}\ym}\FFF^{~}_\ym\HD^{-1}_{\hspace{-.1em}\ym}$
\end{tabular}
\end{align*}
that is apparent owing to relations 
\begin{align*}
\HD^{~}_{\hspace{-.1em}\ym}\hspace{.1em}\bm{\XXX}\hspace{.1em}\HD^{-1}_{\hspace{-.1em}\ym}=
\bm\XXX\bigl|_{\bullet\leftrightarrow\hspace{.1em}\hat\bullet},&~~\textrm{for}~~ 
\bm\XXX=\bm\AAA^{~}_\ym~\textrm{or}~~\bm\FFF^{~}_\ym.
\end{align*}
A covariant differential in the Yang--Mills bundle is
\begin{align*}
{\d}^{~}_\ym&:=
\left(
    \begin{array}{cc}
     \dSg & 0 \\
      0 & \hdSg\\
    \end{array}
\right).
\end{align*}
The covariant differential is compatible to the dual operator such as 
\[
\HD^{~}_{\hspace{-.1em}\ym}({\d}^{~}_\ym)=
{\HD^{~}_{\hspace{-.1em}\ym}}\hspace{.1em}{\d}^{~}_\ym\hspace{.1em}{\HD^{-1}_{\hspace{-.1em}\ym}}=
{\d}^{~}_\ym\bigl|_{d_\SG\leftrightarrow \hat{d}_{\SG}},
\]
and the dual operator preserves the Bianchi identities.

The Clifford algebra and Dirac operator are introduced to the Yang--Mills bundle as follows: 
One-dimensional Clifford algebra represented as 
\begin{align*}
\gamma^{~}_\ym:=\bm{1}_\ym\hspace{.2em}\gamma,~~&\textrm{where}~~
\bm{1}_\ym:=\bm{1}_{\Sp}\otimes\bm{1}_{\SU}. 
\end{align*}
and corresponding chiral and projection operators are, respectively, provided as
\begin{align}
\Gamma^{~}_\ym:=\bm{1}_\ym\hspace{.2em}\Gamma~~\textrm{and}~~
P^\pm_\ym:=\frac{\bm{1}^{~}_\ym\pm\gamma^{~}_\ym}{2}.\label{ClalSU}
\end{align}
They fulfil relations $\gamma^{~}_\ym\Gamma^{~}_\ym+\Gamma^{~}_\ym\gamma^{~}_\ym=0$ and $\Gamma^{2}_\ym=\gamma^{2}_\ym=\bm{1}_\ym$.
The Yang--Mills form is defined using the Clifford algebra as follows:
% Definition
\begin{definition}[Yang--Mills form]\label{YMform}
\begin{align*}
\Fym&:=
\gamma^{~}_\ym\hspace{.2em}\FFF^{~}_\ym=\left(
    \begin{array}{cc}
      0&\hat\FFF_{\SG}^{~}\\
      \FFF_{\SG}^{~}& 0
    \end{array}
  \right)\otimes\bm{1}_\Sp.
\end{align*}
\end{definition}
%%%
The secondary superspace is induced on a vector space $\{\FFF^{~}_{\SU}\}/G_\SU\otimes\{\hat\FFF^{~}_{\SU}\}/G_\SU$ owing to chiral operator $\Gamma^{~}_\ym$.
Clifford algebra $\gamma^{~}_\ym$ indices  the Dirac operator such that: 
\begin{align*}
\D^{~}_\ym&:=\gamma^{~}_\ym{\d}^{~}_\ym=
\left(
    \begin{array}{cc}
      0& \dSg\\
      \hdSg& 0\\
    \end{array}
\right).
\end{align*}
 The Dirac spinor defined over the Yang--Mills bundle is defined as
\begin{align*}
\bm\Psi^{~}_\ym:=
\left(
    \begin{array}{c}
     \bm{\psi}  \\
     {\bm{\hat\psi}}     
    \end{array}
  \right)
~&\textrm{yielding}~~
\HD\left(\bm\Psi^{~}_\ym\right)=\left(
    \begin{array}{r}
    -{\bm{\hat\psi}}  \\
     \bm{\psi}
    \end{array}
  \right).
\end{align*}
The Dirac conjugate of $\bm\Psi^{~}_\ym$ and its dual in this representation are provided as
\begin{align*}
\overline{\bm\Psi}^{~}_\ym~:=~
\bm\Psi^{\dagger}_\ym\left(\gamma^0_\Sp\gamma^{~}_\ym\right)
&=\left(\bm{\psi}^\dagger,{\bm{\hat\psi}}^\dagger\right)
\left(
    \begin{array}{cc}
      \bm{0} & \gamma^0_\Sp\otimes\bm{1}_{\SU}\\
      \gamma^0_\Sp\otimes\bm{1}_{\SU}& \bm{0}  \\
    \end{array}
  \right)
=\left(\bar{{{\bm{\hat\psi}}}},\bar{{\bm{\psi}}}\right),
\end{align*}
and 
\begin{align*}
\HD\left(\overline{\bm\Psi}^{~}_\ym\right)=
\left(\bar{{{\bm{\hat\psi}}}},-\bar{{\bm{\psi}}}\right).
\end{align*}

%%% ----------------------------------------------------------------------
% Lagrangian and equation of motion
%%% ----------------------------------------------------------------------
\section{Lagrangian and equation of motion}
In classical mechanics, equations of motion are extracted from the action integral using the variational method with appropriate boundary conditions.
The Lagrangian is a four-form object consisting of connections, curvatures, sections, and invariant differentials in principal bundles.
This section introduces three fundamental Lagrangian, i.e., Einstein--Hilbert Lagrangian for gravity, matter-field and gauge-filed Lagrangians for the Yang--Mills theory with the Lorentzian metric for simplicity.
Then,we extract equations of motion from the Lagrangian.
Extension to the $\theta$-metric is realized by replacing $\bm\eta\mapsto\bm\eta_\theta$,  $\bm\epsilon\mapsto\bm\epsilon_\theta$, $\sigmaeta\mapsto\sigmath$, and $\vvv\mapsto\vvv_\theta$. 
%
% Einstein--Hilbert action
%
\subsection{Einstein--Hilbert action}
The gravitational Lagrangian four-form consists of curvature $\RRR$ and section $\eee$ in the co-Poincar\'{e} bundle; the definition is given as follows\cite{doi:10.1063/1.4990708,Kurihara_2020}:
\begin{definition}[Gravitational Lagrngian]\label{EHLdef}
% Definition
The gravitational Lagrangian and the corresponding action integral are defined as
\begin{align}
\LLL_{Gr}&:=\frac{1}{\hbar\kappa}\left(
\frac{1}{2}
Str\left[\bm\FFF_{Gr}\wedge\bm\FFF_{Gr}\right]-
\Lambda_\mathrm{c}\vvv_\theta\right),\label{EHLformula}
\end{align}
where $\Lambda_\mathrm{c}$ is the cosmological constant.
Integration of the Lagrangian \[\I_{Gr}:=\int_{\Sigma_4}\LLL_{Gr}\] is referred to as the gravitational action, where $\Sigma_4\subseteq\M$ is appropriate closed and oriented subset of the inertial manifold.
\end{definition}
\noindent
%%%
The cosmological constant has a physical dimension of $[\Lambda_\mathrm{c}]=L^{-2}$ in our units.
This study sets the Lagrangian form and action integral to null physical dimension $[\LLL]=[\I]=1$.
Here, fundamental constants $\hbar\kappa$ appear associated with the gravitational Lagrangian to keep null physical dimension. 
Although the Planck constant appears in the Lagrangian,  the theory is still classical.
% Remark
\begin{remark}
The gravitational Lagrangian defined by \textbf{Definition \ref{EHLdef}} is equivalent to the standard Einstein--Hilbert Lagrangian
\begin{align*}
\LLL_\textrm{EH}:=\frac{1}{\hbar\kappa}\left(\frac{1}{2}\RRR^\bcdots\wedge\SSS_\bcdots-\Lambda_\mathrm{c}\hspace{.1em}\vvv
\right).
\end{align*}
\end{remark}
\begin{proof}
The supertrace of the gravitational form squared is provided as
\begin{align*}
Str[\FFF^{~}_{Gr}\wedge\FFF^{~}_{Gr}]&=
\Tr_{\hspace{.1em}\TM}\left[\hat\FFF^{~}_{\cP}\wedge\FFF^{~}_{\cP}\right]-
\Tr_{\hspace{.1em}\TM}\left[\FFF^{~}_{\cP}\wedge\hat\FFF^{~}_{\cP}\right]=
-2\hspace{.1em}\Tr_{\hspace{.1em}\TM}\left[\FFF^{~}_{\cP}\wedge\hat\FFF^{~}_{\cP}\right].
\end{align*}
\textbf{Theorem 4.2} in Ref.\cite{doi:10.1063/1.4990708} and \textbf{Remark 2.1} in Ref.\cite{Kurihara_2020} show that
\begin{align*}
\Tr_{\hspace{.1em}\TM}\left[\FFF^{~}_{\cP}\wedge\hat\FFF^{~}_{\cP}\right]=
\frac{\sqrt{\sigmaeta}}{2}\RRR^{\bcdot\star}\wedge\SSS_{\star\bcdot}=
-\frac{\sqrt{\sigmaeta}}{2}\RRR^{\bcdot\star}\wedge\SSS_{\bcdot\star},
\end{align*}
Therefore, we obtain
\begin{align*}
\frac{1}{2}
Str\left[\FFF_{Gr}\wedge\FFF_{Gr}\right]=
\frac{\sqrt{\sigmaeta}}{2}\RRR^\bcdots\wedge\SSS_\bcdots;
\end{align*}
thus, $\LLL_{Gr}=\sqrt{\sigmaeta}\hspace{.1em}\LLL_\textrm{EH}$ is maintained.
\end{proof}
\noindent
%%%
A space of the two-form objects is split into two subspaces owing to the Hodge-dual operator as mentioned in \textbf{section \ref{242}}; thus, four-form object $\Tr_\TM[\RRR\wedge\SSS]$ is split into two parts such as 
\begin{align*}
\Tr_\TM[\RRR\wedge\SSS]=\Tr[\RRR^+\wedge\SSS^+]+
\Tr_\TM[\RRR^-\wedge\SSS^-].
\end{align*}
More precisely, it is provided that 
\begin{align*}
\Tr_\TM[\RRR^\pm\wedge\SSS^\pm]&=
\frac{1}{2}\left(R\pm 
i \epsilon^{\hspace{.7em}\stars}_{\bcdots}R^{\bcdots}_{\hspace{.7em}\stars}
\right)\vvv,\\
\Tr_\TM[\RRR\wedge\SSS]&=R\hspace{.2em}\vvv,
\end{align*}
where $R$ is scalar curvature (\ref{RicciR}). 
Coordinate expressions using the trivial basis are
\begin{align}
\left[\RRR\wedge\SSS\right]^{ab}&=\frac{1}{4}
\eta_\stars\hspace{.1em}
R^{a\star}_{~~\bcdots}\hspace{.1em}
S^{\star b}_{~~\bcdots}\hspace{.2em}
\eee^\bcdot\wedge\eee^\bcdot\wedge\eee^\bcdot\wedge\eee^\bcdot
,\nonumber\\&=
\frac{1}{4}\epsilon^\bcdott
\eta^{~}_{\stars}\hspace{.2em}
R^{a\star}_{~~\bcdots}\hspace{.1em}
S^{\star b}_{~~\bcdots}\hspace{.2em}\vvv.\label{RSv}
\end{align}

The torsion-less Riemann-curvature tensor has additional symmetry such that:
\begin{align*}
\eta_{a\bcdot}\eta_{b\bcdot}R^{\bcdots}_{~~cd}&=
\eta_{c\bcdot}\eta_{d\bcdot}R^{\bcdots}_{~~ab},&\textrm{when}~~\TTT^a=0,
\intertext{
and thus, one of $\RRR^\pm$ yields the scalar curvature like
}
\eta_\bcdots[\RRR^+\wedge\SSS]^\bcdots&=
\eta_\bcdots[\RRR^-\wedge\SSS]^\bcdots=\frac{1}{2}R\hspace{.1em}\vvv,
&\textrm{when}~~\TTT^a=0.
\end{align*}
The Euler--Lagrange equation of motion requires a torsion-less condition for the Einstein--Hilbert Lagrangian; thus, the Einstein--Hilbert Lagrangian can be written using only one of the subspaces $V^\pm(\Omega^2)$.
When we choose SD curvature $\RRR^+$ (accordingly with $\www^+$) to construct the Lagrangian, the gravitational Lagrangian is represented as
\begin{align*}
\LLL_{Gr}&=\frac{\sqrt{\sigmaeta}}{\hbar\kappa}\left(\RRR^+\wedge\SSS^+-\Lambda_\mathrm{c}\vvv
\right),&\textrm{when}~~\TTT^a=0.
\end{align*}
We note that the volume form can be written using one of $\SSS^\pm$ such that: 
\begin{align*}
\frac{1}{3!}\vvv&=-\eta_\bcdots
\left[\SSS^+\wedge\SSS^+\right]^\bcdots=\eta_\bcdots
\left[\SSS^-\wedge\SSS^-\right]^\bcdots.
\end{align*}
We note that $\Omega^4\ni\aaa\propto\bbb^\pm\wedge\bbb^\pm$ for $\bbb\in\Omega^2$ for any four-form objects due to $\bbb^\pm\wedge\bbb^\mp=0$.

The gravitational Lagrangian is invariant under the general coordinate transformation and the co-Poincar\'{e} transformation\cite{doi:10.1063/1.4990708}.
In the inertial manifold, the number of independent components of the Lagrangian form is 10. Six components are from the curvature two-form, and four are from the vierbein form.
These degrees of freedom correspond to the total degree of freedom for the Poincar\'{e} group.
In reality, the Einstein--Hilbert Lagrangian is given using only one of the SD- or ASD curvature; thus, only three components of the Lorentz connection are utilised in the Lagrangian.
The Lagrangian form is now treated as a function with two independent forms $(\vomega^{ab}, \eee^c)$ and is applied to a variational operator independently.
This method is commonly known as the Palatini method\cite{Palatini2008,ferreris}.

Equations of motion in a vacuum can be obtained by requiring a stationary condition of the action for the connection and vierbein forms, separately.
From the variation concerning connection form $\www$, an equation of motion is provided as
\begin{align}
\TTT^a\wedge\eee^b&=\left(d\eee^a+
\cG\hspace{.1em}\www^{a}_{~\bcdot}\wedge\eee^\bcdot
\right)\wedge\eee^b=0,\label{torsion}
\end{align}
which is referred to as the torsion-less equation.
The torsion-less property is obtained from the solution of the equation of motion rather than being an independent constraint.
This equation of motion includes six independent equations, which is the same as the number of the independent components of $\www$; thus, the connection form is uniquely determined from (\ref{torsion}) when the vierbein form is given.

Next, taking the variation with respect to the vierbein form, one can obtain an equation of motion as
\begin{align}
\frac{1}{2}\epsilon_{a\bcdots\bcdot}
\RRR^\bcdots\wedge\eee^\bcdot
-\Lambda_{\mathrm c}\VVV_a&=0,\label{EoM2}
\end{align}
where $\VVV_a:=\epsilon_{a\bcdot\bcdots}\eee^\bcdot\wedge\eee^\bcdot\wedge\eee^\bcdot/3!$ is a three-dimensional volume form.
Equation (\ref{EoM2}) is referred to as the Einstein equation in a vacuum (those with the Yang--Mills fields are provided in section \ref{3.3}).
The curvature- and vierbein-forms can be uniquely determined up to $GL(4,\R)$ symmetry by solving equations (\ref{torsion}) and (\ref{EoM2}) simultaneously.

We note that component representations of  (\ref{RSv}) and (\ref{EoM2}) are provided in the $\theta$-metric space as  
\begin{align*}
(\ref{RSv})&\rightarrow
\frac{\sqrt{\sigmath}}{4}\left[\bm\epsilon_\theta\right]^\bcdott
\eta_{\theta\hspace{.1em}\stars}\hspace{.2em}
R^{a\star}_{~~\bcdots}\hspace{.1em}
S^{\star b}_{~~\bcdots}\hspace{.2em}\vvv,\\
(\ref{EoM2})&\rightarrow\left(
R_{ab}-\frac{1}{2}R\hspace{.2em}\eta_{\theta\hspace{.1em}ab}+
\Lambda_{\mathrm c}\hspace{.2em}\eta_{\theta\hspace{.1em}ab}
\right)\vvv.
\end{align*}

%
% Yang--Mills action
%
\subsection{Yang--Mills action}
The Yang--Mills Lagrangian consists of spinor field $\bm{\psi}$, gauge connection $\AAA^{~}_\SU$, and dual curvature $\hat\FFF^{~}_{\SU}$.
The spinor field represents a charge distribution of fermionic matter fields in classical mechanics. 
The gauge field represents a potential function as a source of force fields such as electromagnetic, weak, or strong forces.
Their equations of motion govern the dynamics of fields.
Gauge and matter fields contribute to the structure of space-time through their stress-energy tensor.

%
% matter field
%
\subsubsection{Matter field}
The matter-field Lagrangian-form is defined in the spinor-gauge bundle introduced in section \ref{cb-b} using the Dirac operator given in (\ref{DSG}).
A mass term is introduced in the secondary superpsace of the Yang--Mills bundle as follows:
\begin{align*}
\bm{\mu}^{~}_{\ym}&:=
\gamma\left(
    \begin{array}{cc}
    \mu  & 0\\
    0 &   \hat\mu
    \end{array}
  \right)\bm{1}_\Sp\otimes\bm{1}_\SU.
\end{align*}
Real constants $0\leq \mu,\hat\mu\in\R$ are particle masses in physics, and they have a mass dimension $[\mu]=[\hat\mu]=M$.
The matter-field Lagrangian density is defined as
\begin{align}
\LL_M&:=
\overline{\bm\Psi}_{\ym}\left(i\hspace{.1em}\D^{~}_\ym-
\frac{1}{\hbar}\bm{\mu}^{~}_{\ym}
\right)\bm\Psi_{\ym}.\label{LagrangianM}
\end{align}
Physical constant $\hbar$ is inserted to keep a physical dimension of Lagrangian densities to $[\LL]=L^{-4}$.  
% Definition
\begin{definition}[Matter-field Lagrangian and action]\label{mfl}
The Lagrangian and action integral for the matter field are defined as
\begin{align*}
\LLL_M:=\LL_M
\hspace{.1em}\vvv~~\mathrm{and}~~
\I_M:=\int_{\Sigma_4}\LLL_M.
\end{align*}
\end{definition}
\noindent
%%%
The matter-field Lagrangian is Hermitian self-adjoint.

%
% gauge field
%
\subsubsection{Gauge field}\label{Gaugefield}
The gauge-field Lagrangian is defined by means of a super trace as follows:
% Definition
\begin{definition}[Gauge-field Lagrangian and action]
The gauge-field Lagrangian and corresponding gauge action are defined  as follows:
\begin{align}
\LLL_{\ym}:=-\frac{1}{2}
Str[\Fym\wedge\Fym],&~~~
\I_\ym:=\int_{\Sigma_4}\LLL_{\YM},\label{LagrangianSU}
\end{align}
which is equivalent to the standard definition of the Yang-Mills Lagrangian.
\end{definition}
\noindent
%%%
Supertrace of the Yang--mills form squared is provided as
\begin{align*}
Str[\Fym\wedge\Fym]&=
\Tr_{\hspace{.1em}\SU}\left[\hat\FFF^{~}_\SG\wedge\FFF^{~}_\SG\right]-
\Tr_{\hspace{.1em}\SU}\left[\FFF^{~}_\SG\wedge\hat\FFF^{~}_\SG\right]=
-2\hspace{.1em}\Tr_{\hspace{.1em}\SU}\left[\FFF^{~}_\SG\wedge\hat\FFF^{~}_\SG\right],
\end{align*}
thus, the gauge-field action is represented using the trivial basis as
\begin{align*}
\I_{\ym}=
\int_{\Sigma_4}\|\FFF^{~}_\SG\|^2=\frac{1}{4}\hspace{.1em}\int_{\Sigma_4}
\left(\sum_{I=1}^N\eta^{\bcdot\star}\eta^{\bcdot\star}
\f^I_\bcdots\f^I_\stars
\hspace{.2em}\Tr\left[\tau^{I}_\SU\cdot\tau^{I}_\SU\right]
\right)\vvv,
\end{align*}
which can be verified by direct calculations from the definition (\ref{Fhat}).
This Lagrangian is the same as the standard definition of the Yang--Mills Lagrangian\cite{kobayashi1996foundations, donaldson1990geometry}.

\begin{remark}\label{SDorASD}
The Yang--Mills action with the ${SU(N)}$ principal group in the Eudlidean metric space has extremal at the SD- or ASD curvature.
\end{remark}
\begin{proof}
The Yang--Mills action has the representation owing to SD and ASD curvatures as
\begin{align*}
\I_\ym&=\int_{\Sigma_4}\|\FFF^{~}_\SG\|^2=
\frac{1}{2}\int_{\Sigma_4}\left(
\|\FFF^{+}_\SG\|^2+\|\FFF^{-}_\SG\|^2\right).
\end{align*}
Here, the trace concerning $SU(N)$ indices is not explicitly written for simplicity.
On the other hand, the second Chern class is provided as
\begin{align*}
\sqrt{\sigmaeta}\hspace{.1em}c^{~}_2\left(\FFF^{~}_\SG\right)&=\frac{1}{8\pi^2}
\int_{\Sigma_4}\left(\Tr_{\hspace{.1em}\SU}\left[
\FFF^{~}_\SG\wedge\FFF^{~}_\SG\right]
-\Tr_{\hspace{.1em}\SU}\left[\FFF^{~}_\SG\right]^2
\right),\\&=
\frac{1}{8\pi^2}
\int_{\Sigma_4}\left(
\|\FFF^{+}_\SG\|^2-\|\FFF^{-}_\SG\|^2\right).
\end{align*}
From the first line to the second line, we use $\Tr_{\hspace{.1em}\SU}\left[\FFF^{~}_\SG\right]=0$ due to $\Tr\left[\tau_\SU^I\right]=0$ for a $SU(N)$ gauge group.
When $\|\FFF\|^2$ is positive definite, equivalently, when the space-time manifold has the \emph{Euclidean} metric, they yield $\I_\ym\geq4\pi^2\left|{c}_2\left(\FFF^{~}_\SG\right)\right|$; thus, the remark is maintained. 
\end{proof}
\noindent
In the $\theta$-metric space, it is true only when $\theta=1$.
We note that there are solutions to the Yang--Mills equation other than the SD- nor ASD connection.
Sibner, Sibner and Uhlenbeck\cite{doi:10.1073/pnas.86.22.8610} reported the $SU(2)$ Yang--Mills connections which are not the SD- nor ASD connection in $S^4$.

%
% full Lagrangian and equation of motion in Yang--Mills theory
%
\subsection{Full Lagrangian and equation of motion in Yang--Mills theory}\label{s323}
In summary, the Yang--Mills Lagrangian and action integral is obtained as
\begin{align*}
\LLL_\MYM:=\LLL_M+\LLL_\ym~~\textrm{and}~~
\I_\MYM:=\int\LLL_\MYM,
\end{align*}
and it is expressed using the trivial basis as
\begin{align}
\LLL_\MYM
&=\left\{
\bar\psi\left(i\hspace{.1em}\gamma_\Sp^\bcdot\partial_\bcdot-\frac{i}{2}
\cG\hspace{.1em}\gamma_\Sp^\bcdot
\Varepsilon_\bcdot^\mu\omega_\mu^{~\stars}\Ss_\stars+
{c^{~}_\SU}\hspace{.1em}\gamma_\Sp^\bcdot
%\Varepsilon^\mu_\bcdot\Aa^I_\mu
\Aa^I_\bcdot
 \tau^{~}_I\right)\psi-\frac{1}{\hbar}{\mu}\bar\psi\psi\right\}\vvv \nonumber\\
&~+\left\{(\psi,{c^{~}_\SU},\mu)\rightarrow(\hat\psi,{{c}^{~}_\SU},\hat\mu)\right\}\vvv
+\frac{1}{4}\f^I_\bcdots\f_I^\bcdots\hspace{.1em}\vvv.
%+\frac{1}{4}\hat\f^I_\bcdots\hat\f_I^\bcdots\vvv.
\label{LYM2}
\end{align}
Terms with dual field $\hat\psi$ in formula (\ref{LYM2}) represents the magnetic monopole in physics, which is not confirmed experimentally to date.
This section treats only principal filed $\psi$.
Section \ref{EMsYM} discusses a magnetic monopole in detail.

% Remark
\begin{remark}[Yang--Mills equations of motion]\label{DiracEq}
Two equations of motion is obtained from the Yang--Mills action such that
\begin{subequations}
\begin{align}
\left(i\ds_{\SG}-\frac{\mu}{\hbar}\right)\psi&=0,\label{DiracEq1}\\
d_\SG\hat\FFF^{~}_{\SU}
-{c^{~}_\SU}\hspace{.1em}\bar\psi\gamma_\Sp^\bcdot\psi\hspace{.2em}\VVV_\bcdot&=0.\label{DiracEq2}
\end{align}
\end{subequations}
The first equation is known as the Dirac equation, and the second equation is the Yang--Mills equation.
\end{remark}
\begin{proof}
The Dirac equation can be easily obtained by taking a variational operation on the Yang--Mills Lagrangian concerning  conjugate field $\bar\psi$ as 
\begin{align*}
\hbar\hspace{.1em}\delta_{\bar\psi}\I_\YM=\int_{\Sigma_4}\delta\bar\psi
\left(i\hbar\hspace{.1em}\dSg\psi-{\mu}\psi\right)
\hspace{.1em}\HD(1)=0&\Rightarrow i\hbar\hspace{.1em}\dSg\psi-{\mu}\psi
=0.
\end{align*}
For gauge-field $\AAA^{~}_\SU$, an equation of motion is provided owing to a variational operation concerning connection $\AAA^{~}_\SU$ such that:
\begin{align*}
\delta_{\AAA^{~}_\SU}\I_\YM&=-\frac{1}{2}\int_{\Sigma_4}
\sum_I\left(
(d\hat\FFF^I_{\SU}
-{c^{~}_\SU}\hspace{.1em}\bar\psi\hspace{.1em}\iota_{{\gamma}_\Sp}\vvv\psi)\delta_K^I
+{c^{~}_\SU}f^I_{~JK}\hat\FFF^I_{\SU}\wedge\AAA_\SU^J
\right)\wedge\delta\AAA_\SU^K\\&~
-\frac{1}{2}\int_{\partial\Sigma_4}
\Tr\left[\hat\FFF^{~}_{\SU}\wedge\delta\AAA^{~}_\SU\right]=0,\\
&\Rightarrow
d_\SG\hat\FFF^{~}_{\SU}
={c^{~}_\SU}\hspace{.1em}\bar\psi\gamma_\Sp^\bcdot\psi\hspace{.2em}\VVV_\bcdot,
\end{align*}
where $\partial\Sigma_4$ is the integration boundary in which variation $\delta\AAA^{~}_\SU$ vanishes.
In above calculations, a relation
\begin{align*}
\delta_{\AAA^{~}_\SU}
\left(\bar\psi\hspace{.1em}\iota_{{\gamma}_\Sp}
\AAA_\SU^I\hspace{.1em}\psi\hspace{.1em}\vvv\right)&=
\frac{\delta}{\delta\AAA_\SU^K}
\left(\bar\psi\hspace{.1em}\iota_{{\gamma}_\Sp}\AAA_\SU^I\psi\hspace{.1em}\vvv\right)
\delta\AAA_\SU^K,\\&
=\bar\psi\left(\iota_{{\gamma}_\Sp}\vvv\right)
\psi\hspace{.1em}\delta\AAA_\SU^I
\end{align*}
with 
\[
0=\iota_{{\gamma}_\Sp}(\AAA_\SU^I\wedge\vvv)=
(\iota_{{\gamma}_\Sp}\AAA_\SU^I)\hspace{.2em}\vvv-
\AAA_\SU^I\wedge\left(\iota_{{\gamma}_\Sp}\vvv\right)
\] 
is used.
\end{proof}
\noindent
%%%

An expression of equations (\ref{DiracEq2}) using the trivial basis in $\TM$ is provided as
\begin{subequations}
\begin{align}
&~\hbar\left(i\gamma_\Sp^\bcdot\partial_\bcdot
-\frac{i}{2}\cG\hspace{.1em}\gamma_\Sp^\bcdot\Varepsilon_\bcdot^\mu\omega_\mu^{~\stars}{\Ss}_\stars
+{c^{~}_\SU}\hspace{.1em}\gamma_\Sp^\bcdot
%\Varepsilon^\mu_\bcdot\Aa^I_\mu
\Aa^I_\bcdot
\tau_I\right)\psi={\mu}\hspace{.1em}\psi,\label{YMeqcom1}\\
&~\left(\eta^{\bcdots}\partial_\bcdot\f^I_{~\bcdot a}
%+\frac{1}{2}\cG\hspace{.1em}\Varepsilon_a^\mu\omega_\mu^{~\bcdots}\hspace{.1em}\f^I_{~\bcdots}
+{c^{~}_\SU}\hspace{.1em}\eta^{\bcdots}f^I_{~JK}\Aa^J_{~\bcdot}\f^K_{~\bcdot a}
\right)\tau^{~}_I=
-{c^{~}_\SU}\eta_{a\bcdot}\hspace{.1em}\bar\psi\gamma_\Sp^\bcdot\psi.\label{YMeqcom2}
\end{align}% see Fh-test-20190920.nb
\end{subequations}
In equation (\ref{YMeqcom2}), information of the curved space-time is included in the component representation of $\f^I$ provided in (\ref{FabISU}) through a torsion of the space-time manifold.  

The Noether's theorem ensures vanishing divergence such that:
\begin{align*}
d\left(\bar\psi\hspace{.1em}\iota_{{\gamma}_\Sp}\psi\right)=
\partial_\bcdot\left(\bar\psi\gamma^\bcdot\psi\right)=0,
\end{align*}
which is referred to as a current conservation in physics.

%%% ----------------------------------------------------------------------
% Scalar action
%%% ----------------------------------------------------------------------
\subsection{Scalar action}\label{scalarfield}
This section provides a geometrical construction of the Higgs field.
A scalar field belonging to the fundamental representation of  the gauge field is referred to as the Higgs filed in physics.
Higgs field $\bm\phi\in C^\infty(\TM)$ is introduced as a complex-valued section in the gauge bundle.
The equation of motion for the Higgs field is known as the Klein--Gordon equation and is provided from an action integral using the trivial basis in $\M$ such as
\begin{align*}
\tilde\I_{\textrm{Higgs}}:=-\frac{1}{4}\int_{\Sigma_4}
\left[\eta^{\bcdots}\nabla_\bcdot\bm\phi^\dagger\nabla_\bcdot\bm\phi
-V\left(\bigl|\bm\phi\bigr|\right)
\right]\vvv,
\end{align*}
where $\bigl|\bm\phi\bigr|^2=\bm\phi^\dagger\bm\phi$, $\nabla_\mu=\partial_\mu-i\hspace{.1em}c^{~}_\SU\Aa^I_\mu t_I$ is a $SU(N)$-covariant differential in $\TMM$, and $V\left(\bigl|\bm\phi\bigr|\right)$ is a $SU(N)$-invariant potential term.

The Dirac equation was originated from a genius idea by Dirac to realize a square-root of Laplace operator $\Box:=\eta^{\bcdots}\partial_\bcdot\partial_\bcdot$.
In this study, the Laplacian is defined as a square of the Dirac operator; thus, an action integral of the Higgs field is defined as
\begin{align*}
\I_{\textrm{Higgs}}&:=-\frac{1}{4}\int_{\Sigma_4}
\Tr\left[(\dSg\bm\phi^\dagger)\dSg\bm\phi
-{\bm{1}}_\Sp V\left(\bigl|\bm\phi\bigr|\right)\right]\vvv.
\end{align*}
% Remark
\begin{remark}\label{remB1}
An action integral of the Higgs field is expressed using the trivial basis as
\begin{align*}
\I_{\textrm{Higgs}}&=\int_{\Sigma_4}\hspace{-.1em}\vvv\hspace{.1em}\bm\phi^\dagger
\left\{\Box
+\frac{R}{4}
+c_s\partial^\bcdot\Aa_\bcdot+c_s^2\Aa^\bcdot\Aa_\bcdot
+\frac{1}{2}\left(\partial_\bcdot\Varepsilon^\mu_\bcdot+
c_sA_\bcdot\Varepsilon^\mu_\bcdot\right)\omega_\mu^{~\bcdots}
\right\}\bm\phi\\&\hspace{.5em}
+\int_{\Sigma_4}\hspace{-.1em}\vvv\hspace{.1em}V\left(\bigl|\bm\phi\bigr|\right)
-\frac{1}{4}\int_{\Sigma_4}\hspace{-.1em}\vvv\hspace{.1em}\partial_\bcdot
\left(\Tr\left[
\gamma_\Sp^\bcdot\dSg\right]\bigl|\bm\phi\bigr|^2\right).
\end{align*}
An Euler--Lagrange equation of motion is provided by requiring a stationary condition on a variational operation concerning the field $\bm\phi^\dagger$.
The last term does not contribute to an equation of motion owing to the boundary condition.
\end{remark}
\begin{proof}
Direct calculations provide the result similar with the Lichnerowicz formula\footnote{See, e.g., chapter 1.3 of Ref.\cite{markvorsen2003global}.} such that:
\begin{align*}
\Tr\left[(\dSg\bm\phi^\dagger)(\dSg\bm\phi)\right]
&=
\Tr\left[
\left(\gamma_\Sp^a\partial_a\bm\phi^\dagger
+i\hspace{.1em}\iota_{{\gamma}_\Sp}\OOO\bm\phi^\dagger\right)
\left(\gamma_\Sp^b\partial_b\bm\phi
+i\hspace{.1em}\iota_{{\gamma}_\Sp}\OOO\bm\phi\right)
\right],\\
&=-\Tr\left[
\bm\phi^\dagger\gamma_\Sp^a\left(\partial_a\bm\phi
\gamma_\Sp^b\partial_b\bm\phi\right)
+i\hspace{.1em}\bm\phi^\dagger\gamma_\Sp^a
\left(\partial_a\iota_{{\gamma}_\Sp}\OOO\bm\phi\right)
\right]\\&~
+\Tr\left[
\gamma_\Sp^a\partial_a\bm\phi^\dagger\left(
\gamma_\Sp^b\partial_b\bm\phi
+i\hspace{.1em}\iota_{{\gamma}_\Sp}\OOO\bm\phi
\right)\right]\\&~
+\Tr\left[
i\hspace{.1em}\iota_{{\gamma}_\Sp}\OOO\bm\phi^\dagger
\gamma_\Sp^a\partial_a\bm\phi
\right]
-\Tr\left[\left(\iota_{{\gamma}_\Sp}\OOO\right)^2
\right],\\
&=-4\bm\phi^\dagger\partial^a\partial_a\bm\phi-i\bm\phi^\dagger
\Tr\left[\gamma_\Sp^a\left(\partial_a\iota_{{\gamma}_\Sp}\OOO\right)\bm\phi\right]
\\&~
-\Tr\left[\left(\iota_{{\gamma}_\Sp}\OOO\right)^2\right]\bigl|\bm\phi\bigr|^2
+\partial_\bcdot
\left(\Tr\left[
\gamma_\Sp^\bcdot\dSg\right]\bigl|\bm\phi\bigr|^2\right).
\end{align*}
The second and third terms, respectively, give formulae such that:
\begin{align*}
&\bm\phi^\dagger
\Tr\left[\gamma_\Sp^\bcdot\left(\partial^{~}_\bcdot\iota_{{\gamma}_\Sp}\OOO\right)\bm\phi\right]\\
&=\frac{1}{2} \Tr\left[\gamma_\Sp^\diamond\gamma_\Sp^\diamond\Ss^\bcdots\right]
\eta_{\bcdot\star}\eta_{\bcdot\star}\left(\partial^{~}_\diamond\Varepsilon^\mu_\diamond\omega_\mu^{~\stars}
\right)\bigl|\bm\phi\bigr|^2
+4i\left(c_s\partial^{~}_\bcdot\Aa^\bcdot\right)\bigl|\bm\phi\bigr|^2,\\
&=-2i\left(\partial^{~}_\bcdot\Varepsilon^\mu_\bcdot\omega_\mu^{~\bcdots}\right)
\bigl|\bm\phi\bigr|^2
+4ic_s\left(\partial^{~}_\bcdot\Aa^\bcdot\right)\bigl|\bm\phi\bigr|^2,
\end{align*}
and
\begin{align*}
-\Tr\left[\left(\iota_{{\gamma}_\Sp}\OOO\right)^2\right]\bigl|\bm\phi\bigr|^2=
-2\left(\Varepsilon^\mu_\bcdot\omega_{\mu~\star}^{~\bcdot}\hspace{.2em}
\Varepsilon^\nu_\diamond\omega_{\nu}^{~\star\diamond}
+c^{~}_s\Aa^{~}_\bcdot\Varepsilon^\mu_\bcdot\omega_\mu^{~\bcdots}+2c_s^2\Aa^{~}_\bcdot\Aa^\bcdot\right)
\bigl|\bm\phi\bigr|^2,
\end{align*}
where a $SU(N)$ index of $\Aa$ is omitted for simplicity.
On the other hand, from the definition of the Lorentz curvature,
\begin{align*}
\RRR^{ab}&=d\www^{ab}+\cG\hspace{.1em}\www^a_\bcdot\wedge\www^{\bcdot b},\\
&=\left(\partial_\mu\omega_\nu^{~ab}
+\cG\hspace{.1em}\omega_{\mu~\bcdot}^{~a}\hspace{.2em}\omega_{\nu}^{~\bcdot b}
\right)dx^\mu\wedge dx^\nu,\\
&=\left\{
\left(\partial_\bcdot\Varepsilon^\mu_\bcdot\omega_\mu^{~ab}\right)
-\left(\partial_\bcdot\Varepsilon^\mu_\bcdot\right)\omega_\mu^{~ab}
+\cG\hspace{.1em}\Varepsilon^\mu_\bcdot\omega_{\mu~\star}^{~a}\hspace{.2em}
\Varepsilon^\nu_\bcdot\omega_{\nu}^{~\star b}
\right\}\eee^\bcdot\wedge\eee^\bcdot,\\
&=:\frac{1}{2}R^{ab}_{\hspace{.6em}\bcdots}\hspace{.2em}\eee^\bcdot\wedge\eee^\bcdot,\\
&\Longrightarrow R:=R^{\bcdot\star}_{\hspace{.6em}\bcdot\star}=2\left(
\partial^{~}_\bcdot\omega_\bcdot^{~\bcdots}
+\cG\hspace{.1em}\omega_{\star}^{~\star\bcdot}\omega_{\diamond}^{~\bcdot \diamond}\eta^{~}_{\bcdots}
-\left(\partial^{~}_\bcdot\Varepsilon^\mu_\bcdot\right)\omega_\mu^{~\bcdots}
\right),
\end{align*}
where $\omega_{a}^{~bc}:=\Varepsilon^\mu_a\omega_\mu^{~bc}$.
Finally, by gathering up above results into $\I_{\textrm{Higgs}}$, the remark is maintained.
\end{proof}
\noindent
%%%
Consequently, the equation of motion for the Higgs field is provided as the Euler--Lagrange equation such that:
\begin{multline*}
\left(\Box
+\frac{R}{4}
+\frac{\partial V\left(\bigl|\bm\phi\bigr|\right)}
{\partial\bm\phi}
\right)\bm\phi
+\eta^\bcdots_{~}\sum_I\left(c^{~}_s\partial^{~}_\bcdot\Aa^I_\bcdot
+c_s^2f^I_{~JK}\Aa^J_\bcdot\Aa^K_\bcdot\right)\bm\phi\\
\hspace{.5em}+\frac{\cG}{2}\left(\partial^{~}_\bcdot\Varepsilon^\mu_\bcdot+
c^{~}_s\sum_I\Aa^{I}_\bcdot\Varepsilon^\mu_\bcdot\right)\omega_\mu^{~\bcdots}
\bm\phi=0,
\end{multline*}
where $SU(N)$ indices are indicated explicitly.
It is clear that this equation is $GL(4,\R)$ and $SU(N)$ invariant owing to its construction.
%
% torsion and stress-energy forms of Yang--Mills theory
%
\subsection{Torsion and stress-energy forms of Yang--Mills theory}\label{3.3}
Yang--Mills field is a source of torsion and curvature of space-time.
The right-hand side of the torsion equation (\ref{torsion}) and the Einstein equation (\ref{EoM2}) is given by torsion and a stress-energy tensor of the Yang--Mills field, respectively.
They are defined as follows:
% Definition
\begin{definition}[torsion and stress-energy forms]\label{tse}
Torsion two-form $\TTT_{\YM}$ and stress-energy three-form $\EEE_{\YM}$ of the Yang--Mills Lagrangian are respectively defined as
\begin{align*}
\TTT_{\YM}\wedge\eee:=\frac{\delta\LLL_{\YM}}{\delta\www}
-d\left(\frac{\delta\LLL_{\YM}}{\delta(d\www)}\right)~~\mathrm{and}~~
\EEE_{\YM}:=\frac{\delta\LLL_{\YM}}{\delta\eee}
-d\left(\frac{\delta\LLL_{\YM}}{\delta(d\eee)}\right).
\end{align*}
\end{definition}
%\noindent
%%%

% Remark
\begin{remark}[torsion and stress-energy forms of Yang--Mills Lagrangian]\label{YMtse}
For the Yang--Mills Lagrangian, a torsion form is provided as
\begin{align}
\TTT_{\YM}^a&=\left\{K^a_{~b_1b_2}+\frac{1}{2}\left(
\delta^a_{b_1}K^\bcdot_{~b_2\bcdot}-
\delta^a_{b_2}K^\bcdot_{~b_1\bcdot}\right)
\right\}\eee^{b_1}\wedge\eee^{b_2},\label{tform}
\end{align}
where
\begin{align}
K^a_{~bc}&:=-i\hbar
\left(\bar\psi\gamma_\Sp^a{\Ss}^{~}_{bc}\psi\right),\label{kform}
\end{align}
and the stress-energy form is
\begin{align}
[\EEE_\YM]_a&=\left\{\hbar
\bar\psi\left(i\hspace{.1em}\gamma_\Sp^\bcdot\partial_\bcdot-\frac{i}{2}\cG\hspace{.1em}
\gamma_\Sp^\bcdot\Varepsilon_\bcdot^\mu\omega_\mu^{~\stars}{\Ss}^{~}_\stars+
{c^{~}_\SU}\hspace{.1em}\gamma_\Sp^\bcdot\Varepsilon^\mu_\bcdot\Aa^I_\mu t_I
\right)\psi
-{m_e}\bar\psi\psi
\right.\nonumber
\\ 
&\hspace{2em}\left.-\frac{\hbar}{4}\f^I_\bcdots\f_I^\bcdots\right\}\VVV_a.\label{esform}
\end{align}
\end{remark}
\begin{proof}
For the torsion form, a tensor $\bm{K}$ is introduced as 
\begin{align*}
\frac{\delta\LLL_{\YM}}{\delta\www^{ab}}&=:
K^\bcdot_{~ab}\VVV_\bcdot.
\end{align*}
Coefficients of tensor $\bm{K}$ can be obtained by direct calculations from (\ref{LYM2}) as
\begin{align*}
\frac{\delta\LLL_{\YM}}{\delta\www^{ab}}=
\frac{\delta\LLL_{\YM}}{\delta\omega_\mu^{~ab}}\partial_\mu&=
-i\hbar\hspace{.1em}\cG
\left(\bar\psi\gamma_\Sp^\bcdot{\Ss}^{~}_{ab}\psi\right)
\frac{1}{3!}\epsilon_{\bcdot\star\stars}
\eee^\star\wedge\eee^\star\wedge\eee^\star,\\
&=-i\hbar\hspace{.1em}\cG
\left(\bar\psi\gamma_\Sp^\bcdot{\Ss}^{~}_{ab}\psi\right)\VVV_\bcdot,
\end{align*}
and thus, (\ref{kform}) is provided.
Torsion form (\ref{tform}) can be obtained by solving algebraic equation 
\[
\epsilon_{ab\bcdots}\hspace{.1em}\TTT_{\YM}^\bcdot\wedge\eee^\bcdot=
K^\bcdot_{~ab}\hspace{.1em}\VVV_\bcdot.
\]
Energy-stress form $\EEE_{\YM}$ can be obtained by direct calculations from (\ref{LYM2}).
\end{proof}
%%%
In summary, the Einstein equation and the torsion equation with the Yang--Mills field are respectively provided using the trivial basis in $\TM$ as
\begin{align}
\frac{1}{\kappa}
\left(
\frac{1}{2}\epsilon_{a\bcdots\bcdot}
\RRR^\bcdots\wedge\eee^\bcdot
-\Lambda_{\mathrm c}\VVV_a\right)&=-[\EEE_{\YM}]_a
~\Leftarrow~(\ref{esform}),\label{EoM3}
\end{align}
and
\begin{align*}
\frac{1}{\kappa}\left(d\eee^a+
\cG\hspace{.1em}\www^{a}_{~\bcdot}\wedge\eee^\bcdot\right)
=\hbar\left(-i\bar\psi\gamma_\Sp^a{\Ss}^{~}_{\bcdot\star}\psi+
2\delta^a_{\hspace{.2em}
[\bcdot}{\eta}^{~}_{\star]\diamond}\bar\psi\gamma_\Sp^\diamond\psi
\right)\eee^{\bcdot}\wedge\eee^{\star}.
\end{align*}

The stress-energy due to the gauge field tensor (corresponding to the last term of the right-hand side of (\ref{esform})) is provided as $T^{ab}_\SU:=-{\hbar}\hspace{.1em}\eta^{ab}\f^I_\bcdots\f_I^\bcdots/4$.
Whereas the standard form of the stress-energy tensor is traceless, its trace is not zero.
The standard traceless stress-energy tensor can be obtained by adding total derivative term $\hbar\hspace{.1em}d(\AAA^{~}_\SU\wedge\FFF_\SU)$ to Lagrangian $\LLL_{\ym}$.
This additional term gives term $\hbar\hspace{.1em}\eta_\bcdots\f_I^{a\bcdot}\f_I^{b\bcdot}$ in the stress-energy tensor, and the standard representation of the stress-energy tensor for the gauge field is obtained as
\begin{align}
T^{\overline{ab}}_\SU&:=\hbar\hspace{.1em}\sum_I\left(
\eta_\bcdots\f_I^{a\bcdot}\f_I^{b\bcdot}-\frac{1}{4}\eta^{ab}\f^I_\bcdots\f_I^\bcdots
\right),\label{TabSU}
\end{align}
which is traceless such as $\eta_{\bcdots}T^{\bcdots}_\SU=0$.

%%% ----------------------------------------------------------------------
% Characteristic index
%%% ----------------------------------------------------------------------
\section{Index theorem in general relativistic Yang--Mills theory}
Both Lagrangians of general relativity and the Yang--Mills theory are defined using the supertrace in the secondary superspace.
According to methods given in section \ref{AppIndex} and \ref{themetric}, this section discussed a topological index in the general relativistic Yang--Mills bundle with Euclidean and Lorentz metrics.

%
% Chern class of space-time manifold
%
\subsection{Chern index of space-time manifold}
This section treats the characteristic class of the inertial manifold.
The characteristic class is an element of the cohomology group of the classifying space of the principal bundle.
We utilize the Chren characteristic to classify the inertial manifold.
The first Chern class of the inertial manifold is zero since the curvature $\RRR$ is traceless concerning tensor coefficients in $\TM$.
In the four-dimensional manifold, only the second Chern class can have a non-zero value provided such that:
\begin{align*}
c^{~}_2\left(\RRR\right)=p_1\left(\RRR\right)=
\frac{\Tr_{\TM}\left[\RRR\wedge\RRR\right]}
{8\pi^2}\in H^4(\M,\Z)
\end{align*}
where $\RRR$ is provided as a solution of the Einstein equation and the torsion-less condition.
A flat inertial manifold has no singularity, and its domain is the whole $\M{\cong}\R^4$.
After the one-point compactification of $\M$, the Poincar\'{e} duality theorem insists that $H^4(\M{\cong}S^4,\Z){\cong}H_0(S^4,\Z){\cong}\Z$.
The de\hspace{.2em}Sitter space-time, as well as the Friedmann--Lema\^{i}tre--Robertson--Walker space-time, has topology $M_\textrm{dS}\cong\R{\otimes}S^3$, and after the one-point compactification of the time coordinate, it has cohomology such that:
\begin{align*}
H^4(S^1{\otimes}S^3,\Z)&{\cong}\sum_{i+j=4}
H^i(S^1,\Z){\otimes}H^j(S^3,\Z){\cong}
H_0(S^1,\Z){\otimes}
H_0(S^3,\Z){\cong}\Z\times\Z.
\end{align*}
Here, we use the K\"{u}nneth theorem and the Poincar\'{e} duality theorem.

On the other hand, black hole solutions have a singularity, e.g., the Schwarzschild solution is defined in $\R^1{\otimes}\R^3\setminus\{0\}$ and the Kerr--Newman solution has a domain as $\R^1{\otimes}\R^3{\setminus}S^1$.
The singularity-ring $S^1$ in the Kerr--Newman solution is homotopically contractible; thus, its homology is the same as the Schwarzschild solution's one.
Therefore, it is enough to consider the cohomology of the Schwarzschild solution. 
An oriented and closed domain of the Schwarzschild solution, denoted as $\M_{\hspace{-.1em}\textrm{Schw}}$, is constructed by means of the one-point compactification of $\M$ as follows:
Time variable $t$ is one-point compactified as $\R{\mapsto}S^1$.
A three-dimensional spacial manifold without the origin is the closed manifold such that $\R^3\setminus\{0\}{\cong}S^3$; thus, the domain of the Schwarzschild solution is a closed manifold $\M_{\hspace{-.1em}\textrm{Schw}}{\cong}S^1{\otimes}S^3$.
The cohomology of a domain of the Schwarzschild solution is provided as 
$H^4(\M_{\hspace{-.1em}\textrm{Schw}},\Z){\cong}\Z\times\Z$.

In summary, the second Chern class of the known solutions of the Einstein equation in the inertial bundle has integral cohomology.
In reality, direct calculations of $\Tr_{\TM}[\RRR\wedge\RRR]$ with the $\theta$-metric show 
\begin{align*}
c^{~}_2\left(\RRR\right)&\propto
\eta_{\theta\hspace{.1em}\bcdots}\hspace{.2em}\eta_{\theta\hspace{.1em}\stars}
\RRR^{\bcdot\star}\wedge\RRR^{\star\bcdot}=0,
\end{align*}
for all solutions listed above.
Therefore, the Chern characteristic is also zero for the space-time manifold of the known solution of the Einstein equation.

%
% index in spinor-gauge bundle
%
\subsection{Index in spinor-gauge bundle}\label{EMsYM}
The $\Z_2$-grading superspace of gauge fields consisting of $\AAA_{\SU}$ and $\hat\AAA^{~}_{\SU}$ induces the corresponding superspace of spinor fields $\psi$ and $\hat{\psi}$, respectively.
Whereas a spinor concerning $\hat\AAA^{~}_{\SU}$ (magnetic monopole) has not been observed experimentally, a dual spinor is introduced as a section in the spinor-gauge bundle.
The second term of Lagrangian (\ref{LYM2}) is kept in theory in this section.
This section treats the flat space-time with $\www=0$.

The magnetic-monopole Lagrangian is provided as
\begin{align*}
\LLL_\textrm{MP}&:=\left\{\bm{\bar{\hat{\psi}}}
\left(i\hbar\hspace{.1em}\hat{d}_{\SU}-\hat\mu\bm{1}_\SU\right)
{\bm{\hat\psi}}\right\}\vvv.
\end{align*}
The second Bianchi identities of gauge and its dual fields are provided as
\begin{align}
d^{~}_\SU\FFF^{~}_{\SU}=\hat{d}^{~}_{\SU}\hat\FFF^{~}_{\SU}=0.\label{dfdhf}
\end{align}
Equations of motion for gauge and spinor fields are provided as
\begin{subequations}
\begin{align}
&d_\SU\hat\FFF^{~}_{\SU}
={c^{~}_\SU}\hspace{.1em}\bm{\bar{\hat{\psi}}}\gamma_\Sp^\bcdot
{\bm{\psi}}\hspace{.2em}\VVV_\bcdot,&
\hat{d}_{\SU}\FFF_\SU={{c}^{~}_\SU}\hspace{.1em}\bm{\bar{\hat{\psi}}}
\gamma_\Sp^\bcdot{\bm{\hat\psi}}\hspace{.2em}\VVV_\bcdot,
\label{dfdf}\\
&\left(i\hbar\hspace{.1em}\dSg-
\mu\bm{1}_\SU\right)\bm{\psi}=0,
&\left(i\hbar\hspace{.1em}\hdSg-
\hat\mu\bm{1}_\SU\right){\bm{\hat\psi}}=0.\label{EoMEMS}
\end{align}
\end{subequations}
The equation of motion for ${\bm{\hat\psi}}$ is provided owing to a variational operation on $\I_\YM$ with respect to $\hat\AAA_{\SU}$.
The Bianchi identities (\ref{dfdhf}) and equations of motion (\ref{dfdf}) are generally independent from one another.
The equation of motion for the spinor section in the $\theta$-metric space has an expression using tensor coefficients such that:
\begin{align}
-\sqrt{\sigmath}\left(\delta^I_J\partial_\bcdot+
c^{~}_\SU f^I_{~JK}[\Aa^{~}_\theta]^K_\bcdot\right)
\left[\f_\theta^J\right]^{\bcdot a}\tau^{~}_I
=\left[j^{~}_\theta\right]^a,\label{EoMFtensor}
\end{align}
where 
\begin{align*}
\left[\f_\theta^J\right]^{ab}:=\eta_\theta^{a\bcdot}\eta_\theta^{b\star}
\left[\f_\theta^J\right]_{\bcdot\star}~~\textrm{and}~~
j^a:=
{c^{~}_\SU}\hspace{.1em}\bm{\bar{\hat{\psi}}}\gamma_\theta^a{\bm{\psi}}.
\end{align*}
We note that both $\hat\FFF^{~}_\SU$ and $\VVV_\bullet$ are defined using the Levi-Civita tensor; thus, it can be factored out after rearranging indices.
Therefore, the equation of motion for the spinor section is provided with the connection and curvature of the spinor-gauge bundle.

For the $U(1)$ gauge, the Bianchi identities correspond to the (so-called) first set of the Maxwell equations, and they are not equations but identities.
The second set of the Maxwell equations corresponds to (\ref{dfdf}) and govern the dynamics of electric and magnetic fields.
We note that the dual field $\hat\FFF^{~}_{\SU}$ is the electromagnetic field created by an electron field $\psi$ through the electromagnetic potential $\AAA_{\SU}$ in this study.

% Example
\begin{example}\textbf{electric and magnetic poles:}\label{EXcoulomb}
~\\ \noindent
This example treats a topological index of the spinor-gauge bundle with the $U(1)$-gauge group.
Spinor fields $\psi$ and $\hat\psi$, conventionally called an electron and a magnetic monopole, are referred to as \emph{electric pole} and \emph{magnetic pole} in this study, respectively.
We start a discussion from a Coulomb-type electric-potential of an electric pole in the flat space-time manifold with the $\theta$-metric.
When a Coulomb-type potential in globally flat $\MM$  is provided as 
\begin{align*}
V(r)=\frac{e^{i\pi\theta/2}}{r}
\end{align*}
using a local three-dimensional polar coordinate $\xi^a=(t,r,\theta_1,\theta_2)$.
a connection of the gauge bundle is obtained as
\begin{align}
\AAA^{~}_\Uo&:=V(r)\hspace{.2em}dt\in\Omega^1\left(\TsM\SM\{r=0\}\right).\label{AAAU1}
\end{align}
A curvature is obtained owing to the structure equation of abelian gauge-group such that: 
\begin{align}
\FFF^{~}_\Uo:=d\AAA^{~}_\Uo
=\left(\partial_\bcdot V(r)\right)d\xi^\bcdot{\wedge}dt
=\frac{e^{i\pi\theta/2}}{r^2}\hspace{.2em}dt\wedge dr.\label{FU1}
\end{align}
Thus, the dual curvature is provided as
\begin{align}
\hat\FFF^{~}_\Uo&=\frac{e^{i\pi\theta/2}}{\sqrt{-e^{i\pi\theta}}}\frac{1}{r^2}
(rd\theta_1)\wedge(r\sin{\theta_1}d\theta_2).\label{hFU1}
\end{align}
We note that
\begin{align*}
\left.\frac{e^{i\pi\theta/2}}{\sqrt{-e^{i\pi\theta}}}\right|_{\theta\rightarrow+0}=
\left.\frac{e^{i\pi\theta/2}}{\sqrt{-e^{i\pi\theta}}}\right|_{\theta\rightarrow1}=i.
\end{align*}
We treat the $\theta$-metric space with  $0\leq\theta\leq1$; thus, we write $e^{i\pi\theta/2}/\sqrt{-e^{i\pi\theta}}=i$, hereafter.

Tensor coefficient $\hat\f$ is defined in $\TsM$ except $r=0$ and is static; thus, the domain of $\hat\f$ is $\Sigma_{\hat\f}=\TsM\SM\{r=0\}\cong\R\otimes(\R^3\SM\{r=0\})$.
An equation of motion in the spinor-gauge bundle is provided as
\begin{align}
d\hat\FFF^{~}_\Uo=i \hspace{.1em}d\hspace{-.2em}
\left(\frac{1}{r^2}\hspace{.2em}(rd\theta_1)\wedge(r\sin{\theta_1}d\theta_2)\right)=
{c^{~}_\Uo}\hspace{.1em}\bm{\bar{\hat{\psi}}}\gamma_\Sp^\bcdot{\bm{\psi}}\hspace{.2em}\VVV_\bcdot.\label{ex314EMS}
\end{align}
This integration in three-dimensional volume $D^3$ is performed as
\begin{align}
\int_{D^3}d\hat\FFF^{~}_\Uo=
\int_{\partial\hspace{-.1em}D^3=S^2}\hat\FFF^{~}_\Uo=
i\int_{S^2}\sin{\theta_1}\hspace{.1em}d\theta_1\wedge d\theta_2=
4\pi i.\label{ex314EMSL}
\end{align}
The first equality is just a formal relation because $d\hat\FFF^{~}_\Uo=0$ in reality; this is special for the $U(1)$-gauge bundle.
We note that the covariant differential is just the external derivative for an abelian gauge group; thus, the equation of motion is equivalent to the Bianchi identity.

For the right-hand side of (\ref{ex314EMS}), we obtain a stationary electric-pole solution at the origin.
We integrate it as
\begin{align*}
{c^{~}_\Uo}\int_{D^3}
\bm{\bar{\hat{\psi}}}\gamma_\Sp^\bcdot{\bm{\psi}}\hspace{.2em}\VVV_\bcdot=
\int_0^1\hspace{-.3em}dr
\int_{0}^{\pi}\hspace{-.3em}rd\theta_1
\int_0^{2\pi}\hspace{-.3em}r\sin{\hspace{-.1em}\theta_1}d\theta_2\hspace{.3em}
j^0=4\pi{i},
\end{align*}
where $j^a$ is a static electron current given as 
\begin{align*}
&~
j^0=4\pi{i}\hspace{.1em}\delta^3(\bm{r})=
4\pi{i}\frac{\delta(r)\delta(\theta_1)\delta(\theta_2)}{r^2\sin{\theta_1}},\\
&~j^1=j^2=j^3=0;
\end{align*}
thus, it is consistent with (\ref{ex314EMSL}).
A topological invariant of the electric-pole current is provided as the Chern characteristic such that:
\begin{align*}
ch\left(\hat\FFF^{~}_\Uo\right)&=dim\left(U(1)\right)+c_1\left(\hat\FFF^{~}_\Uo\right)+
\frac{1}{2}\left(c_1\left(\hat\FFF^{~}_\Uo\right)^2-2c^{~}_2\left(\hat\FFF^{~}_\Uo\right)\right),
\end{align*}
where
\begin{align*}
dim\left(U(1)\right)&=1,\\
c_1\left(\hat\FFF^{~}_\Uo\right)&=\frac{i}{2\pi}\int_{S^2}\textrm{Tr}_\Uo[\hat\FFF^{~}_\Uo]=-2,\\
c^{~}_2\left(\hat\FFF^{~}_\Uo\right)&=
\frac{1}{8\pi^2}\int_{D^4}\left(
\textrm{Tr}_\Uo[\hat\FFF^{~}_\Uo\wedge\hat\FFF^{~}_\Uo]-
\textrm{Tr}_\Uo[\hat\FFF^{~}_\Uo]\wedge\Tr[\hat\FFF^{~}_\Uo]
\right)=0.
\end{align*}
We note that $\textrm{Tr}_\Uo[\hat\FFF^{~}_\Uo]=\hat\FFF^{~}_\Uo$ for the $U(1)$ gauge group.
Therefore, we obtain the Chern characteristic of the electric pole as $ch(\hat\AAA^{~}_\Uo)=1$ for any $\theta\in[0,1]$.

%%%
% Figure
\begin{figure*}[tb]
\centering
\includegraphics[width={3.5cm}]{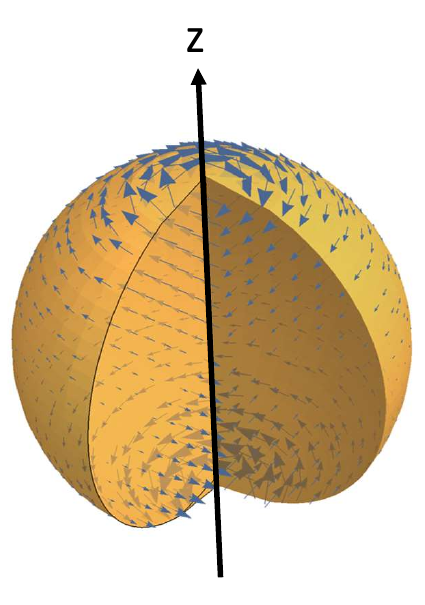}%
 \caption{
Arrows show three-dimensional vector representation of the dual connection for the Coulomb type static potential.
Vectors do not defined on a $z$-axis in this example.
} 
\label{fig0} 
\end{figure*}
The dual connection is provided as a solution of the structure equation for the given dual-curvature such that:
\begin{align*}
d\hat\AAA^{~}_\Uo
=\hat\FFF^{~}_\Uo=\frac{i}{r^2}
(rd\theta_1)\wedge(r\sin{\theta_1}d\theta_2).
\end{align*}
This equation has a unique solution with boundary condition of $\lim_{r\rightarrow\infty}\hat\Aa=0$ such as
\begin{align*}
\hat\AAA^{~}_{\Uo}=
-\frac{i}{r}\cot{\theta_1}\hspace{.2em}(r\sin{\theta_1}d\theta_2)=
-i\cos{\theta_1}d\theta_2,
\end{align*}
and it is represented using Cartesian coordinate $\xi^a=(t,x,y,z)$ as 
\begin{align*}
\hat\AAA^{~}_{\Uo}=i\hspace{.1em}\frac{y\hspace{.1em}dx-x\hspace{.1em}dy}{x^2+y^2}\hspace{.1em}\frac{z}{r}\in\Omega^1\left(\TsM\SM\{x=0\wedge y=0\}\right).
\end{align*}
This solution is known as the Dirac monopole, and a one-dimensional manifold subtracted from $\R^3$ is known as the Dirac string.
Three-dimensional vectors of connection $(\Aa_x,\Aa_y,\Aa_z)$ at $t=0$ are visualized in Figure \ref{fig0}.
We note that the one-dimensional singularity along the Dirac string is removable, and an essential singularity is a point at $x=y=z=0$.
In reality, the same solution is represented using a cylindrical coordinate $\xi_\textrm{cyl}:=(t,\rho,\phi,z) $ as 
\begin{align*}
\hat\AAA^{~}_{\Uo}=\frac{iz}{\sqrt{\rho^2+z^2}}\hspace{.1em}d\phi;
~&\textrm{thus,}~~
\lim_{\rho\rightarrow 0}\hat\AAA^{~}_{\Uo}\big|_{z\neq0}
\rightarrow{i\hspace{.1em}\textrm{sign}(z)d\phi}.
\end{align*}

Before calculating the topological index of the dual spinor-gauge bundle, we consider a relation between the spin-gauge connection and its dual; Figure \ref{figdiag} shows a schematic diagram of its structure.
We start from gauge connection $\AAA^{~}_\SU$; then, the structure equation concerning $\AAA^{~}_\SU$ gives gauge curvature $\FFF^{~}_\SU$.
The Hodge-dual of the gauge field (curvature) gives the equation of motion in the spinor-gauge bundle such as the left equation in (\ref{dfdf}).
We note that dual-curvature $\hat\FFF^{~}_\SU$ is not the curvature corresponding  the spinor-gauge connection, and spinor filed $\psi$ interacts with gauge connection $\AAA^{~}_\SU$.
When we start from the spinor current $j(\psi)$, a solution of equations (\ref{dfdf}), (\ref{EoMEMS}), and the Bianchi identity (\ref{dfdhf}) provides a set of fields $(\AAA^{~}_\SU,\hat\FFF^{~}_\SU)$; then, the gauge curvature $\FFF^{~}_\SU$ is provided as the Hodge-dual of $\hat\FFF^{~}_\SU$ that fulfils the structure equation with dual connection $\hat\AAA^{~}_\SU$.
The equation of motion of the (dual-)spinor field has a tensor-coefficient representation owing to the (dual-)connection and (dual-)curvature, respectively, e.g., as shown in (\ref{EoMFtensor}).
Dual connection $\hat\AAA^{~}_\SU$ is provided as a solution of the structure equation owing to the given dual curvature $\hat\FFF^{~}_\SU$, and the dual-spinor current $\hat{j}(\hat\psi)$ is given as a solution of the equation of motion and interacts with the dual connection.

%%%
% Figure
\begin{figure}[t]
\centering 
\includegraphics[width={11cm}]{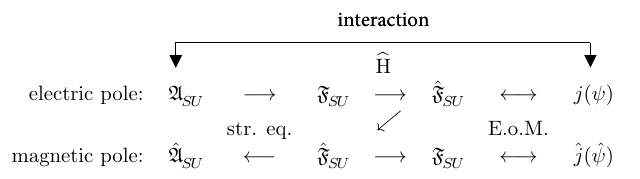}%
 \caption{
A diagram shows a relation between  spin-gauge bundle and its dual bundle. Abbreviations ``str. eq.'' and ``E.o.M.'' mean the structure function and the equation of motion, respectively.
} 
\label{figdiag} 
\end{figure}
%%%
In our case, a topological index of the Dirac monopole is provided owing to gauge connection $\AAA^{~}_\Uo$ given as (\ref{AAAU1}).
A direct calculation of the first Chern class concerning $\AAA^{~}_\Uo$ results indeterminate integration; thus, we exploit the indirect method here.
We note that $\AAA^{~}_\Uo$ is homotopically equivalent to $\hat\AAA^{~}_\Uo$.
When two connection are embedded in the five-dimensional smooth manifold such that:
\begin{align*}
\overline{\AAA}^{~}_\Uo&:=
s\hspace{.1em}{\AAA}^{~}_\Uo+
(1-s)\hspace{.1em}{\hat\AAA}^{~}_\Uo
\in(\Sigma\amalg\hat\Sigma)\otimes[0,1],
\end{align*}
where $\Sigma\cong\hat\Sigma\cong{S^2}$.
Moreover, $\overline{\AAA}^{~}_\Uo$ is the gauge connection at any $s\in[0,1]$ because the gauge group is abelian and the structure equation is a linear equation.
Therefore, two connections have the same Chern index as $ch(\FFF^{~}_\Uo)=ch(\hat\FFF^{~}_\Uo)=1$.
It is true  for any values of $\theta\in[0,1]$.

The stationary condition of the Yang--Mills action provides the equation of motion (\ref{dfdf}); thus, solutions are SD or ASD when the gauge group is $SU(N)$, and the metric is Euclidean.
In this example, the curvature is not SD nor ASD independent from the value of T because the gauge group is $U(1)$.
\QED
\end{example}
%%%

%
% index in gage field
%
\subsection{Index in gauge field}
The four-dimensional \emph{Euclidean} space-time has the SD- or ASD curvature as a solution of the equation of motion obtained as the stationary point of the Yang--Mills action as shown in \textbf{Section \ref{Gaugefield}}.
The instanton solution\cite{ATIYAH1978185}, is an example of the (A)SD solution, and the instanton number is a topological invariant in the four-dimensional Euclidean space-time.
This section treats the instanton solution in a flat space-time ($\www=0$) with a $(4\hspace{-.1em}\times\hspace{-.1em}4)$-matrix representation of the $\theta$-metric space.
The instanton number in the Minkowski space-time is also discussed in this section.

Coupling constant ${c^{~}_\SU}$ is set to unity for simplicity.
The $SU(2)$ generator is extended to a four dimensional vector by adding the zeroth component such that ${\bm{\tau}}\rightarrow{\bm{\tau}}_4:=(\bm{1}_2/2,-i\tau_1,-i\tau_2,-i\tau_3)$, where $\tau_i=\sigma_i/2$ is the $SU(2)$ generator.
A gauge-group element at $\xi\in\M_\theta$ has a representation such that:
\begin{align*}
g(\xi):=v(\xi)\hspace{.1em}\bm{\gamma}_{\theta}=
\sum_{a=0}^3 v^a(\xi)\gamma_{\theta}^a,
\end{align*}
that yields $\textrm{det}[g(\xi)]=\left(\eta^{~}_{\theta\hspace{.1em}\bcdots}\hspace{.1em}\xi^\bcdot\xi^\bcdot\right)^2$, where $v^\bullet(\xi)$ is a position vector in $\TM_\theta$.
A generator of the four-dimensional \emph{rotation} in the $\theta$-metric space is provided as $\Xi:=C_{\bcdots}\gamma^\bcdot_\theta\gamma^\bcdot_\theta$, where $C_\bullets\in\R$ are rotation parameters.
The group element squared is an invariant under the  the four-dimensional \emph{rotation} as
\begin{align*}
g'(\xi)g'(\xi)&=\Xi^{-1}g(\xi)\Xi^{-1}\Xi g(\xi)\Xi=\Xi^{-1}g(\xi)g(\xi)\Xi
=g(\xi)g(\xi),\\&=
\eta^{~}_{\theta\hspace{.1em}\bcdots}\hspace{.1em}\xi^\bcdot\xi^\bcdot\bm1_\Sp.
\end{align*}

An instanton connection at $\xi^a=(t,x,y,z)$ is provided as
\begin{align*}
\Aa_\theta&=\sum_{I=1}^4\Aa_\theta^I
\left(
    \begin{array}{cc}
     [{\bm\tau}_4]_I & 0\\
    0& [{\bm\tau}_4]_I
    \end{array}
  \right)
=\frac{r_\theta^2}{r_\theta^2-\lambda^2}g(\xi)dg^{-1}(\xi)),
\end{align*}
where 
\begin{align*}
g(\xi):=\frac{1}{r^{~}_\theta}\sum_{a=0}^3\xi^a\gamma_\theta^a,&~~\textrm{and}~~
r_\theta^2:=\eta^{~}_{\theta\hspace{.1em}\bcdots}\hspace{.1em}
\xi^\bcdot\xi^\bcdot=e^{i\pi\theta}t^2-x^2-y^2-z^2.
\end{align*}
An instanton connection is provided using a chiral representation of $\gamma^i_\theta$ as
\begin{align*}
[\Aa_\theta]_\bcdot \eee_\theta^\bcdot&=
\frac{1}{r^2_\theta-\lambda^2}\left(
    \begin{array}{cc}
     [\aa_\theta]_\bcdot \eee_\theta^\bcdot& 0\\
    0& [\tilde\aa_\theta]_\bcdot \eee_\theta^\bcdot
    \end{array}
  \right),
\end{align*}
where
\begin{align*}
&~\eee^a_\theta:=(dt_\theta,dx,dy,dz),~~~t_\theta:=ie^{i\pi\theta/2}\hspace{.1em}t,\\
&~
\begin{array}{lcclcc}
[\aa_\theta]_0&=i
\left(
    \begin{array}{cc}
     z & x-iy\\
    x+iy& -z
    \end{array}
  \right),&
[\tilde\aa_\theta]_0&=-i
\left(
    \begin{array}{cc}
     z & x-iy\\
    x+iy& -z
    \end{array}
  \right),\\
{[\aa_\theta]}_1&=
\left(
    \begin{array}{cc}
     -iy & -it_\theta+z\\
    -it_\theta-z&iy
    \end{array}
  \right),&
[\tilde\aa_\theta]_1&=
\left(
    \begin{array}{cc}
     -iy & it_\theta+z\\
    it_\theta-z&iy
    \end{array}
  \right),\\
{[\aa_\theta]}_2&=
i\left(
    \begin{array}{cc}
     x & it_\theta-z\\
    -it_\theta-z&-x
    \end{array}
  \right),&
[\tilde\aa_\theta]_2&=
i\left(
    \begin{array}{cc}
     x & -it_\theta-z\\
    it_\theta-z&-x
    \end{array}
  \right),\\
{[\aa_\theta]}_3&=
\left(
    \begin{array}{cc}
    -it_\theta&-x+iy\\
    x+iy&it_\theta
    \end{array}
  \right),&
[\tilde\aa_\theta]_3&=
\left(
    \begin{array}{cc}
    it_\theta&-x+iy\\
    x+iy&-it_\theta
    \end{array}
  \right).
\end{array}
\end{align*}
The instanton has two sets of $SU(2)$ connections according to $\aa_\theta$ and $\tilde\aa_\theta$ such that:
\begin{equation*}
\AAA_\theta^I=\left[\Aa_\theta^I\right]_\bcdot dx^\bcdot=
\left(\frac{2}{r_\theta^2-\lambda^2}\right)\times\left\{
\begin{array}{cl}
0,&(I=0),\\
-x\hspace{.1em}dt_\theta+t_\theta\hspace{.1em}dx
+z\hspace{.1em}dy-y\hspace{.1em}dz,&(I=1),\\
-y\hspace{.1em}dt_\theta-z\hspace{.1em}dx
+t_\theta\hspace{.1em}dy+x\hspace{.1em}dz,&(I=2),\\
-z\hspace{.1em}dt_\theta+y\hspace{.1em}dx
-x\hspace{.1em}dy+t_\theta\hspace{.1em}dz,&(I=3),
\end{array}
\right.
\end{equation*} 
and
\begin{equation*}
\tilde\AAA_\theta^I=\left[\tilde\Aa_\theta^I\right]_\bcdot dx^\bcdot=
\left(\frac{2}{r_\theta^2-\lambda^2}\right)\times\left\{
\begin{array}{cl}
0,&(I=0),\\
x\hspace{.1em}dt_\theta-t_\theta\hspace{.1em}dx
+z\hspace{.1em}dy-y\hspace{.1em}dz,&(I=1),\\
y\hspace{.1em}dt_\theta-z\hspace{.1em}dx
-t_\theta\hspace{.1em}dy+x\hspace{.1em}dz,&(I=2),\\
z\hspace{.1em}dt_\theta+y\hspace{.1em}dx
-x\hspace{.1em}dy-t_\theta\hspace{.1em}dz,&(I=3),
\end{array}
\right.
\end{equation*} 
The instanton connection in Eudlidean space is provided as $\Aa_E:=\Aa_\theta\big|_{\theta\rightarrow1}$.

Instanton curvatures, $\FFF_\theta$ and $\tilde\FFF_\theta$, are induced owing to the structure equation with $\AAA_\theta$ and $\tilde\AAA_\theta$, respectively; Euclidean and Minkowski curvatures are provided from instanton curvature in the $\theta$-metric space such as 
\begin{align*}
\FFF^{+}_E&=\FFF^{~}_\theta\big|_{\theta\rightarrow 1},\hspace{1.3em}
\FFF^{-}_E~=~\tilde\FFF_\theta\big|_{\theta\rightarrow 1},\\
\FFF^{+}_M&=\FFF^{~}_\theta\big|_{\theta\rightarrow+0},\hspace{.5em}
\FFF^{-}_M~=~\tilde\FFF^{~}_\theta\big|_{\theta\rightarrow+0}.
\end{align*}
They yield $\hat\FFF^\pm_E=\pm\FFF^\pm_E$ and $\hat\FFF^\pm_M=\pm\FFF^\pm_M$; thus, instanton solutions in the $\theta$-metric space provide the SD and ASD solutions simultaneously for the both Euclidean and Minkowski limit.
Moreover, it is true for any $\theta\in[0,1]$ that 
\begin{align}
\HD(\FFF^{~}_\theta)=+\FFF^{~}_\theta~~\textrm{and}~~
\HD(\tilde{\FFF}^{~}_\theta)=-\tilde\FFF^{~}_\theta.\label{HFeqF}
\end{align}
For Eudlidean-metric space, according to a sign of the second Chern class, the SD- or ASD curvature gives a solution of the Yang--Mills equation as shown in \textbf{Remark \ref{SDorASD}}.
On the other hand, in the Minkowski space-time, the SD- and ASD curvature do not always provide the extremal of the Yang--Mills action.
We confirm that instanton curvature ${\FFF}^{\pm}_M$ in the Minkowski space-time is a solution of the Yang--Mills equation:
The Yang--Mills equation in the $\theta$-metric space has a tensor component representation as (\ref{EoMFtensor}).
For the SD curvature, the equation can be written using a half of components as
\begin{align*}
&[\textrm{Y-M eqs}]_\theta:=\\&\left\{
\begin{array}{rcl}
\left(
\nabla_{\hspace{-.2em}1;I}\f^I_{01}+
\nabla_{\hspace{-.2em}2;I}\f^I_{02}+
\nabla_{\hspace{-.2em}3;I}\f^I_{03}\right)
\sigmath^{\hspace{-.2em}-\frac{1}{2}}dt
&=&\bm{\bar{\hat{\psi}}}\left(e^{i\pi\theta/2}\gamma_\Sp^0\right){\bm{\psi}}\hspace{.1em}\sigmath^{\hspace{-.2em}-\frac{1}{2}}\hspace{.2em}dt,\\
\left(-\sigmath^{\hspace{-.2em}-\frac{1}{2}}\hspace{.1em}
\nabla_{\hspace{-.2em}0;I}\f^I_{01}-
\nabla_{\hspace{-.2em}3;I}\f^I_{02}+
\nabla_{\hspace{-.2em}2;I}\f^I_{03}\right)dx
&=&\bm{\bar{\hat{\psi}}}\gamma_\Sp^1{\bm{\psi}}\hspace{.2em}dx,\\
\left(
\nabla_{\hspace{-.2em}3;I}\f^I_{01}-\sigmath^{\hspace{-.2em}-\frac{1}{2}}\hspace{.1em}
\nabla_{\hspace{-.2em}0;I}\f^I_{02}-
\nabla_{\hspace{-.2em}1;I}\f^I_{03}\right)dy
&=&\bm{\bar{\hat{\psi}}}\gamma_\Sp^2{\bm{\psi}}\hspace{.2em}dy,\\
\left(-\nabla_{\hspace{-.2em}2;I}\f^I_{01}+
\nabla_{\hspace{-.2em}1;I}\f^I_{02}-\sigmath^{\hspace{-.2em}-\frac{1}{2}}\hspace{.1em}
\nabla_{\hspace{-.2em}0;I}\f^I_{03}\right)dz
&=&\bm{\bar{\hat{\psi}}}\gamma_\Sp^3{\bm{\psi}}\hspace{.2em}dz,
\end{array}
\right.%\label{SDYMeq}
\end{align*}
where
\begin{align*}
\nabla_{a;J}:=\left(\delta^I_J\partial_a+
f^I_{~JK}[\Aa^{~}_\theta]^K_a\right)\tau^{~}_I.
\end{align*}
The SD Yang--Mills equations in Euclidean and Minkowski metrics are provided as
\begin{align*}
[\textrm{Y-M eqs}]_E:=\lim_{\theta\rightarrow1-0}[\textrm{Y-M eqs}]_\theta,
~~&\textrm{and}~~
[\textrm{Y-M eqs}]_M:=\lim_{\theta\rightarrow+0}[\textrm{Y-M eqs}]_\theta.
\end{align*}
Accordingly, equations have a simple relation to each other such that:
\begin{align*}
t\mapsto-it:[\textrm{Y-M eqs}]_E\mapsto[\textrm{Y-M eqs}]_M;
\end{align*}
thus, the connection with the Minkowski  metrics is provided from the Eudlidean metric as
\begin{align*}
t\mapsto-it:\AAA_E^I\mapsto\AAA_M^I.
\end{align*}
Therefore, the Minkowski metric's SD connection is a solution to the Yang-Mills equation in the Minkowski space-time.
It is also true for the ASD connection. 

$\M_1$ is constructed owing to the one-point compactification of $\M_E$ as follows:
For Euclidean manifold $\M_E$, variable transformations $r_E\mapsto\lambda\rho$ for $0<r_E\leq\lambda$ and $r_E\mapsto\lambda/\rho$ for $\lambda<r_E<\infty$ are applied; thus, two compact manifolds $D^4$ are obtained, namely $D^4_\textrm{in}$ and $D^4_\textrm{out}$, respectively.
Then, surfaces of $D^4_\textrm{in}$ and $D^4_\textrm{out}$ are equated to each other with keeping three angular variables; thus, $\M_E$ is compacitified into union of two compact manifolds such that $\M_E\mapsto\M_1=D_\textrm{in}^4\cup D_\textrm{out}^4$ with $D^4_\textrm{in}\cap D^4_\textrm{out}=\partial{D^4_\textrm{in}}=-\partial{D^4_\textrm{out}}= S^3$ and $\partial(D^4_\textrm{in}{\cup}D^4_\textrm{out})=\emptyset$.
The first Pontrjagin number is obtained as an integration of the Pontrjagin class in compactified volumes such that:
\begin{align*}
p_1(\M_1)=-\left(
\int_{D^4_\textrm{in}}\left.
\frac{\Tr_\SU\left[\FFF^{~}_1\wedge\FFF^{~}_1\right]}
{8\pi^2}\right|_{r_E\rightarrow\lambda \rho}+
\int_{D^4_\textrm{out}}\left.
\frac{\Tr_\SU\left[\FFF^{~}_1\wedge\FFF^{~}_1\right]}
{8\pi^2}\right|_{r_E\rightarrow\lambda/\rho}
\right),
\end{align*}
where
\begin{align*}
\FFF^{~}_1=\FFF^{~}_E&=:\left.\frac{1}{2}\left[\f_E^{~}\right]_{\bcdots}
\eee_\theta^\bcdot\wedge\eee_\theta^\bcdot\right|_{\theta\rightarrow1}.
\end{align*}
A trace of an instanton curvature squared is obtained such that:
\begin{align}
\Tr_\SU\left[\FFF^{~}_1\wedge\FFF^{~}_1\right]&=
\sum_{I=1}^4\epsilon_E^{\bcdots\bcdots}\frac{1}{2}
\left[\f_E^I\right]_{\bcdots}\frac{1}{2}\left[\f_E^I\right]_{\bcdots}
\hspace{.2em}\Tr_\TxT
\left[\left[\textbf{t}_4\right]_I.\left[\textbf{t}_4\right]_I\right]\vvv_\theta,\label{TrFF}\\
&=-48\left(\frac{\lambda}{r_E^2+\lambda^2}\right)^4\vvv_\theta,\nonumber
\end{align}
thus, $p_1(\M_1)$ is calculated as 
\begin{align*}
p_1(\M_1)&=
\frac{6}{\pi^2}
\int_0^1\frac{\rho^3}{(\rho^2+1)^4}
d\rho\hspace{.2em}\int_{S^3}d\Omega_4
~=~1,
\end{align*}
where $d\Omega_4$ is an integration measure for an angular integration on $S^3$, which gives $\int_{S^3}d\Omega_4=2\pi^2$.
The completely anti-symmetric tensor in Eudlidean space is provided as  $[\bm\epsilon_E]^{0123}=[\bm\epsilon_E]_{0123}=1$.

On intersection $D^4_\textrm{in}\cap D^4_\textrm{out}=S^3$, angular variables coincide as
\begin{align*}
\partial D_\textrm{in}\ni\theta_i^\textrm{in}=\theta_i^\textrm{out}\in-\partial D_\textrm{out} ,
\end{align*}
for $i=1,2,3$ with keeping an orientation on their surfaces.
When $\theta_3^\textrm{in}$ maps to $\theta_3^\textrm{out}$ as $n\times\hspace{.2em}\theta_3^\textrm{in}=\theta_3^\textrm{out}$ with $0<n\in\Z$, an instanton number is given as $p_1(\M_1)=n$ due to the $\theta_3$ integration from $0$ to $2n\pi$, which corresponds to the winding number around infinity.

The first Pontrjagin number is provided in the entire space of $W_\theta:={\M_\theta}\otimes(\theta\in[0,1])$.
After the same compactification as in a Euclidean space in $W_\theta$, a four-dimensional polar coordinate is introduced as
\begin{align*}
\begin{array}{ll}
\hspace{.2em}t:=\lambda\rho \cos{\theta_1},&
x:=\lambda\rho \sin{\theta_1}\cos{\theta_2},\\
y:=\lambda\rho \sin{\theta_1}\sin{\theta_2}\cos{\theta_3},&
z:=\lambda\rho \sin{\theta_1}\sin{\theta_2}\sin{\theta_3}.
\end{array}
\end{align*}
The first Pontrjagin number is provided by integrations such that:
\begin{align*}
p_1(\M_\theta)&=
-i\frac{6}{\pi^2}e^{i\pi\theta/2}\int_0^1d\rho\int_0^\pi
\rho^2\sin^2{\theta_1}\hspace{.1em}d\theta_1\int_0^\pi\rho\sin{\theta_2}
\hspace{.2em}d\theta_2\int_0^{2\pi} d\theta_3
\\ 
&\hspace{-2em}\left[
\left\{1+\rho^2\left(1-\left(1+e^{i\pi\theta}\right)\cos^2{\theta_1}\right)\right\}^{-4}
%\right.\\&\hspace{2.em}+\left.
+
\left\{\left(1+e^{i\pi\theta}\right)\cos^2{\theta_1}-\left(1+\rho^2\right)\right\}^{-4}
\right],\\&=1,
\end{align*}
for $0<\theta\leq1$.
We note that $\epsilon_E^{\bcdots\bcdots}$ in (\ref {TrFF}) is replaced by $\sigmath\hspace{.1em}\epsilon_\theta^{\bcdots\bcdots}$ in the $\theta$-metric space. 
The Pontrjagin class with the Minkowski metric is a divergent integral; thus, it is not well-defined.
A domain of the instanton solution in the inertial bundle with the Minkowski metric is not the entire space of $\M_0$ but is $\M_0\hspace{-.3em}\setminus\hspace{-.3em}\M_\lambda$, where $\M_\lambda:=\{\xi\in\M_0\mid{\eta_\bcdots\xi^\bcdot\xi^\bcdot-\lambda^2=0}\}$; thus, manifold $\M_0$ is not closed as a principal bundle with the instanton connection.
Taking the closure of manifold $\M_0$ with respect to $\theta$ from the right ($\theta=+0$), $\M_0$ has the first Pontrjagin number as $p_1(\M_\theta)|_{\theta\rightarrow+0}=1$.
The instanton solution in the Minkowski metric is an example of \textbf{Theorem \ref{EMcobor}}: manifolds $\M_0$ and $\M_1$ are cobordant to each other, and the first Pontrjagin number of $\M_0$ is consistently defined owing to that of $\M_1$.

The SD- and ASD-instanton-curvatures fulfil the Bianchi identity as
\begin{align*}
d\FFF_\theta^\pm+\left[\FFF_\theta^\pm,\AAA_\theta^\pm\right]_\wedge=0.
\end{align*}
When we take the SD- or ASD solution as the principal curvature of the gauge bundle, we have
\begin{align*}
\FFF_\SU^{~}=\FFF_\theta^\pm=\pm\HD(\FFF_\theta^\pm)=\pm\hat\FFF_\SU^{~};
\end{align*}
thus, the equation of motion is provided from (\ref{DiracEq2}) as
\begin{align*}
{d}\hat\FFF_\SU^{~}+
\left[\hat\FFF_\SU^{~},\AAA_\SU^{~}\right]_\wedge=
\pm{d}\FFF_\SU^{~}\pm
\left[\FFF_\SU^{~},\AAA_\SU^{~}\right]_\wedge=
\pm\bar\psi_\textrm{inst}\hspace{.1em}\gamma_\Sp^\bcdot\psi_\textrm{inst}
\hspace{.2em}\VVV_\bcdot=0,
\end{align*}
where $\psi_\textrm{inst}$ is an instanton spinor-field.
Last equality is due to the Bianchi identity.
This is true for any $\theta\in(0,1]$.
Therefore, the source field does not exist in a domain of the instanton curvature.
%%%
% Figure
\begin{figure}[t]
\centering
\includegraphics[width={5cm}]{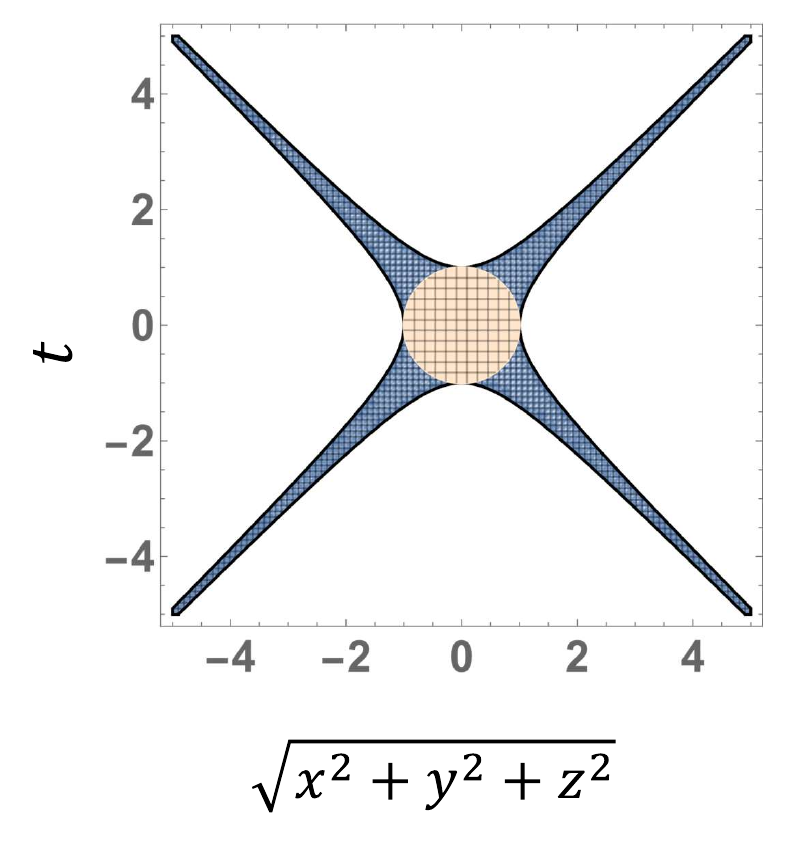}%
 \caption{
This figure shows regions where $\|\FFF^{~}_\SU\|>1$ for the Eudlidean metric (a circle at the origin) and for the Minkowski metric (a mashed region along diagonal lines).
Here, we set $\lambda=1$ in $\FFF^{~}_\SU$.
} 
\label{figins} 
\end{figure}
%%%

A physical interpretation of an instanton in the Minkowski-bundle is as follows:
The SD- and ASD-connections domain with the Minkowski metric is $\M_0\hspace{-.3em}\setminus\hspace{-.3em}\M_\lambda$; thus the domain of $\psi_\textrm{inst}$ is $\M_\lambda$, which is on the light-cone at far from the origin, as shown in Figure \ref{figins}. 
After quantization of spinor field $\psi_\textrm{inst}$, a quantity 
\begin{align*}
p_\textrm{inst}\left(\Sigma_3\right):=\int_{\Sigma_3}\bar\psi_\textrm{inst}\hspace{.1em}\gamma_\Sp^\bcdot\psi_\textrm{inst}\hspace{.2em}\VVV_\bcdot
\end{align*}
may give a provability to observe an instanton in three-dimensional volume $\Sigma_3\in\M_\lambda$ after appropriate normalization; thus, an instanton is localized in the light-cone and is interpreted as a massless particle with an $SU(2)$-charge.
Conservation of the Minkowski-instanton number is ensured as the topological invariant (the first Pontrjagin number) owing to cobordism.

%
% index in chiral representation
%
\subsection{Index in chiral representation}
This section discusses the topological index of the Yang--Mills bundle using a new $2N$-dimensional vector representation, namely a chiral representation.
A spinor-gauge bundle has a $\Z_2$-grading total space as $\S_\Sp^{\otimes N}=({\S_\Sp^+})^{\otimes N}\oplus({\S_\Sp^-})^{\otimes N}$ owing to projection operator $P^\pm_\Sp$ as discussed  in section \ref{spinerbundle}. 
The chiral representation of spinors and covariant differential operator are introduced as
\begin{align*}
\bm{\D}_{ch}&:=\gamma_\ym\hspace{.1em}
\ds_\SG\left(
    \begin{array}{cc}
   P^+_\Sp\otimes\bm{1}_\SU & 0\\
    0& P^-_\Sp\otimes\bm{1}_\SU
    \end{array}
  \right)=\left(
    \begin{array}{cc}
    0& \dSgm\otimes\bm{1}_\SU\\
    \dSgp\otimes\bm{1}_\SU & 0
    \end{array}
  \right),\\
\bm{\Psi}_{ch}&:=\left(
    \begin{array}{cc}
     \bm{\psi}^+ \\ \bm{\psi}^-
    \end{array}
\right),
\end{align*}
where $\dSgpm:=\dSg{\bcdot}P^\pm_\Sp$.
A Dirac conjugate of $\bm{\Psi}_{ch}$ is provided as
\begin{align*}
\bm{\overline\Psi}_{ch}:=\bm{\Psi}^\dagger_{ch}\gamma^0\bm{\gamma}_{\ym}
=\left(\bm{\bar{\hat{\psi}}}^-,\bm{\bar{\hat{\psi}}}^+\right).
\end{align*}
In this representation, the Lagrangian of the spinor field is provided as
\begin{align*}
\LLL_{ch}:=-\left(\hspace{.1em}\overline{\bm{\Psi}}_{ch}\hspace{.2em}i\bm{\D}_{ch}\bm{\Psi}_{ch}+
\frac{\mu}{\hbar}\hspace{.2em}\overline{\bm{\Psi}}_{ch}\bm{\Psi}_{ch}\right)\vvv,
%\left(
%    \begin{array}{cc}
%    \bm{1}_\SU\otimes m_e & 0\\
%    0& -\bm{1}_\SU\otimes m_e
%    \end{array}
%  \right)
\end{align*}
which provides the same equation of motion as those from the matter-field Lagrangian in the Yang--Mills bundle.

A kinetic term of the matter Lagrangian denoted by $K_\SG:=i\bar\psi\hspace{.1em}\ds\psi$ is Hermitian self-conjugate as follows:
For the Dirac operator given in (\ref{DSG}), its Hermitian conjugate is provided that
\begin{align}
K_\SG^\dagger&=
\left(i\hspace{.1em}\bar\psi(\iota_{{\gamma}_\Sp}d)\psi\right)^\dagger
=-i\hspace{.1em}\left(\partial_\star\psi^\dagger\right)(\gamma_\Sp^\star)^\dagger(\gamma_\Sp^0)^\dagger \psi
=-i\hspace{.1em}\left(\partial_\star\bar\psi\right)\gamma_\Sp^\star\psi,\label{Ksp}\\
&=i\hspace{.1em}\bar\psi\gamma_\Sp^\star(\partial_\star\psi)
-i\hspace{.1em}\left(\partial_\star\bar\psi\gamma_\Sp^\star\psi\right),\nonumber
\end{align}
and the last term is vanished owing to the current conservation; thus, $K_\SG$ is a Hermitian self-conjugate as $K_\SG=K_\SG^\dagger$.
Moreover, $\mathrm{Ker}(\dSgm)^\dagger\cong\mathrm{Coker}(\dSgp)=\S^-/\mathrm{Im}(\D^+)$ is provided owing to $(\ref{Ksp})$ and
\begin{align*}
i\hspace{.1em}\bar\psi^+(\iota_{{\gamma}_\Sp}d)\psi^-=
-i\hspace{.1em}\left(\partial_\star\bar\psi^+\right)\gamma_\Sp^\star\psi^-=
-i\hspace{.1em}\bar\psi^+(\overleftarrow{\iota_{{\gamma}_\Sp}d})\psi^-,
\end{align*}
where $(\overleftarrow{\bullet})$ is a differential operator acting from the left. 
Operator $\dSg$ is the Dirac operator; thus the Atyah--Singer index theorem (\textbf{Remark \ref{ASI}}) is maintained in the spinor-gauge bundle after appropriate compactification of base manifold $\M_E$.

The topological index with the Minkowski metric is defined owing to cobordism through the $\theta$-metric space.
The Dirac equation is provided in the $\theta$-metric space as 
\begin{align*}
{\D}_{\hspace{-.1em}\theta}\bm{\psi}&:=\iota_{\bm{\gamma}_\theta}(d-i\cG\hspace{.1em}\wcs_{\hspace{-.1em}\theta})\bm{\psi}\otimes\bm{1}_\SU
-{i}{c^{~}_\SU}\hspace{.1em}\iota_{\bm{\gamma}_\theta}\AAA^{~}_\SU\hspace{.2em}\bm{\psi},
\intertext{where }
\wcs_{\hspace{-.1em}\theta}&:=
\frac{i}{2}\eta^{~}_{\theta\hspace{.1em}\bcdots}\eta^{~}_{\theta\hspace{.1em}\stars}
\Varepsilon^\mu_a\omega^{~\bcdot\star}_{\mu}\hspace{.1em}
{\Ss}^{\bcdot\star}_{\hspace{-.1em}\theta}.
%~~\textrm{and}~~
%{\Ss}_{\hspace{-.1em}\theta}^{ab}:=\frac{i}{4}[\gamma_{\hspace{-.1em}\theta}^a,\gamma_{\hspace{-.1em}\theta}^b].
\end{align*}
Spinor bundle is formally extended to the $\theta$-metric space as $\left(\S^\theta_\Sp,\pi^\theta_\Sp,\M_\theta,Spin_\theta\right)$, where $\M_\theta$ is a Riemannian manifold with metric tensor $\bm\eta^{~}_\theta$ defined in (\ref{Cmetric}) and $\S^\theta_\Sp$ is a Clifford module with the Clifford algebra $\bm{\gamma}^{~}_\theta$ defined in (\ref{gamma-theta}).
%The axial anomaly is considered in it.
At $\theta\hspace{-.2em}=\hspace{-.2em}+0$ and $\theta\hspace{-.2em}=\hspace{-.2em}1$, the $\theta$-metric gives Minkowski and Euclid spinor bundles with $Spin_\theta|_{\theta\rightarrow+0}=Spin(1,3)$ and $Spin_\theta|_{\theta\rightarrow1}=Spin(4)$, respectively.
Dirac spinor $\psi^{~}_\theta$ is presented as a section in the $\theta$-metric spinor bundle; thus, the chiral representation introduced above is extended in the $\theta$-metric spinor bundle.

First, the topological index of Eudlidean space-time is considered in the $\theta$ metric space at $\theta=1$.
Dirac operator $\bm{\D}^E_{ch}:=\bm{\D}^\theta_{ch}|_{\theta\rightarrow1}$ has the topological index introduced in section \ref{AppIndex}, where $\bm{\D}^\theta_{ch}$ is a Dirac operator in the $\theta$-metric extended chiral space.
Euclidean base-manifold $\M_E$ is assumed to be compactified appropriately.
Owing to the McKean--Singer theorem (given as \textbf{Remark \ref{McS}} in section  \ref{AppIndex}), the index of Dirac operator $\bm{\D}_{ch}^E$ is provided using the super trace such that:
\begin{align*}
ind\left(i\bm{\D}_{ch}^E\right)&=
Str\left(e^{-t\hspace{.2em}{\D}_\SG^E\hspace{.2em}{\D}_\SG^E}\right)
=\Tr\left[\Gamma_{\hspace{-.1em}E}\hspace{.2em}
e^{-t\hspace{.2em}{\D}_\SG^E\hspace{.2em}{\D}_\SG^E}\right],\\&=
\Tr\left[e^{-t\hspace{.2em}{{\D}_\SG^E}^{\hspace{-.1em}-}
\hspace{.2em}{{\D}_\SG^E}^{\hspace{-.1em}+}}\right]-
\Tr\left[e^{-t\hspace{.2em}{{\D}_\SG^E}^{\hspace{-.1em}+}
\hspace{.2em}{{\D}_\SG^E}^{\hspace{-.1em}-}}\right],\\
&=\int_{\M_E}\hat{A}(\RRR)ch(\hat\FFF^{~}_\SU).
\end{align*}
where ${\D}_\SG^{E}$ is the Dirac operator in the spinor-gauge bundle in Eudlidean space-time and ${{\D}_\SG^{E}}^\pm$ are its projections owing to the projection operators $P^\pm_E$, that is introduced in section \ref{themetric}.
Operator ${\D}_\SG^{E}\hspace{.2em}{\D}_\SG^{E}$ is an elliptic type; thus, $e^{-t\hspace{.2em}{\D}_\SG^E\hspace{.2em}{\D}_\SG^E}$ is a heat kernel.
The last equality is owing to the Atiyah--Singer index theorem given as \textbf{Remark \ref{ASI}} in section \ref{AppIndex}.

Second, we discuss the topological index with the Minkowski metric, which is not well-defined because operator $\D_\SG\hspace{.2em}\D_\SG$ is a hyperbolic type, and the Atiyah--Singer index theorem is not fulfilled.
Although the index theorem is not applied as its original shape, the topological index is defined in the Minkowski space-time by \textbf{Definition \ref{EMcobor}}.
Minkowski manifold $\M_M:=\left.\M_\theta\right|_{\theta\rightarrow+0}$ and Euclidean manifold $\M_E:=\left.\M_\theta\right|_{\theta\rightarrow1}$ are cobordant to each other with respect to manifold $W_\theta:=\M_\theta\otimes(\theta\in[0,1])$.
Therefore, the topological index in $\M_M$ is well-defined as
\begin{align*}
ind\left(i\bm{\D}_\Sp\right):=
ind\left(i\bm{\D}_{ch}^E\right)
=\int_{\M_E}\hat{A}(\RRR)ch(\hat\FFF^{~}_\SU)
\Rightarrow\int_{\M_E}ch(\hat\FFF^{~}_\SU).
\end{align*}
We note that the second Chern class is zero for know solutions of the Einstein equation; thus,  the index is given only by curvature $\hat\FFF^{~}_\SU$.
The index in the Minkowski space-time is formally denoted as
\begin{align*}
ind\left(i\bm{\D}_\Sp\right)&=
\Tr\left[\Gamma_\Sp\hspace{.2em}
e^{-t\hspace{.2em}{\D}_\SG\hspace{.2em}{\D}_\SG}\right],
\end{align*}
where $\Gamma_\Sp=\gamma^5$ in the standard notation, and it is equivalent to the standard representation of the axial anomaly\cite{PhysRevLett.42.1195,bertlmann2000anomalies}.

%%% ----------------------------------------------------------------------
% Summary and discussions
%%% ----------------------------------------------------------------------
\section{Summary}
\begin{table}[t]
\begin{tabular}{c|cccc|l}
bundle name&bundle&conn.&curv.&sect.&physical object\\
\hline
co-Poincar\'{e}&$\left(\M,\pi_{\mathrm L},\MM,G_\cP\right)$&$\www~~$&$\RRR~~$&$\eee$&gravity\\
%%%
Spinor&$\left(\S_\Sp,\pi_\Sp,\M, Spin(1,3)\right)$&$\AAA^{~}_\Sp$&$\FFF^{~}_\Sp$&$\psi$&matter\\
%%%
Gauge&$\left(\TsM,\pi_{\SU},\M,SO(1,3)\otimes{SU(N)}\right)$&
$\AAA^{~}_\SU$&$\FFF^{~}_\SU$&$\bm\phi$&force, Higgs\\
%%%
SG&$\left(\S_\Sp^{\otimes N},\pi_{\Sp},\M, Spin(1,3)\otimes{SU(N)}\right)$&
$\AAA^{~}_\Sp$&$\FFF^{~}_\Sp$&$\bm{\psi}$&force, matter\\
%%%
Dual SG&$\left(\hat\S_\Sp^{\otimes N},\hat\pi_{\Sp},\M, Spin(1,3)\otimes{SU(N)}\right)$&
$\hat\AAA^{~}_\Sp$&$\hat\FFF^{~}_\Sp$&${\bm{\hat\psi}}$&force, dual-matter\\
%%%
Yang--Mills&\multicolumn{3}{r}{\emph{tensor products of SG and dual SG bundles}}&~&force, matter\\
\hline
\end{tabular}
\vskip 2mm
\caption{
Principal bundles introduced in this study.
In a bundle name, ``SG'' stands for ``spinor-gauge''.
}\label{bundletbl}
\end{table}
This study treated the classical Yang--Mills theory and general relativity in four-dimension.
The Chern--Weil theory concerning several principal bundles twisting to each other provides a common mathematical language to treat physical objects appearing in the Yang--Mills theory in curved space-time.
In reality, we introduced five principal bundles listed in \textbf{Table \ref{bundletbl}}.
The co-Poincar\'{e} bundle is a model of our Universe, which provides the platform of the Yang--Mills theory.
Space-time manifold $\MM$ with, in general, non-zero curvature is a smooth Riemannian four-dimensional manifold.
At each point of $\MM$, inertial manifold $\M$ with vanishing Levi-Civita connection $\Gamma^\lambda_{~\mu\nu}$ is associated as the inertial space owing to Einstein's equivalent principle.
inertial manifold $\M$ has vierbain form $\eee$ as a section that determines a metric of $\M$ and connection $\www$ to govern gravitational interaction among any energetic objects.
Other bundles provide matter and force fields on the base manifold $\M$ as sections and connections, respectively, and are twisted with the co-Poincar\'{e} bundle.

The Yang--Mills theory is represented in the spinor-gauge bundles on the inertial manifold with various structure groups of the gauge symmetry.
Connection and curvature of the bundle belonging to the fundamental representation of the structure group represent a potential function and a force field in physics, e.g., $U(1)$-gauge group as the electromagnetic force, $SU(2)$ as the weak force, and $SU(3)$ as the strong force, in physics.
The matter field is a spinor-valued section in total space $\S_\Sp$ of the spinor-gauge bundle and interacts with each other through the gauge connection.
A scalar field in the gauge bundle with the structure group of $SU(2)$ is called the Higgs field in physics, which is a section of the $SU(2)$-gauge bundle and belongs to the fundamental representation of the $SU(2)$ symmetry.
The equation of motion for the spinor fields consists of the Dirac operator; thus, the Atiyah--Singer index theorem provides the topological invariant in the spinor-gauge bundle.

The Hodge-dual of curvature forms are indispensable objects to constructing Yang-Mills theory; a force field given by curvature has a constraint owing to the Bianchi identity $d^{~}_G\FFF^{~}_G=0$; thus, a gauge curvature $\FFF^{~}_G$ does not have dynamics in the spinor-gauge bundle.
On the other hand, the dual curvature is not a curvature in the spinor-gauge bundle and has an equation of motion $d^{~}_G\hat\FFF^{~}_G=j^\bcdot\VVV_\bcdot$.
Therefore, the dual-curvature has dynamics in the spinor-gauge bundle.
In four-dimensional spaces, the Hodge-dual of a two-form object is again a two-form object; thus, the curvature two-form is decomposed into the SD and ASD parts.
A solution to the equation of motion in Euclidean space is a pure SD- or ASD curvature; thus, the equation of motion is equivalent to the Bianchi identity. Therefore, the gauge field cannot have dynamics in the gauge bundle.

When a $\Z_2$-grading discrete operator, such as the Hodge-dual operator, exists for the principal curvature, duplex superspace appears in the system.
The $\Z_2$-grading operator splits a space of curvatures into the $\Z_2$-grading superspaces $V^\pm(\Omega^2)$, namely the primary superspace, owing to eigenvalues concerning the operator.
Owing to the isomorphism of four-dimensional spaces (\ref{ismor}), a space of one-form objects also splits into two subspaces; SD- and ASD curvatures are obtained using one-form objects in one of these subspaces.
A pair of the principal curvature and its dual curvature provides the Kramers pair.
The secondary $\Z_2$-grading superspace arises using the one-dimensional Clifford algebra in the $(\TxT)$-matrix representation of the Kramers pair.
The Dirac operator is provided in the secondary superspace, and the Atiyah--Singer index theorem provides the additional topological invariants in the principal bundles.
General relativity and Yang--Mills Lagrangians are represented using supertrace in the secondary superspace.
The similarity between general relativity and Yang--Mills gauge fields, in other words, similarity among four forces in Nature, is evident in our representation.

A procedure to extract topological invariants in the secondary superspace of the principal bundles commonly also appears in material physics. 
The topological insulator's edge-bulk correspondence, such as the charge- and spin-pump phenomena, is an example appearing as a topological invariant in the secondary superspace\cite{PhysRevB.27.6083,PhysRevB.74.195312,PhysRevB.76.045302}.
In this case, the time-reversal operator plays the role of the Hodge-dual operator in this study. The electron wave function and its time-reversal provide the Kramers pair of the fermionic section.
The principal bundle consists of the Brillouin zone as the base manifold, and the vector space of the electron wave functions as the total space; it has the $U(2)$ symmetry as the structure group and the Berry connection as the principal connection.
The spin polarization is the topological invariant expressed by the first Chern number.

General relativity and the Yang--Mills theory are fundamental theories of cosmology and elementary particle physics and treat fundamental phenomena in our universe; thus, they are defined in the four-dimensional manifold with the Minkowski metric and the Bianchi identity and an equation of motion are the hyperbolic type.
On the other hand, the Atiyah--Singer index theorem holds for the elliptic-type Dirac operator on a compact manifold.
Therefore, we cannot apply the theorem on general relativity and the Yang--Mills theory as its original shape, even after appropriately compacting the space-time.
We introduced the $\theta$-metric, which smoothly connects the Minkowski space-time to Euclidean space-time.
When a compact manifold with the Minkowski metric is isotopic to that with the Eudlidean metric, two manifolds are cobordant, and characteristic numbers are cobordism invariant.
In this case, a characteristic index in Minkowski space-time is defined owing to the corresponding index in Euclidean space-time, and the index theorem ensures the index is topological invariant.
We provided several examples of topological invariants defined in the $\theta$-metric space and demonstrated topological invariants working in the Minkowski space-time and Euclidean space-time.
Hayashi utilised a similar method to prove the bulk-edge correspondence for two-dimensional type \textbf{A} and type \textbf{AII} topological insulators owing to the cobordism invariance of the index\cite{Hayashi_2017}.
The method using the $\theta$-metric may have various applications on the wide variety of physical topics.

\begin{table}[t]
\centering
\begin{tabular}{llc}
Field&Lagrangian&\textup{\#}equation\\
\hline
%%%
~&~&~\\
Gravity & $\frac{1}{\hbar\kappa}\frac{1}{2}
Str\left[\bm\FFF_{Gr}\wedge\bm\FFF_{Gr}\right]$&(\ref{EHLformula})\\
%%%
~&~&~\\
Gauge field & $-\frac{1}{2}Str[\Fym\wedge\Fym]$&(\ref{LagrangianSU})\\
%%%
~&~&~\\
Cosmological constant &
$-\frac{1}{\hbar\kappa}\Lambda_\mathrm{c}\hspace{.2em}\vvv$&
(\ref{EHLformula})\\
%%%
~&~&~\\
Matter field & $
\overline{\bm\Psi}_{\ym}\left(i\hspace{.1em}\D^{~}_\ym-
\frac{1}{\hbar}\bm{\mu}^{~}_{\ym}\right)\bm\Psi_{\ym}\hspace{.2em}\vvv$
&(\ref{LagrangianM})\\
~&~&~\\
%%%
\hline
\end{tabular}
\vskip 2mm
\caption{
A summary table of the Lagrangians appearing in the Yang--Mills theory.
}\label{Lags}
\end{table}

\section*{Acknowledgements}
I would like to thank Dr.~Y.~Sugiyama, Prof. J. Fujimoto, Prof. T. Kaneko, and Prof. F. Yuasa for their continuous encouragement and fruitful discussions.
%%% ----------------------------------------------------------------------
% APPENDIX
%%% ----------------------------------------------------------------------
\appendix
%%% ----------------------------------------------------------------------
% existence of dual connection and dual spinor
%%% ----------------------------------------------------------------------
\section{existence of dual connection and dual spinor}\label{AppDC}
%
% $U(1)$ gauge group
%
\subsection{$U(1)$ gauge group}
First, this appendix treats the most straightforward case that the $U(1)$ group in the flat space-time ($\www=0$).
In this case, the dual connection is obtained as a solution of structure equation $\hat\FFF^{~}_{\SU}=d\hat\AAA^{~}_{\SU}$, that satisfies the Bianchi identity $d\hat\FFF^{~}_{\SU}=d(d\hat\AAA^{~}_{\SU})=0$.
The existence problem of a solution for this system is equivalent to the differential-geometrical question: ``Is closed-form $\hat\FFF^{~}_{\SU}$ exact ?''.
For a contractible manifold, the answer to this question is positively given by the Poincar\'{e} lemma; the general answer is provided using the \dR theorem (See, e.g., \cite{warner2013foundations}).
The \dR theorem asserts that $k$'{th} \dR cohomology of a smooth manifold $M$, $H_{\dr}^k(M)$, is isomorphic to the real cohomology $H^k(M,\R)$.

The existence problem of the one-form object $\aaa$ such that $d\aaa=\bbb$ for given two-form object $\bbb$ is answered by a following statement\cite{warner2013foundations}:
% Lemma
\begin{remark}\label{deRh1}
Suppose $M$ is $n$-dimensional smooth manifold and  closed $k$-form object $\bbb$ exists in a compact submanifold of $M$; $\Sigma_\bbb\subseteq M$.
%We note that $M$ is not assumed to be a compact manifold.
In this case, following two statements are equivalent to each other:
\begin{enumerate}
\item  $(k-1)$-form object $\aaa$ such that $\bbb=d\aaa$ exists in $\Sigma_\aaa\subseteq\Sigma_\bbb$, where $\Sigma_\aaa$ is a compact and connected manifold. 
\item The $(n-k)$'{th} relative homology of pair $(\Sigma_\aaa,\partial\Sigma_\aaa)$ is vanished;\\ $H_{n-k}(\Sigma_\aaa,\partial\Sigma_\aaa;\R)=0$, where $\partial\Sigma_\aaa\subset\Sigma_\aaa$ is a boundary of compact and connected manifold $\Sigma_\aaa$.
\end{enumerate}
\end{remark}
\begin{proof}
\underline{\emph{1.}~$\Rightarrow$~\emph{2.}}~:
Closedness of two-form $\bbb$ is consistent with an exactness $\bbb=d\aaa$ such that $d\bbb=dd\aaa=0$.
The $k$'th \dR cohomology of domain $\Sigma_\aaa$ concerning the external-derivative operation  is defined as
\begin{align*}
H^k_{\dr}(\Sigma_\aaa)
&:=Z^k(\Sigma_\aaa)/B^k(\Sigma_\aaa)=0,
\end{align*}
where
\begin{align*}
Z^k(\Sigma)&:=\mathrm{Ker}\left(d:\Omega^k(\Sigma_\aaa)\rightarrow\Omega^{k+1}(\Sigma_\aaa)\right),\\
B^k(\Sigma_\aaa)&:=\mathrm{Im}\left(d:\Omega^{k-1}(\Sigma_\aaa)\rightarrow\Omega^k(\Sigma_\aaa)\right).
\end{align*}
Owing to the \dR theorem and the Poincar\'{e} duality-theorem for a compact and connected four-dimensional manifold, there exists isomorphism:
\[
H^k_{\dr}(\Sigma_\aaa)\rightarrow H^k(\Sigma_\aaa;\R)\rightarrow H_{n-k}(\Sigma_\aaa,\partial\Sigma_\aaa;\R).
\] 
Thus, vanishing relative homology  $H_{n-k}(\Sigma_\aaa,\partial\Sigma_\aaa;\R) = 0$ is maintained.
\\
\noindent
\underline{\emph{2.}~$\Rightarrow$~\emph{1.}}~:
When $H_{n-k}(\Sigma_\aaa,\partial\Sigma_\aaa;\R)=0$, the $k$'th \dR cohomology is also vanished owing to the \dR theorem and the Poincar\'{e} duality theorem.
The vanishing $k$'th cohomology induces isomorphism between $Z^k(\Sigma_\aaa)$ and $B^k(\Sigma_\aaa)$.
If $B^k(\Sigma_\aaa)$ is empty, $Z^k(\Sigma_\aaa)$ is also empty, which contradicts to the assumption of the existence of $\bbb\in Z^k(\Sigma_\aaa)$.
Therefore, the existence of  $\aaa\in\Omega^{k-1}(\Sigma_\aaa)$ such that $B^k(\Sigma_\aaa)\cong[d\aaa]\equiv[\bbb]{\cong}Z^k(\Sigma_\aaa)$ is maintained.
\end{proof}
\noindent
%%%
For an abelian gauge symmetry such as a $U(1)$ gauge group, the structure equation is equivalent to the closedness of the curvature; thus, the existence of the dual connection is stated as the following theorem:
% Theorem
\begin{theorem}\label{deRh2}
Suppose Lie algebra valued two-form object $\FFF$ with an abelian gauge group in the flat space-time is given in submanifold $\Sigma_\FFF\subseteq\M$.
Following two statements are equivalent to each other:
\begin{enumerate}
\item
Lie-algebra valued one-form object $\AAA$ exists in compact manifold $\Sigma_\AAA\subseteq\Sigma_\FFF$ as a solution of structure equation $\FFF=d\AAA$, and the  principal bundle is constructed using connection $\AAA$ and curvature $\FFF$.
\item
The second homology  is vanished in $\Sigma_\AAA$; $H_2(\Sigma_{\AAA},\partial\Sigma_{\AAA};\R)=0$.
\end{enumerate}
\end{theorem}
\begin{proof}
Vanishing second homology is equivalent to existence of connection $\AAA$ owing to \textbf{Remark \ref{deRh1}}.
When connection $\AAA$ is provided as a solution of structure equation $\FFF=d\AAA$ from the  given $\FFF$, it is  Lie algebra valued one-form belonging to an adjoint representation as shown in \textbf{Remark \ref{remA1}}.
\end{proof}
%%
% Exmple
\begin{example}\textbf{Coulomb-type potential}\end{example}\noindent
A domain of a connection for thr Coulomb-type potential is $\Sigma_{\Aa}:=\TsM\SM\{r=0\}$, as given in \textbf{Example \ref{EXcoulomb}}, and corresponding curvature $(\ref{FU1})$ has the same domain.
For dual curvature $\hat\FFF^{~}_\Uo$ given by  (\ref{hFU1}), tensor coefficient $\hat\f^{~}_\Uo$ is also defined in $\TM$ except $r=0$ and is static (no time dependence); thus, the domain of $\hat\f^{~}_\Uo$ is $\Sigma_{\hat\f}=\TM\SM\{r=0\}\cong\R\otimes(\R^3\SM\{r=0\})$.
The second homology of a compact submanifold of the domain; $\overline{\Sigma}_{\hat\f}:=D^1{\otimes}(D^3{\SM}\{0\})$ is calculated as $H_2(\overline{\Sigma}_{\hat\f};\R)\cong H_2(D^1,S^0;\R){\oplus}H_2(D^3,S^2{\oplus}S^2;\R)=0$; thus, there exists dual connection $\hat\AAA^{~}_\Uo$ in $\Sigma_{\hat\Aa}\subseteq\Sigma_{\hat\f}$.

% Example
\begin{example}\label{EXLcurrect}\textbf{Linear electric current}\end{example}\noindent
A linear electric current makes a stationary magnetic field around the current in a flat space-time manifold, and it has a $U(1)$ connection such that:
\begin{align*}
\AAA^{~}_\Uo&=
\frac{1}{2}\log{\left(x^2+y^2\right)}\hspace{.1em}dz\in\Omega^1\left(\TsM\SM\{x=0\wedge y=0\}\right).
\end{align*}
Corresponding curvature and dual curvature are respectively provided as
\begin{align*}
\FFF^{~}_\Uo&=\frac{x}{x^2+y^2}\hspace{.1em}dx\wedge dz+\frac{y}{x^2+y^2}\hspace{.1em}dy\wedge dz\\
\overset{\HD}{\longrightarrow}\hat\FFF^{~}_\Uo&=
 i\frac{y}{x^2+y^2}\hspace{.1em}dt{\wedge}dx
-i\frac{x}{x^2+y^2}\hspace{.1em}dt{\wedge}dy,
\end{align*}  
The dual connection is provided as a solution of structure equation (\ref{eq1}) such that
\begin{align*}
\hat\AAA^{~}_\Uo&=
i\tan^{-1}{\frac{y}{x}}\hspace{.2em}dt,
\end{align*}
which is indeterminate at $x\hspace{-.1em}=\hspace{-.1em}y\hspace{-.1em}=\hspace{-.1em}0$ with any $z$; thus, its domain is $\Sigma_{\hat\Aa}\cong\R\otimes\left(\R^3\SM\R\right)$.  
The second homology of $\Sigma_{\hat\Aa}$ is the same as that of \textbf{Example \ref{EXcoulomb}} such that $H_2(\overline{\Sigma}_{\hat\Aa};\R){\cong}H_2(D^1;\R){\oplus} H_2(D^1{\otimes}S^1,T^2;\R)=0$; thus, the dual connection exists in the same domain.
A domain of the dual connection is the same as that of the dual curvature. 

The second Chern class for $\AAA^{~}_\Uo$ and $\hat\AAA^{~}_\Uo$ is provided as 
\[
c^{~}_2(\AAA^{~}_\Uo)=c^{~}_2(\hat\AAA^{~}_\Uo)=
\left((4\pi(x^2+y^2)\right)^{-2}.
\]\QED
%%%
% Example
\begin{example}\textbf{Plane wave}\end{example}\noindent
A plane wave solution in globally flat $\MM$ is considered.
A $U(1)$ connection and corresponding curvature are, respectively, provided using the Cartesian coordinate system as 
\begin{align*}
\AAA^{~}_\Uo&=\sin{(t-z)}\left(dx+dy\right),\\
\FFF^{~}_\Uo&=\cos{(t-z)}\left(dt\wedge dx+dt\wedge dy+dx\wedge dz+dy\wedge dz\right),
\end{align*}
which represent plane wave lying in a $(x,y)$-plane propagating to $z$ direction with speed $c=1$, and its dual curvature is obtained as follows:
\begin{align*}
\FFF^{~}_\Uo=\cos{(t-z)}\left(dt\wedge dx-dt\wedge dy+dx\wedge dz-dy\wedge dz\right).
\end{align*}
The domain of $\FFF^{~}_\Uo$ is the entire space of $\TsM{\cong}\R^4$, which provides the  vanishing second homology; $H_2(\R^4,\slash{0};\R)=0$. 
The dual connection is provided as follows:
\begin{align*}
\hat\AAA^{~}_\Uo=i\sin{(t-z)}\left(dx-dy\right),
\end{align*}
whose domain is also entire  $\TsM$.
%
%In this case, the second Chern class for $\AAA^{~}_\Uo$ and $\hat\AAA^{~}_{\SU}$ is provided as $c^{~}_2(\AAA^{~}_\Uo)=c^{~}_2(\hat\AAA^{~}_\Uo)=0$; thus, two connections belong the same bundle, namely the self-dual bundle.
\QED
%%%

%
% $SU(N)$ gauge group
%
\subsection{$SU(N)$ gauge group}
Next, we treat a non-abelian gauge with $SU(N)$ group in the flat space-time.
In this case, exactness and closedness are corresponding to the structure equation and Bianchi identity, respectively, such that:
\begin{align}
\emph{Covariant-exactness}\hspace{.3em}:&\hspace{1.1em}\hat\FFF^{~}_{\SU}\hspace{.4em}=
d\hat\AAA^{~}_{\SU}
-i\hspace{.1em}{{c}^{~}_\SU}\hspace{.1em}\hat\AAA^{~}_{\SU}\wedge\hat\AAA^{~}_{\SU}
\tag*{(\ref{eq1})\texttt{'}},\label{eq12}\\
\emph{Covariant-closedness}:&\hspace{.1em}\hat{d}_{\SU}\hat\FFF^{~}_{\SU} = 
d\hat\FFF^{~}_{\SU}-i\hspace{.1em}{{c}^{~}_\SU}\hspace{.1em}[\hat\AAA^{~}_{\SU},\hat\FFF^{~}_{\SU}]_\wedge=0
\tag*{(\ref{BianchihSU})\texttt{'}}\label{BianchihSU2},
\end{align}
where $SO(1,3)$-covariant differential $d_\www$ in (\ref{eq1}) and in (\ref{BianchihSU}) is replaced by the external derivative owing to the flat space-time.
As is the case of an abelian gauge, when the dual curvature is \emph{covariant-exact} (given by the structure equation \ref{eq12}, it is trivially \emph{covariant-closed} owing to the Bianchi identity \ref{BianchihSU2}.
The question here is that ``Does dual curvature $\hat\FFF^{~}_{\SU}$ satisfying \ref{BianchihSU2} have a solution for \ref{eq12} ?''.
A key to treat this question is the Chern--Simons form, which is introduced in the spinor-gauge bundle as follows:
% Definition
\begin{definition}[Chern--Simons form]
The Chern--Simons forms in a gauge bundle is defined as a functional of $\hspace{.1em}\AAA_\SU^\#$ in the total space such that: 
\begin{align}
\cs_\SU(\AAA^\#_\SU):=\cs^I_\SU(\AAA^\#_\SU)\hspace{.2em}t_I
:=-i\left[\AAA^\#_\SU,d\AAA_\SU^\#
-i\hspace{.1em}{{c^{~}_\SU}}\hspace{.1em}\AAA_\SU^\#\wedge\AAA_\SU^\#\right]_\wedge.
\label{CSform}
\end{align}
A superscript ``{\footnotesize\textrm{\#}}'' represents a pull-buck of forms from the base manifold to the total manifold  with respect to the bundle map. 
\end{definition}
\noindent
%%%
Hereafter, ``{\footnotesize\textrm{\#}}'' and a subscript ``$\SU$'' are omitted for simplicity unless otherwise denoted.
The Chern--Simons form of a dual connection is defined by $(\ref{CSform})$ with replacing $\AAA^{~}_\SU$ to $\hat\AAA^{~}_{\SU}$.  
Discussions below are common for both principal- and dual connections.

We introduce a map 
\begin{align*}
\pi_\aaa:\Omega^p\rightarrow\Omega^{p-1}:\left[\aaa,\bbb\right]_\wedge\mapsto i\bbb,
\end{align*}
for $\aaa\in\Omega^1$ and $\bbb\in\Omega^{p-1}$.
An anti-derivative operator $\delta_\aaa:\Omega^p\rightarrow\Omega^{p-1}$ is defined in the total space of the principal bundle as follows:
\begin{align}
\delta_\aaa\WWW(\AAA)&=\pi^{~}_\aaa\left(\lim_{\varepsilon \to 0}\frac{\WWW\left(\AAA+\varepsilon\aaa\right)
-\WWW\left(\AAA\right)}{\varepsilon}\right),\label{deltaA}
\end{align}
where $\WWW(\AAA)$ is an invariant polynomial in Weil algebra.
\dR cohomology of Weil algebra with respect to anti-derivative operator $\delta_\aaa$ is trivial such that:
\begin{align}
H_{\dr}^k(W(\AAA);\delta_\aaa)=
\begin{cases}
    \R &(k=0) \\
    0&(k>0),
  \end{cases}\label{cohoweil}
\end{align}
where $W(\AAA)\in\WWW(\AAA)$ is Weil algebra appearing in the gauge bundle.

% Remark
\begin{remark}\label{CSRM1}
The principal curvature form in the total space is an exact form by means of the Chern--Weil theory, and there exists the Chern--Simons form which gives the curvature form such that $\FFF=\delta_\aaa\cs$.
\end{remark}
\begin{proof}
Suppose the curvature is provided as a solution of the structure equation.
The Chern--Simons form is expressed using $\AAA$ as
\begin{align*}
\cs^I\hspace{-.3em}\left(\AAA\right)&=
f^I_{~JK}\hspace{.2em}\AAA^J\wedge\left(
d\AAA^K+\frac{{c}}{2}f^K_{~LM}\AAA^L\wedge\AAA^M\right),\\
&=\frac{1}{2}\left\{f^I_{~JK}\hspace{.2em}d\hspace{-.1em}\left(\AAA^J\wedge\AAA^K\right)
+{c}\hspace{.1em}f^I_{~JK}\hspace{.1em}f^K_{~LM}\hspace{.2em}\AAA^J\wedge\AAA^L\wedge\AAA^M\right\}.
\end{align*}
Anti-derivative operator is applied on the Chern--Simons form as follows:
\begin{align*}
\delta_\aaa\cs
=\pi^{~}_\aaa\left(f^I_{~JK}\hspace{.2em}\aaa^J\hspace{-.2em}\wedge\left(d\AAA^K
+\frac{{c}}{2}f^K_{~LM}\hspace{.2em}\AAA^L\wedge\AAA^M\right)t_I\right)
=-i\pi^{~}_\aaa\left(\left[\aaa,\FFF\right]_\wedge\right)=\FFF,
\end{align*}
where the Jacobi identity $f_{\hspace{-.1em}A\hspace{-.1em}B\hspace{-.1em}E}\hspace{.3em}f_{\hspace{-.1em}E\hspace{-.1em}C\hspace{-.1em}D}+f_{\hspace{-.1em}C\hspace{-.1em}B\hspace{-.1em}E}\hspace{.3em}f_{\hspace{-.1em}A\hspace{-.1em}E\hspace{-.1em}D}+f_{\hspace{-.1em}D\hspace{-.1em}B\hspace{-.1em}E}\hspace{.3em}f_{\hspace{-.1em}A\hspace{-.1em}C\hspace{-.1em}E}=0$ is used.
Vector space $\Omega^{2}$ is interpreted as a dual space of $\Omega^1$ with respect to a bilinear map $\Omega^1\otimes\Omega^2\rightarrow \R[\sss\uuu_N]$ such that
\begin{align*}
\langle\aaa,\FFF\rangle_\wedge&:=\int_\Sigma %\Tr_{\TsM}
[\aaa,\FFF]_\wedge\in\R[\sss\uuu_N],
\end{align*}
where $\Sigma\subset\M$ is a three-dimensional oriented and compact submanifold of $\M$, and $\R[\sss\uuu_N]$ is polynomial algebra of Lie algebra $\sss\uuu_N$ with real coefficients.
This bilinear map is non-degenerated and induces homomorphism $\Omega^{2}\cong{\Omega^1}^\#$; thus, $\FFF$ is recognized as a one-form object in the space of connections.
Therefore, the existence of the Chern--Simons form is maintained owing to its \dR cohomology $H_{\dr}^1(\Sigma_\cs;\delta_\AAA)=0$ and the \dR theorem, where $\Sigma_\cs$ is a domain of the Chern--Simons form.
\end{proof}
\noindent
%%%

% Remark
\begin{remark}\label{CSclose}
The Chern--Simons form is a closed form with respect to the external derivative such that $d\cs=0$, and it is consistent with the structure equation \textup{\ref{eq12}}. 
\end{remark}
\begin{proof}
The Bianchi identity ref{BianchihSU2} is represented using the Chern--Simons form as $d\FFF+\cs=0$, when $\FFF$ is defined using the structure equation \ref{eq12}.
Therefore, the Chern--Simons form is closed such that  $dd\FFF+d\cs=d\cs=0$ with respect to the external derivative in $\TsM$.
The external derivative of the Chern--Simons form is directly calculated as
\begin{align*}
d\cs^I&=f^I_{~JK}\left(d\AAA^J\wedge\FFF^K-\AAA^J\wedge d\FFF^K\right),\\&=
f^I_{~JK}\left(d\AAA^J\wedge\FFF^K+{c} f^K_{~LM}\AAA^J\wedge
\AAA^L\wedge\FFF^M
\right),\\ 
&=f^I_{~JK}\left(
d\AAA^J+\frac{{c}}{2}f^J_{~ML}\AAA^M\wedge\AAA^L
\right)\wedge\FFF^K=
f^I_{~JK}\hspace{.2em}\widetilde\FFF^J\wedge\FFF^K,
\end{align*}
where $\widetilde{\FFF}^J:=d\AAA^J+{c} f^J_{~ML}\AAA^M\wedge\AAA^L/2$.
In this calculaiotn, the Bianchi identity \ref{BianchihSU2} and the Jacobi identity are used.
We note that the structure equation \ref{eq12} is not assumed here. 
When $\widetilde\FFF^J\wedge\FFF^K$ is symmetric in exchange $J$ for $K$, relation $d\cs^I=0$ is maintained.
Therefore, closedness of the Chern--Simon form is consistent with $\widetilde\FFF^I=\FFF^I=d\AAA^I-i\hspace{.1em}{{c}}\hspace{.3em}\left[\AAA\wedge\AAA\right]^I$.
\end{proof}
%%%

Whereas $\ref{eq12}\Rightarrow (d\cs=0)$ is trivial, inverse statement $(d\cs=0)\Rightarrow \ref{eq12}$ for given $\FFF$ is not always true.
When non-singular two-form objects $\FFF^I$ are given, the existence of one-form objects $\AAA^I$ fulfilling both \ref{BianchihSU2} and \ref{eq12} is discussed here.
The Bianchi identity \ref{BianchihSU2} is represented using trivial bases in $\TsM$ such that
\begin{align*}
t_I\left(
\partial_\bcdot\f^I_\bcdots+{c}\hspace{.1em}f^I_{\hspace{.1em}JK}\Aa^J_\bcdot\f^K_\bcdots
\right)\eee^\bcdot\wedge\eee^\bcdot\wedge\eee^\bcdot&=0,
\end{align*}
which is equivalent to
\begin{align}
\epsilon^{a\bcdot\bcdots}\left(
\partial_\bcdot\f^I_\bcdots+{c}\hspace{.1em}f^I_{\hspace{.1em}JK}\Aa^J_\bcdot\f^K_\bcdots
\right)&=0.\label{bicomp}
\end{align}
We note that the flat space-time yields $d\eee=0$.
Equation (\ref{bicomp}) is represented using the dual curvature as
\begin{align*}
\partial_\bcdot[\hat{\f}^I]^{a\bcdot}
+{c}\hspace{.1em}f^I_{\hspace{.1em}JK}[\hat{\f}^K]^{a\bcdot}\Aa^J_\bcdot
=0.%\label{bicomp2}
\end{align*}
Unknown function $\Aa^I_a$ has $N_\Aa:=N_\textrm{ST}\times\hspace{-.2em}N_{SU}$ independent components, where $N_\textrm{ST}$ is a dimension of the space-time manifold and $N_{SU}=N^2-1$ is total number of freedom for the $SU(N)$ symmetry.
Thus, equation (\ref{bicomp}) has $(N_\Aa\hspace{-.2em}\times\hspace{-.2em}N_\Aa)$-matrix representation such that
\begin{align}
\left[\partial\hat{\f}\right]^{(I;a)}
+{c}\left[f\hat{\f}\right]^{(I;a)}_{(J;\bcdot)}\left[\Aa\right]^{(J;\bcdot)}
=0,
\label{bicomp3}
\end{align}
where
\begin{align*}
\left[\partial\hat{\f}\right]^{(I;a)}&:=
\eta^{a\star}\eta^{\bcdots}\partial_\bcdot\hat{\f}^I_{\star\bcdot},\hspace{1em}
\left[f\hat{\f}\right]^{(I;a)}_{(J;b)}:=
\eta^{a\bcdot}f^I_{\hspace{.1em}JK}\hspace{.1em}\hat{\f}^K_{\bcdot b},\\
\left[\Aa\right]^{(J;a)}&:=\eta^{a\bcdot}\Aa^J_\bcdot,
\end{align*}
and $(I;a):=N_\SU\hspace{-.2em}\times\hspace{-.2em}(I\hspace{-.2em}-\hspace{-.2em}1)\hspace{-.2em}+\hspace{-.2em}a$ runs from 1 to $N_\Aa$.
Suppose all $\f^I_\bullets$ (and thus, $\partial_\bullet\f^I_\bullets$ too) are given.
If matrix $[\hat{\f}]$ is invertible, $[f\hat{\f}]$ is also invertible; thus $\AAA$ is uniquely provided from given $\FFF$.
It is proven for the $SU(2)$ Yang--Mills theory:

%Example
\begin{example}\label{su2sol}\textbf{Solution of equation $(\ref{bicomp3})$ for SU(2)}\end{example}\vskip -4mm\noindent
Matrix $\left[f\hat{\f}\right]$ has a block representation such that
\begin{align*}
\left[f\hat{\f}\right]=
\left(
    \begin{array}{ccc}
      \hspace{1em}\bm{0}   & \hspace{1em}\left[\hat{\f}^3\right] & -\left[\hat{\f}^2\right]\\
   -\left[\hat{\f}^3\right] & \hspace{1em}\bm{0}    &  \hspace{1em}\left[\hat{\f}^1\right]\\
      \hspace{1em}\left[\hat{\f}^2\right]&-\left[\hat{\f}^1\right] &    \hspace{1em}\bm{0}
    \end{array}
\right),
\end{align*}
where $\bm{0}$ is a $(4\times4)$ zero-matrix.
Here, matrix
\begin{align*}
\tilde{\bm{F}}:=\left(
\begin{array}{ccc}
\bm{1}&\left[\hat{\f}^2\right]\left[\hat{\f}^1\right]^{-1}&
\left[\hat{\f}^3\right]\left[\hat{\f}^1\right]^{-1}\\
\left[\hat{\f}^1\right]\left[\hat{\f}^2\right]^{-1}&\bm{1}&
\left[\hat{\f}^3\right]\left[\hat{\f}^2\right]^{-1}\\
\left[\hat{\f}^1\right]\left[\hat{\f}^3\right]^{-1}&
\left[\hat{\f}^2\right]\left[\hat{\f}^3\right]^{-1}&\bm{1}
\end{array}
\right),
\end{align*}
are introduced, where $\bm{1}$ is $(4\times4)$ unit-matrix and $[\bullet]^{-1}$ is inverse of matrix $\bullet$.
Notably, all $[\f^\bullet]$ are assumed to be invertible.
Matrix $\left[f\hat{\f}\right]$ is block-diagonalized using $\tilde{\bm{F}}$ as
\begin{align*}
&\tilde{\bm{F}}.\left[f\hat{\f}\right]\\&=\left(
\begin{array}{ccc}
\left[\hat{\f}^3\right]\left[\hat{\f}^1\right]^{-1}\left[\hat{\f}^2\right]
&0&0\\
0&\left[\hat{\f}^1\right]\left[\hat{\f}^2\right]^{-1}\left[\hat{\f}^3\right]
&0\\
0&0&\left[\hat{\f}^2\right]\left[\hat{\f}^3\right]^{-1}\left[\hat{\f}^1\right]
\end{array}
\right)\\
&~~-
\left(
\begin{array}{ccc}
(2\leftrightarrow 3)
&0&0\\
0&(3\leftrightarrow 1)
&0\\
0&0&(1\leftrightarrow 2)
\end{array}
\right)
\end{align*}
Matrix $[\hat{\f}^I]$ are anti-symmetric as $[\hat{\f}^I]^T=-[\hat{\f}^I]$ and $([\hat{\f}^I]^{-1})^T=-[\hat{\f}^I]^{-1}$; thus, diagonal block-matrices are symmetric such that
\begin{align*}
\tilde{\bm{F}}_{J}&:=[\hat{\f}^I][\hat{\f}^J]^{-1}[\hat{\f}^K]-(K \leftrightarrow I),\\
\tilde{\bm{F}}_{J}^T&=[\hat{\f}^K]^T([\hat{\f}^J]^{-1})^T[\hat{\f}^I]^T-(K \leftrightarrow I)
,\\&=
-[\hat{\f}^K][\hat{\f}^J]^{-1}[\hat{\f}^I]+(K \leftrightarrow I)=\tilde{\bm{F}}_{J}.
\end{align*}
Therefore, diagonal blocks of $\tilde{\bm{F}}.[f\hat{\f}]$ are invertible matrices; thus, $[\Aa]$ has a unique solution.
\QED
%%%

Field matrices $\f^I$ are anti-symmetric $(N_\textrm{ST}\hspace{-.2em}\times\hspace{-.2em}N_\textrm{ST})$-matrices; thus, they have $N_\SU\times N_\textrm{ST}\hspace{-.2em}\times\hspace{-.2em}(N_\textrm{ST}-1)/2$ independent components in total.
On the other hand, $\Aa^I$ has $(N_\SU\times N_\textrm{ST})$-independent components.
If curvature $\FFF$ is provided from connection $\AAA$ through the structure equation, all component of $\f^I$ are not independent; there are $N_\SU\times N_\textrm{ST}\hspace{-.2em}\times\hspace{-.2em}(N_\textrm{ST}-3)/2$ constraints.
In this case, algebraic equation $(\ref{bicomp3})$ has a unique solution that fulfils the structure equation.
When arbitrary anti-symmetric matrices $\f^I$ are given without specifying $\Aa^I$, the solution of algebraic equation $(\ref{bicomp3})$ does not fulfil \ref{eq12} in general.
A necessary condition to realize a relation $\ref{BianchihSU2}\Rightarrow\ref{eq12}$ is provided below.

%Theorem
\begin{theorem}\label{deRhSUN}
Suppose $SU(N)$ Lie-algebra valused non-singular two-form object $\FFF$ given in $\Sigma_\f\subseteq\TsM$ is an adjoint representation of $SU(N)$.
Following two statements are equivalent:
\begin{enumerate}
\item
$\FFF$ is an $SU(N)$-curvature associated with $SU(N)$-connection $\AAA$.
\item
Solution of algebraic equation  $d\FFF-i\hspace{.1em}{{c}}\hspace{.1em}[\AAA,\FFF]_\wedge=0$ with respect to $\AAA$ 
%fulfils relation $\FFF=d\AAA-i\hspace{.1em}{{c}}\hspace{.1em}\AAA\wedge\AAA$.
 in compact manifold $\Sigma_\AAA$ has the vanishing first homology  such that;\\ $H_1(\Sigma_\AAA,\partial\Sigma_\AAA;\R)=0$.
\end{enumerate}
\end{theorem}
\begin{proof}
\underline{\emph{1.}~$\Rightarrow$~\emph{2.}}: When $\FFF$ is a curvature, it is represented using structure equation $\FFF=d\AAA-i\hspace{.1em}{{c}}\hspace{.1em}\AAA\wedge\AAA$ and fulfils the Bianchi identity  $d\FFF-i\hspace{.1em}{{c}}\hspace{.1em}[\AAA,\FFF]_\wedge=0$.
~\\
\noindent
\underline{\emph{2.}~$\Rightarrow$~\emph{1.}}: 
When non-singular tensors $\f^I_{ab}$ for $\FFF=\f^I_\bcdots(\eee^\bcdot\wedge\eee^\bcdot){t_I}$ are provided, equation $d\FFF-i\hspace{.1em}{{c}}\hspace{.1em}[\AAA,\FFF]_\wedge=0$ has a unique solution for $\Aa^I_a$.
When $\Sigma_\AAA$ has the vanished first homology; $H_1(\Sigma_\AAA,\partial\Sigma_\AAA;\R)=0$, the statement \emph{2}.\@ equivalently asserts that the exact form $\widetilde\cs:=-i\hspace{.1em}{{c}}\hspace{.1em}[\AAA,\FFF]_\wedge$ exists in oriented closed $\Sigma_{\widetilde\cs}$ corresponding to given $\FFF$ such that $\widetilde\cs=-d\FFF$. 
Therefore, the first-homology of $\widetilde\cs$ is vanished and $\widetilde\cs$ trivially induces exact form $d\FFF$ owing to \textbf{Remark~\ref{deRh1}}. 
On the other hand, statement \emph{1}.\@ also uniquely induces the Chern--Simons form, which fulfils the Bianchi identity $d\FFF+\cs=0$; thus, $\widetilde\cs=\cs$ is fulfilled.
Therefore, the theorem is maintained.
\end{proof}

%
% $SU(N)$ gauge group in curved space-time
%
\subsection{$SU(N)$ gauge group in curved space-time}
Finally, non-abelian gauge of $SU(N)$ group in the curved space-time is considered.
Exactness and closedness are represented, respectively, such that:
\begin{subequations}
\begin{align}
\emph{Covariant\hspace{.5em}exactness}\hspace{.3em}:&\hspace{1.3em}\FFF=
d\AAA+\cG\www\wedge\AAA
-i\hspace{.1em}{{c}}\hspace{.2em}\AAA\wedge\AAA,\vspace{1mm}\\
\emph{Covariant\hspace{.5em}closedness}:&\hspace{.1em}d_G\FFF=
d\FFF+\cG[\www,\FFF]_\wedge-i\hspace{.1em}{{c}}\hspace{.1em}[\AAA,\FFF]_\wedge=0.
\end{align}
\end{subequations}
A component representation of, e.g., $\www\wedge\AAA$, using the trivial bases in $\TsM$ is provided as
\begin{align*}
\www\wedge\AAA&=[\www\wedge\AAA ]_{ab}\hspace{.1em}\eee^a\wedge\eee^b=
\left(\eta_{b\bcdot}\hspace{.1em}\omega_\mu^{\hspace{.3em}\bcdot\star}\Aa^I_\star\Varepsilon^\mu_a\right)
\hspace{.1em}\eee^a\hspace{-.1em}\wedge\eee^b\hspace{.2em}t_I.
\end{align*}
Here, an approximation that spin form $\www$ is a external field that has constant values independent from $\AAA$ and $\FFF$, is exploited.
In this case, the Chern--Simons form is defied as
\begin{align*}
\cs\hspace{-.2em}\left(\www,\AAA\right):=
[\cG\www-i\hspace{.1em}{{c}}\hspace{.1em}\AAA,\FFF]_\wedge,
\end{align*}
and thus, its closedness with respect to the external derivative is maintained; $d\cs\hspace{-.2em}\left(\www,\AAA\right)=0$.
Anti-derivative operator $\delta_{\www\otimes\AAA}$ is defined similarly as (\ref{deltaA}), and the exactness with respect to the anti-derivative operator $\delta_{\www\otimes\AAA}\hspace{.1em}\cs\hspace{-.2em}\left(\www,\AAA\right)=\FFF$ is maintained.
\textbf{Theorem \ref{deRhSUN}} is also maintained in this case because spin form $\www$ is assumed to be constant valued.

%Existence of dual spinor
\subsection{Existence of dual spinor}
It is known that existence of a spin structure in a  bundle is distinguished owing to the first and second Stiefel--Whitney classes\cite{milnor1974characteristic}:
\begin{remark}
Suppose $w^i(M)\in H^i(M,\Z_2)$ is $i^{th}$ Stiefel--Whitney class of orientable manifold $M$.
%\vspace{-5mm}
\begin{enumerate}
\item Manifold $M$ has a spin structure, if and only if 
\begin{align*}
w^2(M)\in H^2(M,\Z_2)\cong Hom(H_2(M;\Z_2),\Z_2)
\end{align*}
vanishes; $w^2(M)=0$.
\item When $M$ has a spin structure, each spin structure has one-to-one correspondence to
\begin{align*}
H^1(M,\Z_2)\cong Hom(\pi_1(M),\Z_2),
\end{align*}
where $\pi_1(M)$ is fundamental group of $M$. 
\end{enumerate}
\end{remark}
\begin{proof}
For a proof of the remark, see, e.g., chapter 4, section 5 in Ref.\cite{cohen1998}
\end{proof}
For a three-dimensional manifold $\M^3_{\overline{r=0}}:=\R^3\SM\{r=0\}$, any closed surface $\Sigma^2_{\overline{r=0}}\subset\M^3_{\overline{r=0}}$ such that $[\Sigma^2_{\overline{r=0}}]\in H_2(\M^3_{\overline{r=0}},\Z_2)$ is homologous to $S^2$, and any vector bundle in $S^2$ is isomorphic to $\pi_1(SO(3))$.
Therefore, both first and second Stiefel--Whitney classes vanish, and only one spin structure exists in $\M^3_{\overline{r=0}}$.
As shown in \textbf{Example~\ref{EXcoulomb}}, the Coulomb type potential is defined in $\R\otimes\M^3_{\overline{r=0}}$; thus, there exists a spinorial section corresponding to an electric monopole in the Yang--Mills theory with $U(1)$-gauge group (an electron in electromagnetic theory).
On the other hand, a manifold $\M^3_{\overline{x=y=0}}:=\R^3\SM\{x=y=0\}$ has a non-trivial two-dimensional surface such as a torus winding around the Dirac string.
Therefore, the second Stiefel--Whitney class does not vanish; $w^2(\M^3_{\overline{x=y=0}})\cong\Z$; thus, the spin structure does not exist in a dual-bundle.
Therefore, in a space with an isolated magnetic monopole, an electron does not exist.
If a pair of magnetic monopoles with opposite charges makes a magnetic dipole, the second Stiefel--Whitney class is trivial; thus, a spin structure exists.

%%%%%%%%%%% end appendix %%%%%%%%%%%
% references
\bibliographystyle{bmc-mathphys} % Style BST file (bmc-mathphys, vancouver, spbasic).
\bibliography{Topology-GYM}      % Bibliography file (usually '*.bib' )
% ------------------------------------------------------------------------
\end{document}